\documentclass[11pt]{amsart}

\usepackage{bbold}
\usepackage{amscd} 
\usepackage{amsmath,amssymb}
\usepackage{mathtools}
\usepackage{mathrsfs}
\newtheorem{theorem}{Theorem}[section]
\newtheorem{lemma}[theorem]{Lemma}
\newtheorem{corollary}[theorem]{Corollary} 
\theoremstyle{definition}
\newtheorem{definition}[theorem]{Definition}

\newtheorem{remark}[theorem]{Remark} 
\numberwithin{equation}{section}

\def\be{\begin{equation}}
\def\ee{\end{equation}}
\def\RE{\mathbb{R}}
\def\CO{\mathbb{C}}
\def\v{\mathsf v}

\def\H{\mathsf H}
\def\Sc{\mathsf S}
\def\L{\mathcal L}

\def\inn{\rm in}
\def\-{\rm in}
\def\+{\rm ex}
\def\ex{\rm ex}
\def\B{\mathscr B}
\def\Bi{\mathsf B}
\def\SL{S\!L}
\def\DL{D\!L}
\def\ran{\text{\rm ran}}
\def\dom{\text{\rm dom}}
\def\ker{\text{\rm ker}}
\def\h{\mathfrak h}
\def\b{\mathfrak b}
\def\X{X}
\def\Y{Y}
\begin{document}

\title{Scattering theory with both regular and singular perturbations}
\author{Andrea Mantile}
\author{Andrea Posilicano}
\address{Laboratoire de Math\'{e}matiques de Reims, UMR9008 CNRS, Universit\'{e} de Reims Champagne-Ardenne, Moulin de la Housse BP 1039, 51687 Reims, France}
\address{DiSAT, Sezione di Matematica, Universit\`a dell'Insubria, via Valleggio 11, I-22100
Como, Italy}
\email{andrea.mantile@univ-reims.fr}
\email{posilicano@uninsubria.it}
\begin{abstract} 
We provide an asymptotic completeness criterion and a representation formula for the scattering matrix of the scattering couple $(A_{\Bi},A)$, where both $A$ and $A_{\Bi}$ are self-adjoint operator and $A_{\Bi}$ formally corresponds to adding to $A$ two terms, one regular and the other singular. In particular, our abstract results apply to the couple $(\Delta_{\Bi},\Delta)$, where $\Delta$ is the free self-adjoint Laplacian in $L^{2}(\mathbb{R}^{3})$ and $\Delta_{\Bi}$ is a self-adjoint operator in a class of Laplacians with both a regular perturbation, given by a short-range potential, and a singular one describing boundary conditions (like Dirichlet, Neumann and semi-transparent $\delta$ and $\delta'$ ones) at the boundary of a open, bounded Lipschitz domain. The results hinge upon a limiting absorption principle for $A_{\Bi}$ and a Kre\u\i n-like formula for the resolvent difference $(-A_{\Bi}+z)^{-1}-(-A+z)^{-1}$ which puts on an equal footing the regular (here, in the case of the Laplacian, a Kato-Rellich potential suffices) and the singular perturbations. 
\end{abstract}
\maketitle

\section{Introduction}
The mathematical scattering theory for short-range potential  is a well developed subject; the existence and completeness of the wave operators can be obtained by two essentially different approaches: the trace-class method and the smooth method (see, e.g., \cite{Y-LNM}). An important object defined in terms of the wave operators is the scattering operator and, even more important from the point of view of its physical applications, the scattering matrix, which is its reduction to a multiplication operator in the spectral representation of the self-adjoint free Laplacian. \par
The scattering problem for singular perturbations of self-adjoint operators, which is outside the original scope of these methods, 
is connected with scattering from obstacles with impenetrable or semi-transparent boundary conditions (see, e.g., \cite{BLL}, \cite{BMN}, \cite{JMPA}-\cite{JDE18}). On this side,  a general   scheme  has been developed in \cite{JMPA} by combining the  construction in \cite{JFA} with an abstract version of the Limiting Absorption Principle (simply LAP in the following) due to W. Renger (see \cite{Reng}) and a variant of the smooth method due to M. Schechter (see \cite{Sch}). In particular, the results in \cite{JMPA} apply to obstacle scattering with a large class of interface conditions on Lipschitz hypersurfaces in any dimension. Let us recall that in \cite{BMN} boundary triple theory and properties of the associated operator-valued Weyl  functions  were used to obtain a similar representation of the scattering matrix for singularly coupled self-adjoint extensions. It is worth to remark that, while the approach in \cite{JMPA} avoids any trace-class condition, these are needed in \cite{BMN} and so  the applications there are limited to the case of smooth obstacles in two dimensions.
\par
The target of the present paper is to provide a general framework for the scattering with both potential type and singular perturbations. Since our concern is the scattering theory with respect to the free Laplacian, we regard the regular and the singular parts of the perturbation as a single object; this constitutes the main novelty of our approach. In particular, we give an abstract resolvent formula, generalizing the one in \cite{JFA}, which puts on an equal footing the two components of the perturbation. Such a representation is a key ingredient in the derivation of LAP which leads then to the main results of the first part: the asymptotic completeness and an explicit formula for the scattering matrix. These results rely on a certain number of assumptions whose validity is carefully analyzed in the second part where we consider the specific case of a short range potential plus a distributional term, supported on a closed surface and describing self-adjoint interface conditions. In this way, we obtain new representation formulae for the scattering matrix which are expected to be relevant in different physical applications involving wave propagation in inhomogeneous media with impenetrable or semi-transparent obstacles. 
\par
Here, in more details, the contents of the paper. \par In Section 2, following the scheme proposed in \cite{JFA}, we provide an abstract resolvent formula for a perturbations $A_{\Bi}$ of the self-adjoint $A$ by a linear combination of the adjoint of two bounded trace-like maps $\tau_{1}:\dom(A)\to\h_{1}$ and $\tau_{2}:\dom(A)\to\h_{2}$; while the kernel of $\tau_{2}$ is required to be dense, so $\tau_{2}^{*}$ plays the role of a singular perturbations,  no further hypothesis is required for $\tau_{1}$ and in applications that allows $\tau_{1}^{*}$ to represent  a regular perturbations by a short-range potential.  In Subsection 2.3, by block operator matrices and the Schur complement, we re-write the obtained resolvent formula in terms of the resolvent of the operator corresponding to the non singular part of the perturbations; that plays an important  role in the subsequent part regarding LAP and the scattering theory.
\par In Section 3,  following the scheme proposed in \cite{JST} and further generalized in \cite{JMPA}, at first we provide, under suitable hypothesis, a Limiting Absorption Principle for $A_{\Bi}$ (see Theorem \ref{LAP}) and then an aymptotic completeness criterion for the scattering couple $(A_{\Bi},A)$ (see Theorem \ref{AC}). Then, by a combination of  LAP with stationary scattering theory in the Birman-Yafaev scheme and the invariance principle, we obtain a representation formula for the scattering matrix of the couple $(A_{\Bi},A)$ (see Theorem \ref{S-matrix}). Whenever $A$ is the free Laplacian in $L^{2}(\RE^{3})$, such a formula contains, as subcases, both the usual formula for the perturbation given by a short-range potential as given, e.g., in \cite{Y-LNM} and the formula for the case of a singular perturbation describing self-adjoint boundary conditions on a hypersurface as given in \cite{JMPA}.
\par In Section 4,
in order to apply our abstract results to the case in which $A$ is the free $3D$ Laplacian and the regular part represents a perturbation by a potential, we give various regularity results for the boundary layer operators associated to $\Delta+\v$, where $\v$ is a potential of Kato-Rellich type. 
\par 
In Sections 5 and 6 we present various applications, where the free Laplacian is perturbed both by a regular term, given by a short range potential $\v$ decaying as $|x|^{-\kappa(1+\epsilon)}$,  and by a singular one describing either separating boundary conditions (as Dirichlet and Neumann ones) or semi-transparent (as $\delta$ and $\delta'$ type ones). In order to satisfy all our hypotheses, we need $\kappa=2$. However, all our hypotheses but a single one (see Lemma \ref{5.4}) hold with $\kappa=1$; we conjecture that the requirement $\kappa=2$ is merely of technical nature and that our results are true for a short range potential decaying as $|x|^{-(1+\epsilon)}$. Finally, let us remark that whenever one is only interested in the construction of the operators and not in the scattering theory, then it is sufficient to assume that $\v$ is a Kato-Rellich potential (see Subsection \ref{suff}). 
\par
 Schr\"odinger operators with a Kato-Rellich potential plus a $\delta$-like perturbation with a $p$-summable strength ($p>2$)  have been already considered in \cite{JDE18}, while for a different construction with a bounded potential  and a $\delta$- or a $\delta'$-like perturbation with bounded strength we refer to \cite{BLL}. None of such references considered the scattering matrix (however, \cite{JDE18} provided a limiting absorption principle). Whenever the singular part of the perturbations is absent, our framework extends from compactly supported potentials in one dimension to short range potentials in three dimensions the kind of results provided in \cite[Section 5]{BN}.
 \par
 Let us notice that, building on the results in \cite{Agmon} and \cite{JMPA}, the abstract models introduced in Section 2 and the related scattering theory presented in Section 3 apply to perturbations of the Laplacian in $\RE^{n}$, $n\ge2$, with a suitable short-range potential plus a singular term supported on a bounded hypersurface of co-dimension one.

\subsection{Some notation and definition.}
{\ }\par
\vskip5pt \noindent $\bullet$  $\|\cdot\|_{\X}$ denotes the norm on the complex Banach space $\X$; in case $\X$ is a Hilbert   space, $\langle\cdot,\cdot\rangle_{\X}$ denotes the (conjugate-linear w.r.t. the first argument) scalar product.
\vskip5pt\noindent $\bullet$ $\langle\cdot,\cdot\rangle_{\X^{*},\X}$ denotes the duality (assumed to be conjugate-linear w.r.t. the first argument) between the dual couple $(\X^{*},\X)$.
\vskip5pt\noindent $\bullet$ $L^{*}:\dom(L^{*})\subseteq \Y^{*}\to \X^{*}$ denotes the dual of the densely defined linear operator $L:\dom(L)\subseteq \X\to \Y$; in a Hilbert   spaces setting $L^{*}$ denotes the adjoint operator.
\vskip5pt\noindent $\bullet$ $\varrho(A)$ and $\sigma(A)$ denote the resolvent set and the spectrum of the self-adjoint operator $A$; $\sigma_{p}(A)$, $\sigma_{disc}(A)$, $\sigma_{ess}(A)$, 
$\sigma_{pp}(A)$, $\sigma_{cont}(A)$, $\sigma_{ac}(A)$,  $\sigma_{sc}(A)$,  denote the point, discrete, essential, pure point, continuous, absolutely continuous and singular continuous spectra.
\vskip5pt\noindent $\bullet$ $\B(\X,\Y)$, $\B(\X)\equiv \B(\X,\X)$, denote the Banach space of bounded linear operator on the Banach space $\X$ to the Banach space $\Y$; ${\|}\cdot {\|}_{\X,\Y}$ denotes the corresponding norm.
\vskip5pt\noindent $\bullet$ ${\mathfrak S}_{\infty}(\X,\Y)$ denotes the space of compact operators on  $\X$ to  $\Y$.
\vskip5pt\noindent $\bullet$ $\X\hookrightarrow \Y$ means that $\X$ is continuously embedded into $\Y$. 
\vskip5pt\noindent $\bullet$ $\Omega\equiv\Omega_{\inn}\subset\RE^{3}$ denotes an open and bounded subset with a Lipschitz boundary $\Gamma$; $\Omega_{\ex}:=\RE^{3}\backslash\overline\Omega$.
\vskip5pt\noindent $\bullet$ $H^{s}(\Omega)$ and $H^{s}(\Omega_{\ex})$ denote the scales of Sobolev spaces. 
\vskip5pt\noindent  $\bullet$ $H^{s}(\RE^{3}\backslash\Gamma):=H^{s}(\Omega)\oplus H^{s}(\Omega_{\ex})$.
\vskip5pt\noindent $\bullet$ $|x|$ denotes the norm of $x\in\RE^{n}$.
 $\langle x\rangle$ denotes the function $x\mapsto (1+|x|^{2})^{1/2}$.
\vskip5pt\noindent $\bullet$ $L_{w}^{2}(\RE^{3})$, $w\in\RE$, denotes the set of complex-valued functions $f$ such that $\langle x\rangle^{w}f\in L^{2}(\RE^{3})$. 
\vskip5pt\noindent  $\bullet$  $H_{w}^{s}(\RE^{3}\backslash\Gamma):=H^{s}(\Omega)\oplus H_{w}^{s}(\Omega_{\ex})$, where $H_{w}^{s}(\Omega_{\ex})$ denotes the weighted Sobolev space relative to the weight $\langle x\rangle^{w}$.
\vskip5pt\noindent $\bullet$ $\gamma_{0}^{\inn/\ex}$ and $\gamma_{1}^{\inn/\ex}$ denote the interior/exterior Dirichlet and Neumann traces on the boundary $\Gamma$. 
\vskip5pt\noindent $\bullet$  $\gamma_{0}:=\frac12(\gamma_{0}^{\inn}+\gamma_{0}^{\ex})$,   
 $\gamma_{1}:=\frac12(\gamma_{1}^{\inn}+\gamma_{1}^{\ex})$.
\vskip5pt\noindent $\bullet$ $[\gamma_{0}]:=\gamma_{0}^{\inn}-\gamma_{0}^{\ex}$,   
 $[\gamma_{1}]:=\gamma_{1}^{\inn}-\gamma_{1}^{\ex}$.  
\vskip5pt\noindent $\bullet$ $\SL_{z}$ and $\DL_{z}$ denote the single- and double-layer operators. 
\vskip5pt\noindent $\bullet$ $S_{z}:=\gamma_{0}\SL_{z}$, $D_{z}:=\gamma_{1}\DL_{z}$.
\vskip5pt\noindent $\bullet$ $D\subset\RE$ is said to be discrete in the open set $E\supset D $ whenever the (possibly empty) set of its accumulations point is contained in $\RE\backslash E$; $D$ is said to be discrete whenever $E=\RE$.
\vskip5pt\noindent $\bullet$ $\mathring D$ denotes the open part of the set $D\subseteq\RE$; $\partial D$ denotes its boundary; $D^{-}:=D\cap(-\infty,0]$.
\vskip5pt\noindent $\bullet$ Given $x\ge 0$ and $y\ge 0$, $x\lesssim y$ means that there exists $c\ge 0$ such that $x\le c\,y$.

\section{An abstract  Kre\u\i n-type resolvent formula}\label{Sec_Krein} 
\subsection{The resolvent formula} Let $A:\dom(A)\subseteq \H\to \H$ be a self-adjoint operator in the Hilbert space $\H$.  We denote by $R_{z}:=(-A+z)^{-1}$, $z\in \varrho(A)$, its resolvent; one has $R_{z}\in\B(\H,\H_{A})$, where $\H_{A}$ is the Hilbert space given by $\dom (A)$ equipped with the scalar product $$\langle u,u\rangle_{\H_{A}}:=\langle (A^{2}+1)^{1/2} u,(A^{2}+1)^{1/2}v\rangle_{\H}\,.
$$  
Let
$$
\h_{k}\hookrightarrow\h_{k}^{\circ}\hookrightarrow\h_{k}^{\ast}\,,\qquad k=1,2\,,
$$
be auxiliary Hilbert spaces with dense continuous embedding; we do not identify $\h_{k}$ with its dual $\h^{*}_{k}$ (however, we use $\h_{k}\equiv\h_{k}^{**}$) and we work with the $\h_{k}^{\ast}$-$\h_{k}$ duality 
$\langle\cdot,\cdot\rangle_{\h_{k}^{\ast},\h_{k}}$ defined in terms of the scalar
product of the intermediate Hilbert space $\h_{k}^{\circ}$. The scalar product and hence the duality are supposed to be conjugate linear with respect to the first variable; notice that $\langle\varphi,\phi\rangle_{\h_{k},\h^{\ast}_{k}}=\langle\phi,\varphi\rangle^{*}_{\h_{k}^{\ast},\h_{k}}$.\par
Given the bounded linear maps
$$\tau_{k}:\H_{A}\rightarrow\h_{k}\,,\qquad k=1,2\,,
$$
such that 
\be\label{tau2}
\text{$\ker(\tau_{2})$ is dense in $\H$ and $\ran(\tau_{2})$ is dense in $\h_{2}$
,} 
\ee
we introduce the bounded operators 
$$
\tau:\H_{A}\to \h_{1}\oplus\h_{2}\,,\qquad \tau u:=\tau_{1}u\oplus\tau_{2} u\,,
$$
and
$$
G_{z}:\h_{1}^{\ast}\oplus \h_{2}^{*}\to\H\,,\qquad  
G_{z}:=(\tau R_{\bar{z}})^{\ast}\,,\qquad z\in\varrho(A)\,.
$$
We further suppose that there exist reflexive Banach spaces $\b_{k}$, $k=1,2$, with dense continuous embeddings $\h_{k}\hookrightarrow\b_{k}$ (hence $\b_{k}^{*}\hookrightarrow\h_{k}^{*}$),
such that
$\ran(G_{z}|\b_{1}^{*}\oplus\b_{2}^{*})$ is contained in the domain of definition of some (supposed to exist) $(\b_{1}\oplus\b_{2})$-valued extension of $\tau$ (which we denote by the same symbol) in such a way that 
\begin{equation}
\tau G_{z}|\b_{1}^{*}\oplus\b_{2}^{*}\in{\B}(\b_{1}^{*}\oplus\b_{2}^{*},\b_{1}\oplus\b_{2})\,. \label{tauG}%
\end{equation}
Given these hypotheses, 
we set $\Bi=(B_{0},B_{1},B_{2})$, with 
\be\label{B012}
B_{0}\in{\B}(\b_{2}^{*},\b_{2,2}^{*})\,,\quad B_{1}\in\B(\b_{1},\b_{1}^{*})\,,\quad
B_{2}\in\B(\b_{2},\b^{*}_{2,2})\,, \quad\text{$\b_{2,2}$ a reflexive Banach space,}
\ee
\be\label{Bse}
\quad B_{1}=B_{1}^{*}\,,\qquad B_{0}B^{*}_{2}=B_{2}B^{*}_{0}\,,
\ee
and introduce the map
\be\label{Lambda}
Z_{\Bi}\ni z\mapsto\Lambda_{z}^{\!\Bi}\in{\B}(\b_{1}\oplus\b_{2},\b_{1}^{*}\oplus\b_{2}^{*})
\,,\qquad 
\Lambda_{z}^{\!\Bi}:=(M_{z}^{\Bi})^{-1}(  B_{1}\oplus B_{2})  \,,
\ee
where
\be\label{ZB}
Z_{\Bi}:=\big\{z\in\varrho(A): (M_{w}^{\Bi})^{-1}
\in{\B}(\b^{*}_{1}\oplus\b^{*}_{2,2},\b^{*}_{1}\oplus\b^{*}_{2})\,,\ w=z,\bar{z}\big\}
\ee
$$
M_{z}^{\Bi}:=(  1\oplus B_{0})  -(  B_{1}\oplus B_{2}) \tau G_{z} \in {\B}(\b^{*}_{1}\oplus\b^{*}_{2},\b^{*}_{1}\oplus\b^{*}_{2,2})\,.
$$
\begin{theorem}\label{Th_Krein} Suppose hypotheses \eqref{tau2}, \eqref{tauG}, \eqref{B012} and \eqref{Bse} hold and that $Z_{\Bi}$ defined in \eqref{ZB} is not empty. Then, defined $\Lambda_{z}^{\!\Bi}$ as in \eqref{Lambda},
\begin{equation}
R_{z}^{\Bi}:=R_{z}+G_{z}\Lambda_{z}^{\!\Bi}G_{\bar{z}}^{\ast}\,,\quad z\in
Z_{\Bi}\,, \label{resolvent}%
\end{equation}
is the resolvent of a self-adjoint operator $A_{\Bi}$ and $Z_{\Bi}%
=\varrho(A_{\Bi})\cap\varrho(A)$.
\end{theorem}
\begin{proof} By \eqref{Bse}, one gets
\begin{align*}
\big(  (  1\oplus B_{0})  -(  B_{1}\oplus B_{2})\tau
G_{\bar{z}}\big)(  B_{1}\oplus B^{*}_{2})  =&(  B_{1}\oplus B_{2})
\big(  (  1\oplus B_{0}^{*})  -\tau G_{\bar{z}}(  B_{1}\oplus B_{2}^{*})\big)\\
=&(  B_{1}\oplus B_{2})
\big(  (  1\oplus B_{0})  -(  B_{1}\oplus B_{2})\tau G_{{z}}\big)^{*}  \,.
\end{align*}
This entails, by the definitions \eqref{Lambda} and \eqref{ZB},
\be\label{PS1}
(\Lambda^{\!\Bi}_{z})^{*}=\Lambda^{\!\Bi}_{\bar z}\,.
\ee
By the resolvent identity, there follows
\begin{align*}
& ((  1\oplus B_{0})  -(  B_{1}\oplus B_{2})
\tau G_{z})  -((  1\oplus B_{0})  -(  B_{1}\oplus B_{2})  \tau G_{w}) \\ 
&=(  B_{1}\oplus B_{2})  \tau(G_{w}-G_{z})
=(z-w)(  B_{1}\oplus B_{2})  \tau R_{w}G_{z}\\
&=(z-w)(B_{1}\oplus B_{2})  G_{\bar{w}}^{\ast}G_{z}\,,
\end{align*}
which entails 
\begin{align*}
& ( (  1\oplus B_{0})  -(  B_{1}\oplus B_{2})
\tau G_{w})  ^{-1}-((  1\oplus B_{0})  -
(B_{1}\oplus B_{2})  \tau G_{z})  ^{-1}  \\
=&(z-w)((  1\oplus B_{0})  -(  B_{1}\oplus B_{2})  \tau G_{w})  ^{-1}
(  B_{1}\oplus B_{2})G_{\bar{w}}^{\ast}G_{z}((  1\oplus B_{0})
 -(B_{1}\oplus B_{2})  \tau G_{z})  ^{-1}\,,
\end{align*}
and hence
\be\label{PS2}
\Lambda_{w}^{\!\Bi}-\Lambda_{z}^{\!\Bi}=(z-w)\Lambda_{w}^{\!\Bi}G_{\bar{w}}^{\ast
}G_{z}\Lambda_{z}^{\!\Bi}\,.
\ee
By \eqref{PS1} and $\eqref{PS2}$, 
$$(R_{z}^{\Bi})^{*}=R_{\bar z}^{\Bi}\,,\qquad R_{z}^{\Bi}=R_{w}^{\Bi}+(w-z)R_{z}^{\Bi}R_{w}^{\Bi}\,.
$$
(see \cite[page 113]{JFA}). Hence, $R_{z}^{\Bi}$ is the resolvent of a self-adjoint operator whenever it is injective (see, e.g., \cite[Theorems 4.10 and 4.19]{Stone}). By \eqref{resolvent},
\begin{align*}
&(B_{1}\oplus B_{2})\tau R_{z}^{\Bi}=(B_{1}\oplus B_{2})\big(1+\tau G_{z}\Lambda_{z}^{\!\Bi}\big)G_{\bar z}^{*}=\big((B_{1}\oplus B_{2})+(B_{1}\oplus B_{2})\tau G_{z}\Lambda_{z}^{\!\Bi}\big)G_{\bar z}^{*}\\
=&\big((B_{1}\oplus B_{2})+\big((1\oplus B_{0})-\big((1\oplus B_{0})-(B_{1}\oplus B_{2})\tau G_{z}\big)\big)\Lambda_{z}^{\!\Bi}\big)G_{\bar z}^{*}=(1\oplus B_{0})\Lambda_{z}^{\!\Bi}G_{\bar z}^{*}\,.
\end{align*}
Thus, if $R_{z}^{\Bi}u=0$ then 
$$
0\oplus 0=(1\oplus B_{0})\Lambda_{z}^{\!\Bi}G_{\bar z}^{*}u=\big(\Lambda_{z}^{\!\Bi}G_{\bar z}^{*}u\big)_{1}\oplus B_{0}\big(\Lambda_{z}^{\!\Bi}G_{\bar z}^{*}u\big)_{2}
$$
By
$$
G_{z}(\phi_{1}\oplus\phi_{2})=G^{1}_{z}\phi_{1}+G^{2}_{z}\phi_{2}\,, \qquad G^{k}_{z}:=(\tau_{k}R_{\bar{z}})^{\ast}\,,
$$
there follows 
\begin{equation}
0=R_{z}^{\Bi}u=R_{z}u+G^{1}_{z}\big(\Lambda_{z}^{\!\Bi}G_{\bar{z}}^{\ast}
u\big)_{1}+G^{2}_{z}\big(\Lambda_{z}^{\!\Bi}G_{\bar{z}}^{\ast}u\big)_{2}=R_{z}u+G^{2}_{z}
\big(\Lambda_{z}^{\!\Bi}G_{\bar{z}}^{\ast}u\big)_{2}\,.\label{z2}%
\end{equation}
Since 
the denseness of $\ker(\tau_{2})$ implies
$\text{\textrm{ran}}(G^{2}_{z})\cap \dom(A)=\{0\}$ 
(see \cite[Remark 2.9]{JFA}), the relation \eqref{z2} gives $G^{2}_{z}(\Lambda_{z}^{\!\Bi}G_{\bar{z}}^{\ast}u)_{2}=0$. Thus  $R_{z}^{\Bi}u=0$ compels $R_{z}u=0$ and hence $u=0$.\par 
Finally, the equality $Z_{\Bi}=\varrho(A_{\Bi})\cap\varrho(A)$ is consequence of \cite[Theorem 2.19 and
Remark 2.20]{CFP}.
\end{proof}
\begin{remark} Looking at the previous proof, one notices that Theorem \ref{Th_Krein} holds without requiring the denseness of $\ran(\tau_{2})$; that hypothesis comes into play in later results.
\end{remark}
\begin{remark} By \eqref{resolvent}, if $u\in\dom(A_{\Bi})$, then $u=u_{0}+G_{z}(\phi_{1}\oplus\phi_{2})$ for some $u_{0}\in\H_{A}$ and $\phi_{1}\oplus\phi_{2}\in \b_{1}^{*}\oplus\b_{2}^{*}$; hence, by \eqref{tauG}, 
$$
\tau:\dom(A_{\Bi})\to \b_{1}\oplus\b_{2}\,.
$$  
\end{remark}
\subsection{An additive representation}
At first, let us introduce the Hilbert space $\H_{A}^{*}$ defined as the completion of $\H$ endowed with the scalar product $$\langle u,v\rangle_{\H_{A}^{*}}:=\langle(A^{2}+1)^{-1/2}u,(A^{2}+1)^{-1/2}v\rangle_{\H}\,.
$$
Notice that that $R_{z}$ extends to a bounded bijective map (which we denote by the same symbol) on $\H_{A}^{*}$ onto $\H$. The linear operator $A$, being a densely defined bounded operator on $\H$ to  $\H_{A}^{*}$, extends to a bounded operator $\overline {\!A}:\H\to\H_{A}^{*}$ given by its closure. Moreover, denoting by $\langle\cdot,\cdot\rangle_{\H_{A}^{*},\H_{A}}$  the pairing obtained by extending the scalar product in $\H$, since $A$ is self-adjoint and since $\dom(A)$ is dense in $\H$,
$$
\langle u,Av\rangle_{\H}=\langle\, \overline{\!A}u,v\rangle_{\H_{A}^{*},\H_{A}}\,,\qquad u\in\H\,,\ v\in\H_{A}, \,.
$$
Further, we define $\tau^{*}:\h_{1}^{*}\oplus\h^{*}_{2}\to\H_{A}^{*}$ by 
\be\label{tau*}
\langle \tau^{*}\phi,u\rangle_{\H_{A}^{*},\H_{A}}=\langle\phi,\tau u\rangle_{\h_{1}^{*}\oplus\h^{*}_{2},\h_{1}\oplus\h_{2}}\,,\qquad u\in\H_{A}\,,\ \phi\in\\h_{1}^{*}\oplus\h^{*}_{2}\,.
\ee
Obviously, $\tau^{*}(\phi_{1}\oplus\phi_{2})=\tau_{1}^{*}\phi_{1}+\tau^{*}_{2}\phi_{2}$, where 
$\tau_{k}^{*}:\h_{k}\to\H_{A}^{*}$, $k=1,2$, are defined in the same way as $\tau^{*}$.
\par
Let us notice that $R_{z}:\H_{A}^{*}\to \H$ is the adjoint, with respect the pairing $\langle\cdot,\cdot\rangle_{\H_{A}^{*},\H_{A}}$, of $R_{\bar z}:\H_{A}\to \H$ and it is the inverse of $(-\overline{\!A} +z):\H\to\H_{A}^{*}$; therefore 
\be\label{blw}
G_{z}=R_{z}\tau^{*}\,.
\ee
\begin{lemma}
\label{le-green}Let $A_{\Bi}:\dom(A_{\Bi})\subseteq\H\to\H$ be the self-adjoint
operator provided in Theorem \ref{Th_Krein} and define
\begin{equation}
\rho_{\Bi}:\dom(A_{\Bi})\rightarrow\h_{1}^{*}\oplus\h_{2}^{*}\,,\quad
\rho_{\Bi}(R_{z}^{\Bi}u):=(\pi_{1}^{\ast}\oplus 1)\Lambda_{z}^{\!\Bi}G_{\bar{z}}^{\ast
}u\,,\qquad u\in \H\,, \quad z\in\varrho(A_{\Bi})\cap\varrho(A)\,,\label{rho_def}%
\end{equation}
where $\pi_{1}$ denotes the orthogonal projection onto the subspace
$\overline{\ran(\tau_{1})}$. Then, the definition of $\rho_{\Bi}$ is
well-posed, i.e.,
\[
R_{z_{1}}^{\Bi}u_{1}=R_{z_{2}}^{\Bi}u_{2}\quad\implies\quad (\pi_{1}^{\ast}\oplus 1)\Lambda_{z_{1}}^{\!\Bi}G_{\bar{z}_{1}}^{\ast}u_{1}=(\pi_{1}^{\ast}\oplus 1)
\Lambda_{z_{2}}^{\!\Bi}G_{\bar{z}_{2}}^{\ast}u_{2}%
\]
and 
\begin{equation}\label{GF}
\langle u,A_{\Bi}v\rangle_{\H}=\langle A
u,v\rangle_{\H}+
\langle\tau u,\rho_{\Bi} v\rangle_{\h_{1}\oplus\h_{2},\h^{*}_{1}\oplus\h_{2}^{*}}\,,\quad u\in \dom(A),\,v\in
\dom(A_{\Bi})\,.
\end{equation}
\end{lemma}
\begin{proof}
Let $v=R_{z}^{\Bi}u=v_{z}+G_{z}\Lambda_{z}^{\!\Bi}\tau v_{z}$, where $v_{z}%
:=R_{z}u$ (hence $\tau v_{z}=G_{\bar{z}}^{\ast}u$). Then
\begin{align*}
&  \langle u,A_{\Bi}v\rangle_{\H}-\langle A
u,v\rangle_{\H}\\
=  &  -\langle u,(-A_{\Bi}+z)v\rangle_{\H}%
+\langle(-A+\bar{z})u,v\rangle_{\H}\\
=  &  -\langle u,(-A+z)v_{z}\rangle_{\H}%
+\langle(-A+\bar{z})u,v_{z}+G_{z}\Lambda_{z}^{\!\Bi}\tau v_{z}%
\rangle_{\H}\\
=  &  \langle(-A+\bar{z})u,G_{z}\Lambda_{z}^{\!\Bi}\tau v_{z}\rangle
_{\H}=\langle\tau u,\Lambda_{z}^{\!\Bi}\tau v_{z}%
\rangle_{\h_{1}\oplus\h_{2},\h^{*}_{1}\oplus\h_{2}^{*}}\\
=  &  \langle (\pi_{1}\oplus 1)\tau u,\Lambda_{z}^{\!\Bi}\tau v_{z}\rangle
_{\h_{1}\oplus\h_{2},\h^{*}_{1}\oplus\h_{2}^{*}}=\langle\tau u,(\pi_{1}^{\ast}\oplus 1)\Lambda
_{z}^{\!\Bi}\tau v_{z}\rangle_{\h_{1}\oplus\h_{2},\h^{*}_{1}\oplus\h_{2}^{*}}\,.
\end{align*}
Suppose now that $R_{z_{1}}^{\Bi}u_{1}=R_{z_{2}}^{\Bi}u_{2}$. Then, by the above
identities, one gets, for any $u\in \dom(A)$,
\[
\langle \tau^{\ast}(\pi_{1}^{\ast}\oplus 1)(\Lambda_{z_{1}}^{\!\Bi}G_{\bar{z}_{1}%
}^{\ast}u_{1}-\Lambda_{z_{2}}^{\!\Bi}G_{\bar{z}_{2}}^{\ast}u_{2}),u\rangle
_{\H_{A}^{*},\H_{A}}=0\,.
\]
Hence $\tau^{\ast}((\pi_{1}^{\ast}\oplus 1)\Lambda_{z_{1}}^{\!\Bi}G_{\bar{z}_{1}}%
^{\ast}u_{1}-(\pi_{1}^{\ast}\oplus 1)\Lambda_{z_{2}}^{\!\Bi}G_{\bar{z}_{2}}^{\ast}%
u_{2})=0$. However, $\ker(\tau^{\ast})\cap\ran((\pi_{1}^{\ast}\oplus 1))=\{0\}$ since $\pi_{1}^{\ast}\oplus 1$ is the projector onto the subspace
orthogonal to $\ker(\tau^{\ast})$.
\end{proof}
The next Lemma provides a sort of abstract boundary conditions holding for the elements in $\dom(A_{\Bi})$:
\begin{lemma}\label{abc} Let $A_{\Bi}$ be the self-adjoint operator in Theorem \ref{Th_Krein}. Then, for any $z\in\varrho(A_{\Bi})\cap\varrho(A)$, one has the representation
$$
\dom(A_{\Bi})=\{u\in\H:u_{z}:=u-G_{z}\rho_{\Bi}u\in \dom(A)\}\,,
$$
$$
(-A_{\Bi}+z)u=(-A+z)u_{z}\,.
$$
Moreover, 
$$
u\in \dom(A_{\Bi})\quad\Longrightarrow\quad (\pi_{1}^{*}B_{1}\oplus B_{2})\tau u=(1\oplus B_{0})\rho_{\Bi} u\,.
$$
\end{lemma}
\begin{proof} Since $G_{z}=R_{z}\tau^{*}$ (see \eqref{blw} below) and $\pi_{1}^{*}\oplus 1$ is the projection onto the orthogonal to $\ker(\tau^{*})$, one has $G_{z}=G_{z}(\pi_{1}^{*}\oplus 1) $. 
Hence, $u\in\dom(A_{\Bi})$ if and only if $u=R_{z}v+G_{z}(\pi_{1}^{*}\oplus 1) \Lambda_{z}^{\Bi}G^{*}_{\bar z}v=R_{z}v+G_{z}\rho_{\Bi}u$. Therefore, 
$$
\dom(A_{\Bi})=\{u\in\H:u=u_{z}+G_{z}\rho_{\Bi}u\,, \ u_{z}\in\dom(A)\}\,.
$$
Moreover, given any $u\in\dom(A)$, $u=R^{\Bi}_{z}v$, one has
$$
(-A+z)u_{z}=(-A+z)R_{z}v=(-A_{\Bi}+z)R^{\Bi}_{z}v=(-A_{\Bi}+z)u\,.
$$
Finally, given $u=R^{\Bi}_{z}v\in\dom (A_{\Bi})$,  one has
\begin{align*}
&(\pi_{1}^{*}B_{1}\oplus B_{2})\tau u=(\pi_{1}^{*}\oplus 1)(B_{1}\oplus B_{2})\tau R^{\Bi}_{z}v\\
=&(\pi_{1}^{*}\oplus 1)\big((B_{1}\oplus B_{2})G_{\bar z}v+ 
(B_{1}\oplus B_{2})\tau G_{z}\big((1\oplus B_{0})-(B_{1}\oplus B_{2})\tau G_{z}\big)^{-1}(B_{1}\oplus B_{2})G_{\bar z}v\big)\\
=&(\pi_{1}^{*}\oplus 1)(1\oplus B_{0})\Lambda^{\Bi}_{z}G_{\bar z}v=(1\oplus B_{0})(\pi_{1}^{*}\oplus 1)\Lambda^{\Bi}_{z}G_{\bar z}v=(1\oplus B_{0})\rho_{\Bi} u\,.
\end{align*}
\end{proof} 
Now, we provide an additive representation of the self-adjoint $A_{\Bi}$ in Theorem \ref{Th_Krein}. 
\begin{theorem}\label{Th-add} Let $A_{\Bi}:\dom(A_{\Bi})\subseteq \H\to\H$ be the self-adjoint operator appearing in Theorem \ref{Th_Krein}. Then
$$
A_{\Bi}=\overline {\!A}+\tau^{*}\!\rho_{\Bi}\,,
$$
where $\rho_{\Bi}$ is defined in \eqref{rho_def}. In particular, if $B_{0}^{-1}\in\B(\b_{2,2}^{*},\b_{2}^{*})$, then
$$
A_{\Bi}=\overline {\!A}+\tau_{1}^{*}B_{1}\tau_{1}+\tau_{2}^{*}B_{0}^{-1}\!B_{2}\tau_{2}\,.
$$ 
\end{theorem}  
\begin{proof} By \eqref{GF}, for any $u\in\dom(A_{\Bi})$ and $v\in\H_{A}$,
\begin{align*}
\langle A_{\Bi}u,v\rangle_{\H_{A}^{*},\H_{A}}\equiv&\langle A_{\Bi}u,v\rangle_{\H}=
\langle u,Av\rangle_{\H}+\langle\rho_{\Bi} u,\tau v\rangle_{\h^{*}_{1}\oplus\h^{*}_{2},\h_{1}\oplus\h_{2}}\\
=&\langle\, \overline{\!A}u+\tau^{*}\!\rho_{\Bi} u,v\rangle_{\H_{A}^{*},\H_{A}}\,.
\end{align*}
By Lemma \ref{abc} and by $\tau_{1}^{*}\pi_{1}^{*}=(\pi_{1}\tau_{1})^{*}=\tau_{1}^{*}$,
$$
\tau^{*}\!\rho_{\Bi} =\tau^{*}(\pi^{*}_{1}B_{1}\tau_{1}\oplus B_{0}^{-1}B_{1}\tau_{2})=
\tau_{1}^{*}B_{1}\tau_{1}+\tau_{2}^{*}B_{0}^{-1}\!B_{2}\tau_{2}\,.
$$
\end{proof}

\subsection{An alternative resolvent formula.} At first, let us notice that hypothesis \eqref{tauG}, can be re-written as
$$
\tau_{j}G^{k}_{z}|\b_{k}\in\B(\b_{k}^{*},\b_{j})\,,\quad j,k=1,2\,,\qquad G^{k}_{z}:=(\tau_{k}R_{\bar z})^{*}\,.
$$
Moreover, 
$$
{M}_{z}^{\Bi}=(1\oplus B_{0})+(B_{1}\oplus B_{2})\tau G_{z}=
\begin{bmatrix}M_{z}^{B_{1}}&B_{1}\tau_{1}G^{2}_{z}\\
B_{2}\tau_{2}G^{1}_{z}&M_{z}^{B_{0},B_{2}}
\end{bmatrix}
$$
where
$$
M_{z}^{B_{1}}:=1-B_{1}\tau_{1}G^{1}_{z}\,,\qquad 
M_{z}^{B_{0},B_{2}}:=B_{0}-B_{2}\tau_{2}G^{2}_{z}\,.
$$
Then, supposing all the inverse operators appearing in the next formula exist, by the inversion formula for block operator matrices, one gets
\begin{align}\label{MB-1}
&({M}_{z}^{\Bi})^{-1}=
\begin{bmatrix}(M^{B_{1}}_{z})^{-1}+(M^{B_{1}}_{z})^{-1}B_{1}\tau_{1}G^{2}_{z}({C}_{z}^{\Bi})^{-1}B_{2}\tau_{2}G^{1}_{z}(M^{B_{1}}_{z})^{-1}&
(M^{B_{1}}_{z})^{-1}B_{1}\tau_{1}G^{2}_{z}({C}_{z}^{\Bi})^{-1}\\
({C}_{z}^{\Bi})^{-1}B_{2}\tau_{2}G^{1}_{z}(M^{B_{1}}_{z})^{-1}&({C}_{z}^{\Bi})^{-1}
\end{bmatrix}
\end{align}
where ${C}^{\Bi}_{z}$ denotes the second Schur complement, i.e.,
\begin{align*}
{C}^{\Bi}_{z}:=&M_{z}^{B_{0},B_{2}}-B_{2}\tau_{2}G^{1}_{z}
(M_{z}^{B_{1}})^{-1}B_{1}\tau_{1}G^{2}_{z}\\
=&M_{z}^{B_{0},B_{2}}\left(1-(M_{z}^{B_{0},B_{2}})^{-1}B_{2}\tau_{2}G^{1}_{z}
(M_{z}^{B_{1}})^{-1}B_{1}\tau_{1}G^{2}_{z}\right)\\
=&M_{z}^{B_{0},B_{2}}\left(1-\Lambda_{z}^{\!B_{0},B_{2}}\tau_{2}G^{1}_{z}
\Lambda_{z}^{\!B_{1}}\tau_{1}G^{2}_{z}\right)\,,
\end{align*}
\be\label{LB1}
\Lambda_{z}^{\!B_{1}}:=(1-B_{1}\tau_{1}G^{1}_{z})^{-1}B_{1}\,,
\ee
\be\label{LB02}
\Lambda_{z}^{\!B_{0},B_{2}}:=(B_{0}-B_{2}\tau_{2}G^{2}_{z})^{-1}B_{2}\,.
\ee
Regarding the well-posedness of \eqref{MB-1}, taking into account the definition of ${C}^{\Bi}_{z}$, one has 
$$
Z_{\Bi}=\big\{z\in\varrho(A): (M^{\Bi}_{z})^{-1}\in{\B}(\b^{*}_{1}\oplus\b^{*}_{2,2},\b^{*}_{1}\oplus\b^{*}_{2})\,,\ w=z,\bar{z}\big\}\supseteq \widehat Z_{\Bi}\,,
$$
where
\be\label{wZB}
\widehat Z_{\Bi}:=\big\{z\in Z_{B_{1}}\cap Z_{B_{0},B_{2}}:\left(1-\Lambda_{w}^{\!B_{0},B_{2}}\tau_{2}G^{1}_{w}
\Lambda_{w}^{\!B_{1}}\tau_{1}G^{2}_{w}\right)^{-1}\in\B(\b_{2}^{*}),\quad w=z,\bar z\big\}
\ee
\be\label{ZB1}
Z_{B_{1}}:=\big\{z\in\varrho(A): (1-B_{1}\tau_{1}G^{1}_{w})^{-1}\in{\B}(\b^{*}_{1})\,,\ w=z,\bar{z}\big\}\,,
\ee
\be\label{ZB02}
Z_{B_{0},B_{2}}:=\big\{z\in\varrho(A): (B_{0}-B_{2}\tau_{2}G^{2}_{w})^{-1}\in{\B}(\b^{*}_{2,2},\b^{*}_{2})\,,\ w=z,\bar{z}\big\}\,,
\ee
Therefore, supposing that $\widehat Z_{\Bi}$ is not empty, for any $z\in\widehat Z_{\Bi}$, by \eqref{resolvent} and by
$$
({C}_{z}^{\Bi})^{-1}B_{2}=\Sigma^{\Bi}_{z}\Lambda_{z}^{\!B_{0},B_{2}}\,,\qquad \Sigma^{\Bi}_{z}:=\left(1-\Lambda_{z}^{\!B_{0},B_{2}}\tau_{2}G^{1}_{z}
\Lambda_{z}^{\!B_{1}}\tau_{1}G^{2}_{z}\right)^{-1}\,,
$$
 one has
\begin{align*}
\Lambda_{z}^{\!\Bi}=(M^{\Bi}_{z})^{-1}
\begin{bmatrix}B_{1}&0\\
0&B_{2}\end{bmatrix}
=&
\begin{bmatrix}\Lambda^{\!B_{1}}_{z}+\Lambda^{\!B_{1}}_{z}\tau_{1}G^{2}_{z}\Sigma^{\Bi}_{z}\Lambda_{z}^{\!B_{0},B_{2}}\tau_{2}G^{1}_{z}\Lambda^{\!B_{1}}_{z}&
\Lambda^{\!B_{1}}_{z}\tau_{1}G^{2}_{z}\Sigma^{\Bi}_{z}\Lambda_{z}^{\!B_{0},B_{2}}\\
\Sigma^{\Bi}_{z}\Lambda_{z}^{\!B_{0},B_{2}}\tau_{2}G^{1}_{z}\Lambda^{\!B_{1}}_{z}&\Sigma^{\Bi}_{z}\Lambda_{z}^{\!B_{0},B_{2}}
\end{bmatrix}
\,.
\end{align*}
Therefore
\be\label{res}
R_{z}^{\Bi}=R_{z}+\begin{bmatrix}G^{1}_{z}&G^{2}_{z}\end{bmatrix}
\begin{bmatrix}\Lambda^{\!B_{1}}_{z}+\Lambda^{\!B_{1}}_{z}\tau_{1}G^{2}_{z}\Sigma^{\Bi}_{z}\Lambda_{z}^{\!B_{0},B_{2}}\tau_{2}G^{1}_{z}\Lambda^{\!B_{1}}_{z}&
\Lambda^{\!B_{1}}_{z}\tau_{1}G^{2}_{z}\Sigma^{\Bi}_{z}\Lambda_{z}^{\!B_{0},B_{2}}\\
\Sigma^{\Bi}_{z}\Lambda_{z}^{\!B_{0},B_{2}}\tau_{2}G^{1}_{z}\Lambda^{\!B_{1}}_{z}&\Sigma^{\Bi}_{z}
\Lambda_{z}^{\!B_{0},B_{2}}\end{bmatrix}
\begin{bmatrix}{G^{1*}_{\bar z}}\\{G^{2*}_{\bar z}}\end{bmatrix}\,.
\ee
In particular, taking $\Bi=(1,B_{1},0)$, one gets, for any $z\in Z_{B_{1}}$,
\begin{align}\label{res1}
R_{z}^{B_{1}}:=R_{z}^{(1,B_{1},0)}
=&R_{z}+\begin{bmatrix}G^{1}_{z}&G^{2}_{z}\end{bmatrix}
\begin{bmatrix}\Lambda_{z}^{\!B_{1}}&0\\
0&0
\end{bmatrix}
\begin{bmatrix}{G^{1*}_{\bar z}}\\{G^{2*}_{\bar z}}\end{bmatrix}
=R_{z}+G^{1}_{z}\Lambda_{z}^{\!B_{1}}{G^{1*}_{\bar z}}
\end{align}
while, taking $\Bi=(B_{0},0,B_{2})$, one gets, for any $z\in Z_{B_{0},B_{2}}$, 
\begin{align}\label{res2}
R_{z}^{B_{0},B_{2}}:=R_{z}^{(B_{0},0,B_{2})}
=&R_{z}+\begin{bmatrix}G^{1}_{z}&G^{2}_{z}\end{bmatrix}
\begin{bmatrix}0&
0\\
0&\Lambda_{z}^{\!B_{0},B_{2}}
\end{bmatrix}
\begin{bmatrix}{G^{1*}_{\bar z}}\\{G^{2*}_{\bar z}}\end{bmatrix}
=R_{z}+G^{2}_{z}\Lambda_{z}^{\!B_{0},B_{2}}{G^{2*}_{\bar z}}\,.
\end{align}
Therefore, by Theorem \ref{Th_Krein} with $\Bi=(1,B_{1},0)$, one gets
\begin{corollary}\label{cor1} Let $\tau_{1}\in\B(\H_{A},\h_{1})$ such that $\tau_{1}G^{1}_{z}|\b^{*}_{1}\in\B(\b_{1}^{*},\b_{1})$ and let $B_{1}\in\B(\b_{1},\b^{*}_{1})$ self-adjoint; suppose that  
$Z_{B_{1}}$ defined in \eqref{ZB1} is not empty. Then 
\be\label{res1.1}
R_{z}^{B_{1}}=R_{z}+G^{1}_{z}\Lambda_{z}^{\!B_{1}}{G^{1*}_{\bar z}}\,,\qquad z\in Z_{B_{1}}\,,
\ee
where $\Lambda_{z}^{\!B_{1}}$ is defined in \eqref{LB1}, is the resolvent of a self-adjoint operator $A_{B_{1}}$ and $Z_{B_{1}}=\varrho(A_{B_{1}})\cap\varrho(A)$.
\end{corollary}
By Theorem \ref{Th_Krein} with $\Bi=(B_{0},0,B_{2})$, one gets
\begin{corollary}\label{cor02} Let $\tau_{2}\in\B(\H_{A},\h_{2})$ satisfy \eqref{tau2} be such that $\tau_{1}G^{1}_{z}|\b^{*}_{2}\in\B(\b_{2}^{*},\b_{2})$ and let $B_{0}\in\B(\b_{2}^{*},\b_{2,2}^{*})$, $B_{2}\in \B(\b_{2},\b_{2,2}^{*})$ be such that $B_{0}B_{2}^{*}=B_{2}B_{0}^{*}$; suppose that
$Z_{B_{0},B_{2}}$ defined in \eqref{ZB02} is not empty. Then 
\be\label{res2.1}
R_{z}^{B_{0},B_{2}}=R_{z}+G^{2}_{z}\Lambda_{z}^{\!B_{0},B_{2}}{G^{2*}_{\bar z}}\,,\qquad z\in Z_{B_{0},B_{2}}
\ee
where $\Lambda_{z}^{\!B_{0},B_{2}}$ is defined in \eqref{LB02}, is the resolvent of a self-adjoint operator $A_{B_{0},B_{2}}$ and $Z_{B_{0},B_{2}}=\varrho(A_{B_{0},B_{2}})\cap\varrho(A)$.
\end{corollary}
Supposing $\widehat Z_{\Bi}\not=\varnothing $, by \eqref{res}, by \eqref{res1} and by the relations
\begin{align}\label{GB1}
G^{B_{1}}_{z}:=(\tau_{2}R_{\bar z}^{B_{1}})^{*}=&
(\tau_{2}R_{\bar z}+\tau_{2} G^{1}_{\bar z}\Lambda^{\!B_{1}}_{\bar z}{G^{1*}_{z}})^{*}\\
=&G_{z}^{2}+G_{z}^{1}\Lambda^{\!B_{1}}_{z}\tau_{1}G_{z}^{2}\nonumber
\end{align}
\begin{align}\label{GB1*}
{G^{B_{1}*}_{\bar z}}=\tau_{2}R_{z}^{B_{1}}=&
\tau_{2}R_{z}+\tau_{2} G^{1}_{ z}\Lambda^{\!B_{1}}_{z}{G^{1*}_{\bar z}}
\\=&G_{\bar z}^{2*}+\tau_{2}G_{z}^{1}\Lambda^{\!B_{1}}_{z}G_{\bar z}^{1*}\nonumber
\end{align}
\begin{align*}
\widehat M_{z}^{\Bi}=B_{0}-B_{2}\tau_{2}G^{B_{1}}_{z}=&B_{0}-B_{2}\tau_{2}G_{z}^{2}+\tau_{2}G_{z}^{1}\Lambda^{\!B_{1}}_{z}\tau_{1}G_{z}^{2}\\
=&M_{z}^{B_{0},B_{2}}+B_{2}\tau_{2}G_{z}^{1}\Lambda^{\!B_{1}}_{z}\tau_{1}G_{z}^{2}\\
=&M_{z}^{B_{0},B_{2}}\big(1+\Lambda^{\!B_{0},B_{1}}_{z}\tau_{2}G_{z}^{1}\Lambda^{\!B_{1}}_{z}\tau_{1}G_{z}^{2}\big)
\end{align*}
\be\label{wLB1}
\widehat \Lambda^{\Bi}_{z}:=(\widehat M_{z}^{\Bi})^{-1}B_{2}=(B_{0}-B_{2}\tau_{2}G^{B_{1}}_{z})^{-1}B_{2}=\Sigma_{z}^{\Bi}\Lambda^{\!B_{0},B_{2}}_{z}
\ee
 one gets 
\begin{align}
&\Lambda_{z}^{\Bi}=\begin{bmatrix}\Lambda^{\!B_{1}}_{z}+\Lambda^{\!B_{1}}_{z}\tau_{1}G^{2}_{z}\widehat \Lambda^{\Bi}_{z}\tau_{2}G^{1}_{z}\Lambda^{\!B_{1}}_{z}&
\Lambda^{\!B_{1}}_{z}\tau_{1}G^{2}_{z}\widehat \Lambda^{\Bi}_{z}\\
\widehat \Lambda^{\Bi}_{z}\tau_{2}G^{1}_{z}\Lambda_{z}^{\!B_{1}}&
\widehat \Lambda^{\Bi}_{z}\,;
\end{bmatrix}\label{LB-new}\\
=&\left(1+\begin{bmatrix}\Lambda^{\!B_{1}}_{z}&0\\0&\widehat \Lambda^{\Bi}_{z}
\end{bmatrix}\begin{bmatrix}\tau_{1}G^{2}_{z}\widehat \Lambda^{\Bi}_{z}\tau_{2}G^{1}_{z}&
\tau_{1}G^{2}_{z}\\
\tau_{2}G^{1}_{z}&0
\end{bmatrix}\,\right)\begin{bmatrix}\Lambda^{\!B_{1}}_{z}&0\\0&\widehat \Lambda^{\Bi}_{z}
\end{bmatrix}\,.\label{LB-new2}
\end{align}
Therefore 
\begin{align}\label{res-new}
R_{z}^{\Bi}=&R_{z}+\begin{bmatrix}G^{1}_{z}&G^{2}_{z}\end{bmatrix}
\begin{bmatrix}\Lambda^{\!B_{1}}_{z}+\Lambda^{\!B_{1}}_{z}\tau_{1}G^{2}_{z}\widehat \Lambda^{\Bi}_{z}\tau_{2}G^{1}_{z}\Lambda^{\!B_{1}}_{z}&
\Lambda^{\!B_{1}}_{z}\tau_{1}G^{2}_{z}\widehat \Lambda^{\Bi}_{z}\\
\widehat \Lambda^{\Bi}_{z}\tau_{2}G^{1}_{z}\Lambda_{z}^{\!B_{1}}&
\widehat \Lambda^{\Bi}_{z}
\end{bmatrix}
\begin{bmatrix}
{G^{1*}_{\bar z}}\\ 
{G^{2*}_{\bar z}}
\end{bmatrix}\\
=&R_{z}+\begin{bmatrix}G^{1}_{z}&G^{2}_{z}\end{bmatrix}
\begin{bmatrix}\Lambda^{\!B_{1}}_{z}G^{1*}_{\bar z}+\Lambda^{\!B_{1}}_{z}\tau_{1}G^{2}_{z}\widehat \Lambda^{\Bi}_{z}\tau_{2}G^{1}_{z}\Lambda^{\!B_{1}}_{z}G^{1*}_{\bar z}+
\Lambda^{\!B_{1}}_{z}\tau_{1}G^{2}_{z}\widehat \Lambda^{\Bi}_{z}G^{2*}_{\bar z}\\
\widehat \Lambda^{\!\Bi}_{z}\tau_{2}G^{1}_{z}\Lambda_{z}^{\!B_{1}}G^{1*}_{\bar z}+
\widehat \Lambda^{\Bi}_{z}G^{2*}_{\bar z}
\end{bmatrix}\nonumber\\
=&R_{z}+G^{1}_{z}\Lambda^{\!B_{1}}_{z}G^{1*}_{\bar z}+G^{1}_{z}\Lambda^{\!B_{1}}_{z}\tau_{1}G^{2}_{z}\widehat \Lambda^{\Bi}_{z}\tau_{2}G^{1}_{z}\Lambda^{\!B_{1}}_{z}G^{1*}_{\bar z}+
G^{1}_{z}\Lambda^{\!B_{1}}_{z}\tau_{1}G^{2}_{z}\widehat \Lambda^{\Bi}_{z}G^{2*}_{\bar z}\nonumber\\
&+G^{2}_{z}\widehat \Lambda^{\Bi}_{z}\tau_{2}G^{1}_{z}\Lambda_{z}^{\!B_{1}}G^{1*}_{\bar z}+
G^{2}_{z}\widehat \Lambda^{\Bi}_{z}G^{2*}_{\bar z}\nonumber\\
=&R^{B_{1}}_{z}+G_{z}^{B_{1}}\widehat \Lambda^{\Bi}_{z}G_{\bar z}^{B_{1}*}\nonumber\,.
\end{align}
This also entails, by \cite[Theorem 2.19 and Remark 2.20]{CFP}, that if $\widehat Z_{\Bi}\not=\varnothing $, then $\widehat Z_{\Bi}=Z_{\Bi}=\varrho(A_{\Bi})\cap\varrho(A_{B_{1}})$. Summing up, one has the following 
\begin{theorem}\label{Th-alt-res} Assume that  hypotheses \eqref{tauG}, \eqref{B012} and \eqref{Bse} hold and that $\widehat Z_{\Bi}$ defined in \eqref{wZB} is not empty. Then, for any $z\in \varrho(A_{\Bi})\cap\varrho(A_{B_{1}})$, the resolvent $R_{z}^{\Bi}$ in \eqref{resolvent} has the representation \eqref{res-new} and  
\be\label{alt-res}
R^{\Bi}_{z}=R_{z}^{B_{1}}+G_{z}^{B_{1}}\widehat \Lambda^{\Bi}_{z}G_{\bar z}^{B_{1}*}\,,\qquad
z\in \varrho(A_{\Bi})\cap\varrho(A_{B_{1}})\,,
\ee
where $R_{z}^{B_{1}}$, $G_{z}^{B_{1}}$ and $\widehat \Lambda^{\Bi}_{z}$ are defined in \eqref{res1.1}, \eqref{GB1} and \eqref{LB1}.
\end{theorem}
\begin{remark}Let us notice that the resolvent formula \eqref{alt-res} is of the same kind of the one in \eqref{res2.1}, whenever one replaces $A$ with $A_{B_{1}}$. 
\end{remark}
Let us now introduce the map
$$
\widehat \rho_{\Bi}:\dom(A_{\Bi})\to\h_{2}^{*}\,,\qquad \widehat \rho_{\Bi}(R^{\Bi}_{z}u):=\widehat\Lambda_{z}^{\Bi}G^{B_{1}*}_{\bar z}u\,.
$$
By the definition of $\rho_{\Bi}$ in \eqref{rho_def} and by  \eqref{GB1*}, \eqref{LB-new}, one obtains the relation
$$
\rho_{\Bi}u=\pi_{1}^{*}B_{1}\tau_{1}u\oplus\widehat \rho_{\Bi}u\,.
$$
Then, by using the same kind of arguments as in the proofs of Lemma \ref{abc} and Theorem \ref{Th-add}, one gets the following
\begin{theorem}\label{alt-abc} Let $A_{\Bi}$ be the self-adjoint operator in Theorem \ref{Th-alt-res}. Then, for any $z\in\varrho(A_{\Bi})\cap\varrho(A_{B_{1}})$, one has the representation
$$
\dom(A_{\Bi})=\{u\in\H:u_{z}:=u-G^{B_{1}}_{z}\widehat\rho_{\Bi}u\in \dom(A_{B_{1}})\}\,,
$$
$$
(-A_{\Bi}+z)u=(-A_{B_{1}}+z)u_{z}\,.
$$
Moreover, 
$$ 
A_{\Bi}=\overline {\!A}+\tau_{1}^{*}B_{1}\tau_{1}+\tau_{2}^{*}\widehat \rho_{\Bi}\,,
$$
and
$$
u\in \dom(A_{\Bi})\quad\Longrightarrow\quad B_{2}\tau_{2} u=B_{0}\widehat\rho_{\Bi} u\,.
$$
\end{theorem}

\section{The Limiting Absorption Principle and the Scattering Matrix}\label{Sec-LAP}
Now, given the measure space $(M,{\mathcal B},m)$, we suppose that $\H=L^{2}(M,{\mathcal B},m)\equiv L^{2}(M)$. 
Given a measurable $\varphi:M\to [1,+\infty)$, we define the weighted $L^{2}$-space 
\be\label{Lphi}
L_{\varphi}^{2}(M,{\mathcal B},m)\equiv L_{\varphi}^{2}(M):=\{\text{$u:M\to\CO$ measurable} : \varphi u\in L^{2}(M)\}\,.
\ee
By $\varphi\ge 1$,  
$$
L^{2}_{\varphi}(M)\hookrightarrow L^{2}(M)\hookrightarrow L_{\varphi^{-1}}^{2}(M)\simeq L_{\varphi}^{2}(M)^{*}\,.
$$
From now on $\langle\cdot,\cdot\rangle $ and  $\|\cdot\|$ denote the scalar product and the corresponding norm on  $L^{2}(M)$; $\langle\cdot,\cdot\rangle_{\varphi} $ and  $\|\cdot\|_{\varphi}$ denote the scalar product and the corresponding norm on  $L_{\varphi}^{2}(M)$. \par
Then we introduce the following hypotheses: \vskip5pt\par\noindent
(H1) $A_{B_{1}}$ is bounded from above and there exists a positive $\lambda_{1}\ge\sup\sigma(A_{B_{1}})$, such that 
$R^{B_{1}}_{z}\in \B(L^{2}_{\varphi}(M))$ for any $z\in \varrho(A_{B_{1}})$ such that $\text{Re}(z)> \lambda_{1}$; 
\vskip5pt\par\noindent
(H2) $A_{B_{1}}$ satisfies a Limiting Absorption Principle (LAP for short), i.e. there exists a (eventually empty) closed set with zero Lebesgue measure ${e}(A_{B_{1}})\subset \RE$ such that, for all $\lambda\in \RE\backslash{e}(A_{B_{1}})$, the limits
\be\label{LAP1}
R^{B_{1},\pm}_{\lambda}:=\lim_{\epsilon\searrow   0}R^{B_{1}}_{\lambda\pm i\epsilon}
\ee
exist in $\B(L_{\varphi}^{2}(M),L_{\varphi^{-1}}^{2}(M))$ and the maps $z\mapsto R^{B_{1},\pm}_{z}$, where $R^{B_{1},\pm}_{z}\equiv R^{B_{1}}_{z}$ whenever $z\in\varrho(A_{B_{1}})$, are continuous on $(\RE\backslash{e}(A_{B_{1}}))\cup\CO_{\pm}$ to $\B(L_{\varphi}^{2}(M),L_{\varphi^{-1}}^{2}(M))$;
\vskip5pt\par\noindent
(H3) for any compact set $K\subset \RE\backslash{e}(A_{B_{1}})$ there exists $c_{K}>0$ such that for any $\lambda\in K$ and for any $u\in L_{\varphi^{2}}^{2}(M)\cap\ker(R^{B_{1},+}_{\lambda}-R^{B_{1},-}_{\lambda})$ one has
\be\label{(H3)}
\|R^{B_{1},\pm}_{\lambda}u\|\le c_{K}\|u\|_{\varphi^{2}}\,;
\ee 
We split next hypothesis (H4) in two separate points:\vskip5pt\par\noindent
(H4.1) $A_{\Bi}$ is bounded from above;
\vskip5pt\par\noindent
(H4.2) the embedding $\h_{2}\hookrightarrow\b_{2}$ is compact and there exists a positive $\lambda_{2}>\sup\sigma(A_{B_{1}})$, such that  $G_{z}^{B_{1}}\in \B(\h^{*}_{2},L_{\varphi^{2+\eta}}^{2}(M))$ for some $\eta>0$ and for any $z\in\varrho(A_{B_{1}})$ such that $\text{Re}(z)>\lambda_{2}$.
\vskip5pt\par\noindent
Then, $A_{\Bi}$ satisfies a Limiting Absorption Principle as well:
\begin{theorem}\label{LAP} Suppose hypotheses (H1)-(H4) hold. Then the limits
\be\label{LAP2}
R^{\Bi,\pm}_{\lambda}:=\lim_{\epsilon\searrow   0}R^{\Bi}_{\lambda\pm i\epsilon}
\ee
exist in $\B(L_{\varphi}^{2}(M),L_{\varphi^{-1}}^{2}(M))$ for all $\lambda\in\RE\backslash{e}(A_{\Bi})$, where ${e}(A_{\Bi}):={e}(A_{B_{1}})\cup\sigma_{p}(A_{\Bi})$, and 
${e}(A_{\Bi})\backslash{e}(A_{B_{1}})$ is a (possibly empty) discrete set in $\RE\backslash {e}(A_{B_{1}})$; the maps $z\mapsto R^{\Bi,\pm}_{z}$, where $R^{\Bi,\pm}_{z}\equiv R^{\Bi}_{z}$ whenever $z\in\varrho(A_{\Bi})$, are continuous on $(\RE\backslash {e}(A_{\Bi}))\cup\CO_{\pm}$ to $\B(L_{\varphi}^{2}(M),L_{\varphi^{-1}}^{2}(M))$. Moreover 
\be\label{ess}
\sigma_{ess}(A_{\Bi})=\sigma_{ess}(A_{B_{1}})\,.
\ee
\end{theorem}
\begin{proof}  We use \cite[Theorem 3.1]{JMPA} (which builds on \cite{Reng}). By (H1), \eqref{alt-res} and (H4.2), $R^{B_{1}}_{z}$ and $R^{\Bi}_{z}$ are in $\B(L^{2}_{\varphi}(M))$ and  $z\mapsto R^{B_{1}}_{z}$ and $z\mapsto R^{\Bi}_{z}$ are continuous since pseudo-resolvents in $\B(L^{2}_{\varphi}(M))$; $A_{\Bi}$ is bounded from above by (H4.1). Therefore hypothesis (H1) in \cite{JMPA} holds true. Our hypotheses (H2) and (H3) coincides with the same ones in \cite{JMPA}. By (H4.2), the embedding $\b^{*}_{2}\hookrightarrow\h^{*}_{2}$ is compact. From $\widehat\Lambda^{\Bi}_{z}\in\B(\b_{2},\b_{2}^{*})$ and \eqref{alt-res}, follows that 
$R^{\Bi}_{z}-R^{B_{1}}_{z}\in{\mathfrak S}_{\infty}(L^{2}(M),L_{\varphi^{2+\gamma}}^{2}(M))$. Therefore hypothesis (H4) in \cite{JMPA} holds and the statement is a consequence of \cite[Theorem 3.1]{JMPA}. Finally, \eqref{ess} is an immediate consequence of Weyl's Theorem.
\end{proof}
Let us now assume that \vskip8pt\noindent
(H5) the limits
\be\label{limG}
G_{\lambda}^{B_{1},\pm}:=\lim_{\epsilon\searrow   0}G^{B_{1}}_{\lambda\pm i\epsilon}
\ee
exist in $\B(\h_{2} ^{*},L^{2}_{\varphi^{-1}}(M))$ for any $\lambda\in\RE\backslash{e}(A_{B_{1}})$ and the maps $z\mapsto G_{z}^{B_{1},\pm}$, where $G_{z}^{B_{1},\pm}\equiv G^{B_{1}}_{z}$ whenever $z\in\varrho(A_{B_{1}})$, are continuous on $(\RE\backslash{e}(A_{B_{1}}))\cup\CO_{\pm}$ to $\B(\h_{2} ^{*},L^{2}_{\varphi^{-1}}(M))$; moreover, the linear operators $G_{z}^{B_{1},\pm}$ are injective. 
\vskip8pt\noindent
Then, by \cite[Lemma 3.6]{JMPA}, one gets the following:
\begin{lemma}\label{bound}
Assume that (H1)-(H5) hold. Then, for any open and bounded $I$ s.t. $\overline{I}\subset \RE\backslash{e}(A_{\Bi})$, one has 
\be\label{supLambda}
\sup_{(\lambda,\epsilon)\in I\times (0,1)}{\|}\widehat\Lambda^{\Bi}_{\lambda\pm i\epsilon}{\|}_{\h_{2} ,\h_{2}^{*}}<+\infty\,.
\ee
Moreover, for any $\lambda\in \RE\backslash{e}(A_{\Bi})$,  the limits
\be\label{limLambda}
\widehat\Lambda^{\Bi,\pm}_{\lambda}:=\lim_{\epsilon\searrow   0}\widehat\Lambda^{\Bi}_{\lambda\pm i\epsilon}
\ee
exist in $\B(\h_{2} ,\h_{2} ^{*})$ and
\be\label{limRes}
R^{\Bi,\pm}_{\lambda}=R^{B_{1},\pm}_{\lambda}+G^{B_{1},\pm}_{\lambda}\widehat\Lambda^{\Bi,\pm}_{\lambda}(G^{B_{1},\mp}_{\lambda})^{*}\,.
\ee
\end{lemma}
By the same reasoning as at the end of \cite[proof of Theorem 5.1]{JMPA}, one can improve the result regarding \eqref{limLambda}:  
\begin{corollary}\label{limbb} Suppose hypotheses (H1)-(H5) hold. Then the limits \eqref{limLambda} exist in $\B(\b_{2},\b_{2}^{*})$.
\end{corollary}
Before stating the next results, we recall the following: 
\begin{definition} Given two self-adjoint operators $A_{1}$ and $A_{2}$ in the Hilbert space $\H$, we say that completeness holds for the scattering couple $(A_{1},A_{2})$ whenever  the strong limits 
$$
W_{\pm}(A_ {1},A_{2})
:=\text{s-}\lim_{t\to\pm\infty}e^{itA_{1}}e^{-itA_{2}}P_{2}^{ac}
\,,\qquad
W_{\pm}(A_{2},A_{1})
:=\text{s-}\lim_{t\to\pm\infty}e^{itA_{2}}e^{-itA_{1}}P_{1}^{ac}\,,
$$ 
exist everywhere in $\H$ and 
$$\ran(W_{\pm}(A_ {1},A_{2}))=\H_{1}^{ac }\,,\qquad\ran(W_{\pm}(A_{2},A_{1}))=\H_{2}^{ac }\,,
$$
$$ W_{\pm}(A_ {1},A_{1})^{*}=W_{\pm}(A_{2},A_ {1})\,,
$$ 
where $P_{k}^{ac}$ denotes the orthogonal projector onto the absolutely continuous subspace $\H_{k}^{ac}$ of $A_{k}$. Furthermore, we say the asymptotic completeness holds for the scattering couple $(A_{1},A_{2})$ whenever, beside completeness, one has 
$$
\H_{1}^{ac }=(\H_{1}^{pp})^{\perp}\,,\qquad \H_{2}^{ac }=(\H_{2}^{pp})^{\perp}\,,
$$ 
where $\H_{k}^{pp}$ denotes the pure point subspace of $A_{k}$; equivalently, whenever $
\sigma_{sc}(A_{1})=\sigma_{sc}(A_{2})=\varnothing $.
\end{definition}  
Our next hypothesis is \vskip8pt\noindent
(H6) completeness hold for the scattering  couple $(A_{B_{1}},A)$. 
\vskip8pt\noindent
\begin{theorem}\label{AC} Suppose that (H1)-(H6) hold. Then completeness holds for the couple $(A_{\Bi},A)$. \par 
If furthermore $\sigma_{sc}(A)=\varnothing $ and 
\vskip3pt\noindent
$i)$ the set of accumulation points of $e(A_{B_{1}})\cap \mathring{\sigma}_{ess}(A_{B_{1}})$ is discrete in $\mathring{\sigma}_{ess}(A_{B_{1}})$,
\newline
$ii)$ 
the boundary of ${\sigma}_{ess}(A_{B_{1}})$ is countable, 
\vskip3pt\noindent
then asymptotic completeness holds for the couple $(A_{\Bi},A)$.
\end{theorem}
\begin{proof} By \eqref{alt-res} and by the same proof as in Lemma \ref{le-green}, one gets 
\begin{equation}\label{alt-gf}
\langle u,A_{\Bi}v\rangle_{L^{2}(M)}-\langle A_{B_{1}}
u,v\rangle_{L^{2}(M)}=
\langle\tau_{2} u,\widehat \rho_{\Bi} v\rangle_{\h_{2},\h_{2}^{*}}\,,\quad u\in \dom(A_{B_{1}}),\,v\in
\dom(A_{\Bi})\,,
\end{equation}
where
\begin{equation*}
\widehat\rho_{\Bi}:\dom(A_{\Bi})\rightarrow\h_{2}^{*}\,,\quad
\widehat\rho_{\Bi}(R_{z}^{\Bi}u):=\widehat\Lambda_{z}^{\Bi}G^{B_{1}*}_{\bar{z}}u\,,\qquad u\in \H\,, \quad z\in\varrho(A_{\Bi})\cap\varrho(A_{B_{1}})\,.
\end{equation*}
Then, by hypotheses (H1)-(H5) and by \cite[Theorems 2.8 and 3.8]{JMPA} (compare \eqref{alt-gf} and Lemma \ref{bound} here with (2.19) and Lemma 3.6 there and notice that hypothesis (H6) there is included in our hypothesis (H4)) one gets the completeness for the couple $(A_{\Bi},A_{B_{1}})$. By (H6) and the chain rule for the wave operators (see \cite[Theorem 3.4, Chapter X]{Kato}), one then gets completeness for the scattering couple $(A_{\Bi},A)$. \par 
To conclude the proof it remains to show that 
$\sigma_{sc}(A_{\Bi})=\varnothing$. Let $\H_{\Bi}^{pp}$ denote the pure point subspace of $A_{\Bi}$ and,  given $u\in (\H_{\Bi}^{pp})^{\perp}$, we denote by $\mu_{u}^{\Bi}$ be the corresponding spectral measure. By our choice of $u$, one gets $\text{supp}(\mu_{u}^{\Bi})\subseteq\sigma_{cont}(A_{\Bi})\subseteq\sigma_{ess}(A_{\Bi})=\sigma_{ess}(A_{B_{1}})$. Let us define
\begin{align*}
&e_{ess}(A_{B_{1}}):=e(A_{B_{1}})\cap \mathring{\sigma}_{ess}(A_{B_{1}})\,,
\\
&e_{ess}(A_{\Bi}):=(e(A_{B_{1}})\cup\sigma_{p}(A_{\Bi}))\cap \mathring{\sigma}_{ess}(A_{B_{1}})\,,
\end{align*}
and denote by $e'_{ess}(A_{B_{1}})$ the set of accumulation points of $e_{ess}(A_{B_{1}})$. Since an open set minus a discrete subset is still open, one has
$$
\mathring{\sigma}_{ess}(A_{B_{1}})\backslash e'_{ess}(A_{B_{1}})=\bigcup_{n\ge 1} I_{n}\,,
$$
where the $I_{n}$'s are open intervals. Moreover, since $I_{n}\cap e'_{ess}(A_{B_{1}})=\varnothing$, then $I_{n}\cap e_{ess}(A_{B_{1}})$ is discrete in $I_{n}$ and so $I_{n}\backslash (I_{n}\cap e_{ess}(A_{B_{1}}))$ is open. This yields
$$
I_{n}\backslash (I_{n}\cap e_{ess}(A_{B_{1}}))=\bigcup_{m\ge 1} I_{n,m}\,,
$$ 
where the $I_{n,m}$'s are open intervals. By Theorem \ref{LAP}, the set of accumulation points of $e(A_{\Bi})\backslash e(A_{B_{1}})$ is contained in $e(A_{B_{1}})$; therefore $I_{n,m}\cap (e(A_{\Bi})\backslash e(A_{B_{1}}))$ is discrete in $I_{n,m}$. As before, \par\noindent ${I_{n,m}\backslash (I_{n,m}\cap (e(A_{\Bi})\backslash e(A_{B_{1}})))}$ is open and we get
$$
I_{n,m}\backslash (I_{n,m}\cap (e(A_{\Bi})\backslash e(A_{B_{1}})))=\bigcup_{\ell\ge 1} I_{n,m,\ell}\,,
$$ 
where the $I_{n,m,\ell}$'s are open intervals. Hence,
\begin{align*}
&\mathring{\sigma}_{ess}(A_{B_{1}})\backslash e_{ess}(A_{\Bi})=
\mathring{\sigma}_{ess}(A_{B_{1}})\backslash (e_{ess}(A_{B_{1}})\cup 
e_{ess}(A_{\Bi})\backslash e_{ess}(A_{B_{1}}))\\
=&\Big(\bigcup_{n\ge 1} I_{n}\cup e'_{ess}(A_{B_{1}})\Big)\backslash
\big(e_{ess}(A_{B_{1}})\cup 
e_{ess}(A_{\Bi})\backslash e_{ess}(A_{B_{1}})\big)\\
=&\Big(\Big(\bigcup_{n\ge 1} I_{n}\backslash e_{ess}(A_{B_{1}})\Big)\cup
(e'_{ess}(A_{B_{1}})\backslash e_{ess}(A_{B_{1}}))\Big)\backslash(e(A_{\Bi})\backslash e(A_{B_{1}}))
\\
=&\Big(\bigcup_{n,m\ge 1} I_{n,m}\cup
(e'_{ess}(A_{B_{1}})\backslash e_{ess}(A_{B_{1}}))\Big)\backslash(e(A_{\Bi})\backslash e(A_{B_{1}}))
\\
=&\Big(\bigcup_{n,m\ge 1} I_{n,m}\backslash(e(A_{\Bi})\backslash \big(e(A_{B_{1}})\Big)\cup
(e'_{ess}(A_{B_{1}})\backslash e_{ess}(A_{\Bi})\big)
\\
=&\Big(\bigcup_{n,m,\ell\ge 1} I_{n,m,\ell}\Big)\cup
(e'_{ess}(A_{B_{1}})\backslash e_{ess}(A_{\Bi})\big)\,.
\end{align*}
This gives
\begin{align*}
\text{supp}(\mu_{u}^{\Bi})\subseteq &\sigma_{ess}(A_{B_{1}})
=\big(\mathring{\sigma}_{ess}(A_{B_{1}})\backslash e_{ess}(A_{\Bi})\big)\cup\partial\sigma_{ess}(A_{B_{1}})\cup e_{ess}(A_{\Bi})\\
=&\Big(\bigcup_{n,m,\ell\ge 1} I_{n,m,\ell}\Big)\cup\partial\sigma_{ess}(A_{B_{1}})\cup e_{ess}(A_{\Bi})\cup
e'_{ess}(A_{B_{1}})
\end{align*}
By standard arguments  (see e.g.  \cite[proof of Thm. 6.1]{Agmon} or  \cite[top of page 178]{ReSi IV}) applied to any of the open intervals $I_{n,m,\ell}$, one gets the absolute continuity of the spectral function $\lambda\mapsto \mu_{u}^{\Bi}(-\infty,\lambda]$ on any compact interval in $I_{n,m,\ell}$; hence 
\begin{align*}
\text{supp}((\mu_{u}^{\Bi})_{sing})\subseteq &
\partial\sigma_{ess}(A_{B_{1}})\cup e_{ess}(A_{\Bi})\cup
e'_{ess}(A_{B_{1}})\\
=&\partial\sigma_{ess}(A_{B_{1}})\cup e_{ess}(A_{B_{1}})\cup (e_{ess}(A_{\Bi})\backslash e_{ess}(A_{B_{1}}))\cup e'_{ess}(A_{B_{1}})\,.
\end{align*}
By Theorem \ref{LAP}, $e(A_{\Bi})\backslash e(A_{B_{1}})$ is discrete (hence countable) in $\RE\backslash e(A_{B_{1}})$; by (i) and (ii), the sets  $e'_{ess}(A_{B_{1}})$, $e_{ess}(A_{B_{1}})$ and 
$\partial\sigma_{ess}(A_{B_{1}})$ are countable. Henceforth, the support of the singular continuous component of $\mu_{u}^{\Bi}$ is contained in a countable set. This implies $\text{supp}((\mu_{u}^{\Bi})_{sing})=\varnothing$. Therefore, $u$ has a null projection onto $\H_{\Bi}^{sc}$,  the singular continuous subspace of $A_{\Bi}$. This gives $(\H_{\Bi}^{pp})^{\perp}=\H_{\Bi}^{ac}$, where $\H_{\Bi}^{ac}$ denote the absolutely continuous subspace of $A_{\Bi}$.
\end{proof}
\begin{remark} Since, by Corollary \ref{cor1}, $A_{\Bi}=A_{B_{1}}$ whenever $\Bi=(1,B_{1},0)$, 
Theorem \ref{AC} also provides the asymptotic completeness of the couple $(A_{B_{1}},A)$.
\end{remark}
\subsection{A representation formula for the scattering matrix}
According to Theorem \ref{AC}, under the assumptions there stated, the scattering operator 
$$
{\mathsf S}_{\Bi}:=W_{+}(A_{\Bi},A)^{*}W_{-}(A_{\Bi},A)
$$
is a well defined unitary map.  Let  
\be\label{F0}
F:L^{2}(M)_{ac}\to\int^{\oplus}_{\sigma_{ac}(A)}(L^{2}(M)_{ac})_{\lambda}\,d\eta(\lambda)
\ee
be  a unitary map which diagonalizes the absolutely continuous component of $A$, i.e., a direct integral representation of $L^{2}(M)_{ac}$, the absolutely continuous subspace relative to $A$, with respect to the spectral measure of the absolutely continuous component of $A$  (see e.g. \cite[Section 4.5.1]{BW}). We define the scattering matrix  $${\Sc}^{\Bi}_{\lambda}:(L^{2}(M)_{ac})_{\lambda}\to (L^{2}(M)_{ac})_{\lambda}$$ by the relation (see e.g. \cite[Section 9.6.2]{BW})
$$
F{\mathsf S}_{\Bi}F^{*}u_{\lambda}={\Sc}^{\Bi}_{\lambda}u_{\lambda}
\,.
$$
Now, following the same scheme as in \cite{JMPA}, which uses the Birman-Kato invariance principle and the Birman-Yafaev general scheme in stationary scattering theory, we provide an explicit relation between  ${\Sc}^{\Bi}_{\lambda}$ and $\Lambda^{\!\Bi,+}_{\lambda}:=\lim_{\epsilon\searrow 0}\Lambda^{\!\Bi}_{\lambda+i\epsilon}$. \par 
Given $\mu\in \varrho(A)\cap\varrho(A_{\Bi})$, we consider the scattering couple $(R^{\Bi}_{\mu}, R_{\mu})$ and the strong limits 
$$
W_{\pm}(R^{\Bi}_{\mu},R_{\mu})
:=\text{s-}\lim_{t\to\pm\infty}e^{itR^{\Bi}_{\mu}}e^{-itR_{\mu}}P^{\mu}_{ac}\,,
$$ 
where $P^{\mu}_{ac}$ is the orthogonal projector onto the absolutely continuous subspace of $R_{\mu}$; we prove below that such limits exist everywhere in $L^{2}(M)$. Let ${\mathsf S}^{\mu}_{\Bi}$ the corresponding scattering operator  
$$
{\mathsf S}_{\Bi}^{\mu}:=W_{+}(R^{\Bi}_{\mu},R_{\mu})^{*}W_{-}(R^{\Bi}_{\mu},R_{\mu})\,.
$$
Using the unitary operator $F_{\mu}$ which diagonalizes the absolutely continuous component of $R_{\mu}$, i.e. $(F_{\mu}u)_{\lambda}:=\frac1\lambda(Fu)_{\mu-\frac1\lambda}$, $\lambda\not=0$ such that $\mu-\frac1\lambda\in \sigma_{ac}(A)$, one defines the scattering  matrix  $${\Sc}^{\Bi,\mu}_{\lambda}:(L^{2}(M)_{ac})_{\mu-\frac1\lambda}\to (L^{2}(M)_{ac})_{\mu-\frac1\lambda}$$  corresponding to the scattering operator ${\mathsf S}^{\mu}_{\Bi}$ by the relation
$$
F_{\mu}{\mathsf S}^{\mu}_{\Bi}F_{\mu}^{*}u^{\mu}_{\lambda}={\Sc}^{\Bi,\mu}_{\lambda}u^{\mu}_{\lambda}
\,.
$$ 
We introduce a further hypothesis (H7), which we split in four separate points:\vskip5pt\par\noindent
(H7.1) $A$ is bounded from above and satisfies a Limiting Absorption Principle: there exists a (eventually empty) closed set ${e}(A)\subset\RE$ of zero Lebesgue measure such that for all $\lambda\in \RE\backslash{e}(A)$ the limits
\be\label{LAP0}
R^{\pm}_{\lambda}:=\lim_{\epsilon\searrow   0}R_{\lambda\pm i\epsilon}
\ee
exist in $\B(L_{\varphi}^{2}(M),L_{\varphi^{-1}}^{2}(M))$; 
\vskip5pt\par\noindent
(H7.2) $G^{1}_{z}\in\B(\h_{1}^{*},L^{2}_{\varphi}(M))$ for any $z\in\varrho(A)$ and the limits
\be\label{limG1}
G_{\lambda}^{{1},\pm}:=\lim_{\epsilon\searrow   0}G^{{1}}_{\lambda\pm i\epsilon}
\ee
exist in $\B(\h_{1} ^{*},L^{2}_{\varphi^{-1}}(M))$ for any $\lambda\in\RE\backslash{e}(A)$; 
\vskip5pt\par\noindent
(H7.3) the limits
\be\label{limLB1}
\Lambda^{\!B_{1},\pm}_{\lambda}:=\lim_{\epsilon\searrow   0}\Lambda^{\!B_{1},\pm}_{\lambda\pm i\epsilon}
\ee
exist in $\B(\h_{1},\h_{1}^{*})$ for any $\lambda\in\RE\backslash{e}(A_{B_{1}})$;
\vskip5pt\par\noindent
(H7.4) the limits
\be\label{limtG1}
\tau_{2} G^{1,\pm}_{\lambda}:=\lim_{\epsilon\searrow   0}\tau_{2} G^{1}_{\lambda\pm i\epsilon}
\ee
exist in $\B(\b^{*}_{1},\b_{2})$ for any $\lambda\in\RE\backslash{e}(A_{B_{1}})$.
\begin{remark}\label{rem3.6} By $\tau_{2}G^{1}_{z}=\tau_{2} (\tau_{1}R_{\bar z})^{*}=(\tau_{1} (\tau_{2}R_{ z})^{*})^{*}=(\tau_{1}G^{2}_{\bar z})^{*}$, hypothesis (H7.4) entails the existence in $\B(\b_{2},\b_{1}^{*})$, for any $\lambda\in\RE\backslash{e}(A_{B_{1}})$,  of the limits  
\be
\tau_{1} G^{2,\pm}_{\lambda}:=\lim_{\epsilon\searrow   0}\tau_{1} G^{2}_{\lambda\pm i\epsilon}\,.
\ee
\end{remark}
\begin{remark} Whenever one strengthens hypotheses (H7.2) as in (H5), then, by the same kind of proof that leads to the existence of the limit \eqref{limLambda} (see \cite[Lemma 3.6]{JMPA}), one gets the existence of the limits requested in hypotheses (H7.3).
\end{remark}
\begin{lemma}\label{rmH7} Suppose that (H1)-(H5) and (H7) hold. Then
\be\label{rmH7-1}
R_{\lambda}^{B_{1},\pm}=R^{\pm}_{\lambda}+G_{\lambda}^{{1},\pm}\Lambda^{\!B_{1},\pm}_{\lambda}(G_{\lambda}^{{1},\mp})^{*}\,;
\ee
\be\label{Gz2}
G^{2}_{z}\in \B(\h_{2} ^{*},L^{2}_{\varphi}(M))\,,\qquad z\in\varrho(A_{B_{1}})\cap\varrho(A)\,,
\ee 
the limits
\be\label{limG2}
G_{\lambda}^{{2},\pm}:=\lim_{\epsilon\searrow   0}G^{{2}}_{\lambda\pm i\epsilon}
\ee
exist in $\B(\h_{2} ^{*},L^{2}_{\varphi^{-1}}(M))$ for any $\lambda\in\RE\backslash{e}(A_{B_{1}})$ and
\be\label{limGB1}
G_{\lambda}^{B_{1},\pm}=G_{\lambda}^{2,\pm}+G_{\lambda}^{1,\pm}\Lambda^{\!B_{1},\pm}_{\lambda}\tau_{1}G_{\lambda}^{2,\pm}\,;
\ee
the limits 
$$
\Lambda_{\lambda}^{\!\Bi,\pm}:=
\lim_{\epsilon\searrow   0}\Lambda_{\lambda\pm i\epsilon}^{\!\Bi}
$$
exist in $\B(\h_{1}\oplus \b_{2},\h_{1}^{*}\oplus \b_{2}^{*})$ and
\begin{align}
\Lambda_{\lambda}^{\!\Bi,\pm}=&\begin{bmatrix}\Lambda^{\!B_{1},\pm}_{\lambda}+
\Lambda^{\!B_{1},\pm}_{\lambda}\tau_{1}G^{2,\pm}_{\lambda}\widehat \Lambda^{\Bi,\pm}_{\lambda}\tau_{2}G^{1,\pm}_{\lambda}\Lambda^{\!B_{1},\pm}_{\lambda}&
\Lambda^{\!B_{1},\pm}_{\lambda}\tau_{1}G^{2,\pm}_{\lambda}\widehat \Lambda^{\Bi,\pm}_{\lambda}\\
\widehat \Lambda^{\Bi,\pm}_{\lambda}\tau_{2}G^{1,\pm}_{\lambda}\Lambda_{\lambda}^{\!B_{1},\pm}&
\widehat \Lambda^{\Bi,\pm}_{\lambda}
\end{bmatrix}\label{LBpm}\\
=&\left(1+\begin{bmatrix}\Lambda^{\!B_{1},\pm}_{\lambda}&0\\0&\widehat \Lambda^{\Bi,\pm}_{\lambda}
\end{bmatrix}\begin{bmatrix}\tau_{1}G^{2,\pm}_{\lambda}\widehat \Lambda^{\Bi,\pm}_{\lambda}\tau_{2}G^{1,\pm}_{\lambda}&
\tau_{1}G^{2,\pm}_{\lambda}\\
\tau_{2}G^{1,\pm}_{\lambda}&0
\end{bmatrix}\,\right)\begin{bmatrix}\Lambda^{\!B_{1},\pm}_{\lambda}&0\\0&\widehat \Lambda^{\Bi,\pm}_{\lambda}
\end{bmatrix}\,.\label{LBpm2}
\end{align}
\end{lemma}
\begin{proof} The relation \eqref{rmH7-1} is an immediate consequence of \eqref{res1.1} and (H7.1)-(H7.3). By \eqref{GB1}, 
$$
G_{z}^{2}=G^{B_{1}}_{z}-G_{z}^{1}\Lambda^{\!B_{1}}_{z}\tau_{1}G_{z}^{2}
$$
and \eqref{Gz2} follows from (H4.2) and (H7.2). Then, Remark \ref{rem3.6}, (H.5) and (H7.3) entail 
\eqref{limG2} and \eqref{limGB1}. Finally, \eqref{LBpm} and \eqref{LBpm2} are consequence of  \eqref{LB-new}, \eqref{LB-new2}, Corollary \ref{limbb}, (H7.3), Remark \ref{rem3.6} and (H7.4). 
\end{proof}
Before stating the next results, let us notice the relations 
\be\label{RR}
\left(-R_{\mu}+z\right)^{-1}=\frac1z\,\left(1+\frac1{z}\,R_{\mu-\frac1z}\right)\,,
\quad \left(-R_{\mu}^{\Bi}+z\right)^{-1}=\frac1z\,\left(1+\frac1{z}\,R^{\Bi}_{\mu-\frac1z}\right)\,,
\ee
Therefore, by (H7.1) and Theorem \ref{LAP}, the limits 
\be\label{RRlim1}
\left(-R_{\mu}+(\lambda\pm i0)\right)^{-1}:=\lim_{\epsilon\searrow   0}\left(-R_{\mu}+(\lambda\pm i\epsilon)\right)^{-1}\,,\quad \lambda\not=0\,,\ \mu-\frac1\lambda\in\RE\backslash{e}(A)\,,
\ee 
\be\label{RRlim2}
\left(-R_{\mu}^{\Bi}+(\lambda\pm i0)\right)^{-1}:=
\lim_{\epsilon\searrow   0}\left(-R_{\mu}^{\Bi}+(\lambda\pm i\epsilon)\right)^{-1}\,,
\quad \lambda\not=0\,,\ \mu-\frac1\lambda\in\RE\backslash{e}(A_{\Bi})\,,
\ee
exist in $\B(L^{2}_{\varphi}(M),L^{2}_{\varphi^{-1}}(M))$.
\begin{theorem}\label{BK}
Suppose that hypotheses (H1)-(H7) hold. Then the strong limits 
\be\label{WR}
W_{\pm}(R^{\Bi}_{\mu},R_{\mu})
:=\text{s-}\lim_{t\to\pm\infty}e^{itR^{\Bi}_{\mu}}e^{-itR_{\mu}}P^{\mu}_{ac}
\ee
exist everywhere in $L^{2}(M)$. 
Moreover, for any  $\lambda\not=0$ such that $\mu-\frac1\lambda\in \sigma_{ac}(A)\cap(\RE\backslash{e}(A_{\Bi}))$, one has
\be\label{S1}
{\Sc}^{\Bi,\mu}_{\lambda}=1-2\pi i\,\L^{\mu}_{\lambda}
\Lambda^{\!\Bi}_{\mu}\big(1+G^{*}_{\mu}\big(-R_{\mu}^{\Bi}+(\lambda+ i0)\big)^{-1}G_{\mu}\Lambda^{\!\Bi}_{\mu}\big)
(\L^{\mu}_{\lambda})^{*}\,,
\ee
where
\be\label{SR}
\L^{\mu}_\lambda: \h_{1}^{*}\oplus\h_{2}^{*}\to(L^{2}(M)_{ac})_{\mu-\frac1\lambda}\,,\quad 
\L^{\mu}_{\lambda}(\phi_{1}\oplus\phi_{2}):=\frac1\lambda(FG_{\mu}(\phi_{1}\oplus\phi_{2}))_{\mu-\frac1\lambda}\,.
\ee
\end{theorem} 
\begin{proof}
By \eqref{resolvent}, one has $R_{\mu}^{\Bi}-R_{\mu}=G_{\mu}\Lambda^{\! B}_{\mu}G^{*}_{\mu}$ and we can use \cite[Theorem 4', page 178]{Y} (notice that the maps there denoted by $G$ and $V$ corresponds to our $G_{\mu}^{*}$ and $\Lambda^{\! B}_{\mu}$ respectively). Let us check that the hypotheses there required are satisfied. Since $G^{*}_{\mu}\in \B(L^{2}(M),\h_{1}\oplus\h_{2})$, the operator $G_{\mu}$ is $|R_{\mu}|^{1/2}$-bounded. By (H7.2) and \eqref{Gz2}, one has $G_{z}\in\B(\h_{1}^{*}\oplus\h_{2}^{*}, L^{2}_{\varphi}(M))$ for any $z\in\varrho(A_{B_{1}})\cap\varrho(A)\supset[\lambda_{1},+\infty)\ni\mu$. Therefore, by \eqref{RRlim1}, \eqref{RRlim2}, (H7.1), Theorem \ref{LAP} and (H4), the limits 
$$
\lim_{\epsilon\searrow   0}\, G^{*}_{\mu}(-R _{\mu}+(\lambda\pm i\epsilon))^{-1}\,,
$$
$$
\lim_{\epsilon\searrow   0}\, G^{*}_{\mu}(-R^{\Bi}_{\mu}+(\lambda\pm i\epsilon))^{-1}\,,
$$
$$
\lim_{\epsilon\searrow   0}\, G^{*}_{\mu}(-R^{\Bi}_{\mu}+(\lambda\pm i\epsilon))^{-1}G_{\mu}
$$
exist. Therefore, to get the thesis we need to check the validity of the remaining hypothesis in 
\cite[Theorem 4', page 178]{Y}: $G^{*}_{\mu}$ is weakly-$R _{\mu}$ smooth, i.e., by \cite[Lemma 2, page 154]{Y}, 
\be\label{in1.1}
\sup_{0<\epsilon<1}\epsilon\,{\|}G^{*}_{\mu} (-R _{\mu}+(\lambda\pm i\epsilon))^{-1}{\|}_{L^{2}(M),\h_{1}\oplus\h_{2}}^{2}\le c_{\lambda}<+\infty\,,\quad\text{a.e. $\lambda$}\,.
\ee
By \eqref{RR}, this is consequence of
\be\label{in2.2}
\sup_{0<\epsilon<1}\epsilon\,{\|}G^{*}_{\mu}R_{\mu-\frac1\lambda\pm i\epsilon}{\|}_{L^{2}(M),\h_{1}\oplus\h_{2}}^{2}\le C_{\lambda}<+\infty\,,\quad\text{a.e. $\lambda$}\,.
\ee
By \cite[(3.16)]{JMPA},
\begin{align*}
&\epsilon\,{\|}G_{\lambda\pm i\epsilon}{\|}_{\h^{*}_{1}\oplus\h^{*}_{2},L^{2}(M)}^{2}\\
\le& 
\frac12\,(|\mu-\lambda|+\epsilon)\, {\|}G_{\mu}{\|}_{\h_{1}^{*}\oplus\h_{2}^{*},L_{\varphi}^{2}(M)} \left(\|G_{\lambda- i\epsilon}\|_{\h_{1}^{*}\oplus\h_{2}^{*},L_{\varphi^{-1}}^{2}(M)} +
\|G_{\lambda+ i\epsilon}\|_{\h_{1}^{*}\oplus\h_{2}^{*},L_{\varphi^{-1}}^{2}(M)} \right)\,.
\end{align*}
Then, \eqref{in2.2} follows from \eqref{limG1}, \eqref{limG2} and the equality
\begin{align*}
&{\|}G^{*}_{\mu} R _{z}{\|}_{L^{2}(M),\h_{1}\oplus\h_{2}}
={\|}\tau R_{\mu}R _{z}{\|}_{L^{2}(M),\h_{1}\oplus\h_{2}}= 
{\|}\tau R _{z}R _{\mu}{\|}_{L^{2}(M),\h_{1}\oplus\h_{2}}\\
=&
{\|}R _{\mu}(\tau R _{z})^{*}{\|}_{\h_{1}^{*}\oplus\h_{2}^{*},L^{2}(M)}\le
{\|}R _{\mu}{\|}_{L^{2}(M),L^{2}(M)} {\|}G_{\bar z}{\|}_{\h_{1}^{*}\oplus\h_{2}^{*},L^{2}(M)}\,.
\end{align*}
Thus, by \cite[Theorem 4', page 178]{Y}, 
the limits \eqref{WR} exist everywhere in $L^{2}(M)$ and the corresponding scattering matrix is given by \eqref{S1}, where $\L^{\mu}_{\lambda}\phi:=(F^{\mu}G_{\mu}\phi)_{\lambda}=\frac1\lambda(FG_{\mu}\phi)_{\mu-\frac1\lambda}$. 
\end{proof}
\begin{theorem}\label{S-matrix} Suppose that hypotheses (H1)-(H7) hold. Then the scattering matrix of the couple $(A_{\Bi},A)$ has the representation
\begin{equation}\label{S-M}
{\Sc}^{\Bi}_{\lambda}=1-2\pi i\L_{\lambda}\Lambda^{\!\Bi,+}_{\lambda}\L_{\lambda}^{*}\,,\quad \lambda\in\sigma_{ac}(A)\cap(\RE\backslash{e}(A_{\Bi}))\,,
\end{equation}
where  $\L_\lambda: \h_{1}^{*}\oplus\h_{2}^{*}\to(L^{2}(M)_{ac})_{\lambda}$ is the $\mu$-independent linear operator defined by  
\be\label{LLL}
\L_{\lambda}(\phi_{1}\oplus\phi_{2}):=(\mu-\lambda)(FG_{\mu}(\phi_{1}\oplus\phi_{2}))_{\lambda}
\ee
and $\Lambda^{\!\Bi,+}_{\lambda}$ is given in \eqref{LBpm}.
\end{theorem}
\begin{proof}
By Theorem \ref {AC}, Theorem \ref{BK} and by Birman-Kato invariance principle (see e.g. \cite[Section II.3.3]{BW}), one has
$$
W_{\pm}(A_{\Bi},A)=W_{\pm}(R^{\Bi}_{\mu},R_{\mu})
$$
and so
$$
{\mathsf S}_{\Bi}={\mathsf S}_{\Bi}^{\mu}\,.
$$
Thus, since $(F^{\mu}u)_{\lambda}=\frac1\lambda(Fu)_{\mu-\frac1\lambda}$,  one obtains (see also \cite[Equation (14), Section 6, Chapter 2]{Y})
\be\label{SS}
{\Sc}^{\Bi}_{\lambda}={\Sc}^{\Bi,\mu}_{(-\lambda+\mu)^{-1}}\,.
\ee
By \cite[Lemma 4.2]{JMPA}, for any $z\not=0$ such that $\mu-\frac1z\in \varrho(A_{\Bi})\cap\varrho(A)$, there holds
$$
\Lambda^{\!\Bi}_{\mu}\left(1+G^{*}_{\mu}\left(-R_{\mu}^{\Bi}+z\right)^{-1}G_{\mu}\Lambda^{\! B}_{\mu}\right)=\Lambda^{\!\Bi}_{\mu-\frac1z}\,.
$$
Hence, whenever $z=\lambda\pm i\epsilon$ and $\mu-\frac1{\lambda}\in\RE\backslash{e}(A_{\Bi})$, one gets, as $\epsilon\downarrow 0$, 
$$
\Lambda^{\!\Bi}_{\mu}\big(1+G^{*}_{\mu}\left(-R_{\mu}^{\Bi}+(\lambda\pm i0)\right)^{-1}G_{\mu}\Lambda^{\!\Bi}_{\mu}\big)=\Lambda^{\!\Bi,\pm}_{\mu-\frac1{\lambda}}\,.
$$
The proof is then concluded by Theorem \ref{BK}, by \eqref{SS} and by setting  
$\L_{\lambda}:=\L^{\mu}_{{(-\lambda+\mu)^{-1}}}$. The operator $\L_{\lambda}$ is $\mu$-independent by invariance principle (see the proof in \cite[Corollary 4.3]{JMPA} for an explicit check).
\end{proof}
\begin{remark}
By \eqref{LBpm}, 
$$
\Lambda^{\!\Bi,\pm}_{\lambda}=\begin{bmatrix}\Lambda^{\!B_{1},\pm}_{z}&0\\0&0\end{bmatrix}+\widetilde \Lambda^{\!\Bi,\pm}_{\lambda}\,,
$$
where
$$
\widetilde \Lambda^{\!\Bi,\pm}_{\lambda}:=
\begin{bmatrix}
\Lambda^{\!B_{1},\pm}_{\lambda}\tau_{1}G^{2,\pm}_{\lambda}\widehat \Lambda^{\Bi,\pm}_{\lambda}\tau_{2}G^{1,\pm}_{\lambda}\Lambda^{\!B_{1},\pm}_{\lambda}&
\Lambda^{\!B_{1},\pm}_{\lambda}\tau_{1}G^{2,\pm}_{\lambda}\widehat \Lambda^{\Bi,\pm}_{\lambda}\\
\widehat \Lambda^{\Bi,\pm}_{\lambda}\tau_{2}G^{1,\pm}_{\lambda}\Lambda_{\lambda}^{\!B_{1},\pm}&
\widehat \Lambda^{\Bi,\pm}_{\lambda}
\end{bmatrix}\,.
$$
Therefore, defining 
$$
\L^{1}_{\lambda}\phi_{1}:=\L_{\lambda}(\phi_{1}\oplus 0)\,,
$$ 
one gets
$$
{\Sc}^{\Bi}_{\lambda}={\Sc}^{B_{1}}_{\lambda}-2\pi i\L_{\lambda}\widetilde\Lambda^{\!\Bi,+}_{\lambda}\L_{\lambda}^{*}\,,
$$
where
\be\label{SB1}
{\Sc}^{B_{1}}_{\lambda}=1-2\pi i\L^{1}_{\lambda}\widetilde\Lambda^{\!B_{1},+}_{\lambda}(\L^{1}_{\lambda})^{*}
\ee
is the scattering matrix relative to the couple $(A_{B_{1}},A)$.  Moreover, in the case $B_{1}=0$, 
defining 
$$
\L^{2}_{\lambda}\phi_{2}:=\L_{\lambda}(0\oplus\phi_{2})\,,
$$ 
one gets the following representation formula for the scattering couple $(A_{B_{0},B_{2}},A)$ (compare with \cite[Corollary 4.3]{JMPA}):
$$
{\Sc}^{B_{0},B_{2}}_{\lambda}=1-2\pi i\L^{2}_{\lambda}\Lambda^{\!B_{0},B_{2},+}_{\lambda}(\L^{2}_{\lambda})^{*}\,.
$$
Let us further notice that, whenever $A$ is the free Laplacian in $L^{2}(\RE^{3})$ and $B_{1}$ corresponds to a perturbation by a regular potential as in Section 5 below, then \eqref{SB1} gives the usual formula for the scattering matrix for a short-range potential (see, e.g., \cite[Section 8]{Y-LNM}).
\end{remark}

\section{Kato-Rellich perturbations and their layers potentials}

\subsection{\label{Sec_V}Potential perturbations}

In this section we suppose that the real-valued potential $\v$ is of Kato-Rellich type, i.e.,  $\v\in L^{2}(\RE^{3})+L^{\infty}(\RE^{3})$, equivalently,
\begin{equation}
\v=\v_{2}+\v_{\infty} \,,\qquad \v_{2}\in L^{2}(  \mathbb{R}^{3}) \,,\qquad \v_{\infty}\in L^{\infty}
(\mathbb{R}^{3})  \,. \label{K-R}%
\end{equation}
We use the same symbol $\v$ to denote both the potential function and the corresponding multiplication operator $u\mapsto \v u$. 
\par
Given $\Omega\subset\RE^{3}$, open and bounded with a Lipschitz boundary $\Gamma$, we define $H^{s}(\RE^{3}\backslash\Gamma)\hookleftarrow H^{s}(\RE^{3})$ by  
$$
H^{s}(\RE^{3}\backslash\Gamma):=H^{s}(\Omega)\oplus H^{s}(\Omega_{\ex})\,,\qquad s\ge 0\,.
$$
We refer to \cite[Chapter 3]{McL} for the definition of the Sobolev spaces $H^{s}(\RE^{3})$, $H^{s}(\Omega)$ and $H^{s}(\Gamma)$. One has 
$$
H^{s}(\RE^{3}\backslash\Gamma)=H^{s}(\RE^{3})\,,\qquad 0\le s<1/2\,.
$$
Since (see \cite[Theorems 3.29 and 3.30]{McL}),
$$
H^{s}(  \mathcal O)   ^{\ast}
=H_{\overline{\mathcal O}}^{-s}(  \mathbb{R}^{3})  \,,\qquad
s\in\mathbb{R}\,,
$$
$H_{\overline{\mathcal O}}^{-s}(  \mathbb{R}^{3})$ denoting the set of distributions $H^{-s}(  \mathbb{R}^{3})$ with support in ${\overline{\mathcal O}}$,
one has
\be
H^{s}(  \mathbb{R}^{3}\backslash\Gamma)^{\ast}=H^{s}(  \Omega)^{\ast}
\oplus H^{s}(  \mathbb{R}^{3}\backslash\overline\Omega)^{\ast}=
H^{-s}_{\overline\Omega}(\RE^{3})\oplus H^{-s}_{\Omega^{c}}(\RE^{3})\hookrightarrow H^{-s}(\RE^{3})
\,.
\label{dual}%
\ee
Let us notice that
\be\label{B}
\B(H^{s}(  \mathbb{R}^{3}\backslash\Gamma),H^{t}(  \mathbb{R}^{3}\backslash\Gamma)^{\ast})
\hookrightarrow
\B(H^{s}(  \mathbb{R}^{3}),H^{-t}(  \mathbb{R}^{3}))\,,\qquad s,t\ge 0\,,
\ee
and
\be\label{B*}
\B(H^{-s}( \mathbb{R}^{3}),H^{t}(  \mathbb{R}^{3}))\hookrightarrow\B(H^{s}(  \mathbb{R}^{3}\backslash\Gamma)^{*},H^{t}(  \mathbb{R}^{3}\backslash\Gamma))
\,,\qquad s,t\ge 0\,.
\ee
\begin{lemma}\label{v}
\begin{equation}\label{v-sob}
\v\in{\B}(  H^{1+s}(  \RE^{3}\backslash\Gamma) ,H^{1-s}(  \RE^{3}\backslash\Gamma)^{*})\,,\qquad -1\le s\le 1  \,. 
\end{equation}
\end{lemma}
\begin{proof}
Given $u= u_{\inn}\oplus u_{\ex}\in  H^{2}(  \mathbb{R}^{3}\backslash\Gamma)    
$ one has%
\[
\|\v_{\infty}u\| _{L^{2}(\mathbb{R}^{3})}\leq\|\v\|_{L^{\infty}(\mathbb{R}^{3})}\| u\| _{L^{2}(\mathbb{R}^{3})}\leq\|\v\|_{L^{\infty}(\mathbb{R}^{3})}\| u\|_{H^{2}(\mathbb{R}^{3}\backslash\Gamma)  }\,.
\]
and
\begin{align*}
\|\v_{2}u\|_{L^{2}(\mathbb{R}^{3})}=&\|\v_{2}\|_{L^{2}(\Omega)}\|u_{\inn}\|_{L^{\infty}(\Omega)}+\|\v_{2}\|_{L^{2}(\mathbb{R}^{3}\backslash\overline\Omega)}
\|u_{\ex}\|_{L^{\infty}(\mathbb{R}^{3}\backslash\overline\Omega)}\\
\lesssim &
\|\v_{2}\|_{L^{2}(\Omega)}\|u_{\inn}\|_{H^{2}(\Omega)}+
\|\v_{2}\|_{L^{2}(\mathbb{R}^{3}\backslash\overline\Omega)}
\|u_{\ex}\|_{H^{2}(\mathbb{R}^{3}\backslash\overline\Omega)}\\
\lesssim &\|\v_{2}\|_{L^{2}(\mathbb{R}^{3})}
\|u\|_{H^{2}(\mathbb{R}^{3}\backslash\Gamma)}\,.
\end{align*}
Hence $\v\in{\B}(  H^{2}(  \RE^{3}\backslash\Gamma) ,L^{2}(  \RE^{3}))$. Then, for any $u,v\in H^{2}(\RE^{3}\backslash\Gamma)$, one has
\begin{align*}
&\big|\langle \v u,v\rangle_{H^{2}(\RE^{3}\backslash\Gamma)^{*},H^{2}(\RE^{3}\backslash\Gamma)}\big|
=\big|\langle \v u,v\rangle_{L^{2}(\RE^{3})}\big|\\
=&
\big|\langle u,\v v\rangle_{L^{2}(\RE^{3})}\big|
\le \|\v\|_{H^{2}(\RE^{3}\backslash\Gamma),L^{2}(\RE^{3})}\|u\|_{L^{2}(\RE^{3})}
\|v\|_{H^{2}(\RE^{3}\backslash\Gamma)}
\end{align*}
and so $u\mapsto\v u$ extends to a map in $\B(L^{2}(  \RE^{3}),H^{2}(  \RE^{3}\backslash\Gamma)^{*} )$. The proof is then concluded by interpolation.
\end{proof}
In the following, $R_{z}$ denotes the resolvent of the free Laplacian, i.e.,
\be\label{free}
R_{z}:=\left(  -\Delta+z\right)  ^{-1}\in\B(H^{s}(\RE^{3}),H^{s+2}(\RE^{3}))\,,\quad s\in\RE\,.
\ee
Since $\v$ is of Rellich-Kato type, one has (see, e.g., \cite[Section 3, $\S$5, Chap. V]{Kato}): 
\begin{theorem}\label{KR}
$\Delta+\v:H^{2}(  \mathbb{R}^{3})\subset  L^{2}(  \mathbb{R}^{3})\to L^{2}(  \mathbb{R}^{3}) $ is
self-adjoint and semi-bounded from  above. Moreover, for $z\in\CO$ sufficiently far away from $(-\infty, 0]$, $\|\v R_{z}\|_{L^{2}(\mathbb{R}^{3}),L^{2}(\mathbb{R}^{3})}<1$, 
and
\be
R_{z}^{\v}:=(-(\Delta+\v)+z)^{-1}=R_{z}+R_{z}(  1- \v R_{z})  ^{-1}\v R_{z}
\,,\label{Rvz}
\ee
\begin{equation}
(  1-\v R_{z})  ^{-1}=\sum_{k=0}^{+\infty}\left(  \v R_{z}\right)  ^{k}
\in{\B}(  L^{2}(\mathbb{R}^{3}))\,.
\label{L-v-est}%
\end{equation}
\end{theorem}
\begin{remark}\label{tpt} Let us notice that Theorem \ref{KR} could be obtained by Corollary \ref{cor1} by taking $\tau_{1}u:=u$ and $B_{1}=\v$. Hence, \eqref{Rvz} holds for any $z$ in $\varrho(\Delta+\v)\cap\CO\backslash(-\infty,0]$ and $(  1+\v R_{z})  ^{-1}\in {\B}(  L^{2}(\mathbb{R}^{3}))$ there.
\end{remark}
\begin{remark}\label{sa} 
By \eqref{free}, \eqref{Rvz}, \eqref{L-v-est}, \eqref{v-sob} and \eqref{B}, one has $R^{\v}_{\bar z}\in
\B(L^{2}(  \mathbb{R}^{3}),H^{2}(  \mathbb{R}^{3}))$ and hence $(R_{\bar z}^{\v})^{*}\in \B(H^{-2}(  \mathbb{R}^{3}),L^{2}(  \mathbb{R}^{3}))$.  Since $(\Delta+\v)$ is self-adjoint in $L^{2}(\RE^{3})$, $(R_{\bar z}^{\v})^{*}|L^{2}(\RE^{3})=
R^{\v}_{z}$. Therefore, $R^{\v}_{z}:L^{2}(\RE^{3})\subset H^{-2}(\RE^{3})\to L^{2}(\RE^{3})$ extends to an operator in $\B(H^{-2}(  \mathbb{R}^{3}),L^{2}(  \mathbb{R}^{3}))$ which, by abuse of notation, we still denote by $R_{ z}^{\v}$ and which coincides with $(R_{\bar z}^{\v})^{*}$. Then, by interpolation, one gets
\be\label{Rvz-int}
R_{z}^{\v}\in \B(H^{s-1}(  \mathbb{R}^{3}),H^{s+1}(  \mathbb{R}^{3}))\,,\qquad -1\le s\le 1\,.
\ee
\end{remark}
\begin{remark} By \eqref{Rvz}, 
\be\label{btr}
(  1- \v R_{z})  ^{-1}\v=(-\Delta+z)R^{\v}_{z}(-\Delta+z)-(-\Delta+z)\,.
\ee
Hence, by \eqref{Rvz-int},  $(  1- \v R_{z})  ^{-1}\v\in \B(H^{2}(\RE^{3}),L^{2}(\RE^{3}))$ extends to  a map 
\be\label{L-int}
\Lambda^{\!\v}_{z}\in \B(H^{s+1}(  \mathbb{R}^{3}),H^{s-1}(  \mathbb{R}^{3}))\,,
\qquad -1\le s\le 1 
\ee
With such a notation, $R_{z}^{\v}$ in \eqref{Rvz-int} has the representation
\be\label{RF-int}
R_{z}^{\v}=R_{z}+R_{z}\Lambda^{\!\v}_{z}R_{z}\,,\qquad 
\Lambda^{\!\v}_{z}|H^{2}(\RE^{3})=(  1- \v R_{z})  ^{-1}\v\,.
\ee
\end{remark}
\begin{remark}\label{Lt} Since $\|R_{z}\v\|_{L^{2}(\mathbb{R}^{3}),L^{2}(\mathbb{R}^{3})}=\|(R_{z}\v)^{*}\|_{L^{2}(\mathbb{R}^{3}),L^{2}(\mathbb{R}^{3})}=\|\v R_{\bar z}\|_{L^{2}(\mathbb{R}^{3}),L^{2}(\mathbb{R}^{3})}
<1$ whenever $z\in\CO$ is sufficiently far away from $(-\infty,0]$, one has
\be\label{Lt2}
(  1-R_{z}\v)  ^{-1}=\sum_{k=0}^{+\infty}\left(  R_{z}\v\right)  ^{k}
\in{\B}(  L^{2}(\mathbb{R}^{3}))
\ee
and
\be\label{Lt-2}
\v(  1- R_{z}\v)  ^{-1}\in\B(L^{2}(\RE^{3}),H^{-2}(\RE^{3}))\,.
\ee
Then, 
$$
\big((  1- R_{z}\v)  ^{-1}\v\big)^{*}=
\big(\v(  1- R_{z}\v)^{-1}\big)^{*}=\v((  1- R_{z}\v)^{*})^{-1}=\v(  1- R_{\bar z}\v)  ^{-1}
$$
and so 
$$
\B(H^{-2}(  \mathbb{R}^{3}),L^{2}(  \mathbb{R}^{3}))\ni(R^{\v}_{z})^{*}=R_{\bar z}+R_{\bar z}\v(  1- R_{\bar z}\v)^{-1}R_{\bar z}=
R^{\v}_{\bar z}=R_{\bar z}+R_{\bar z}\Lambda_{\bar z}^{\!\v}R_{\bar z}\,.
$$
Therefore 
\be\label{LtL2}
\Lambda_{z}^{\!\v}|L^{2}(\RE^{3})=\v(  1- R_{z}\v)  ^{-1}\,.
\ee
\end{remark}
\begin{lemma}
\be\label{Lvz-int}
\Lambda_{z}^{\!\v}\in{\B}(  H^{1+s}(\RE^{3}\backslash\Gamma),
H^{1-s}(\RE^{3}\backslash\Gamma)^{*})
\,, \qquad -1\le s\le 1\,.
\end{equation}
\end{lemma}
\begin{proof} By Lemma \ref{v} and by \eqref{L-v-est}, one has $\Lambda_{z}^{\!\v}=(  1+\v R_{z})  ^{-1}\v\in{\B}(  H^{2}(\RE^{3}\backslash\Gamma),L^{2}(\RE^{3}))$. By Lemma \ref{v}, \eqref{Lt2} and \eqref{LtL2},
$\Lambda_{z}^{\!\v}\in{\B}( L^{2}(\RE^{3}), H^{2}(\RE^{3}\backslash\Gamma)^{*})$.
The proof is then concluded by interpolation.
\end{proof}
By $H^{1-s}(\RE^{3}\backslash\Gamma)^{*}\hookrightarrow H^{s-1}(\RE^{3})$ and \eqref{free} one has
\begin{corollary}
\be\label{RLvz}
R_{z}\Lambda_{z}^{\!\v}\in{\B}(  H^{s}(\RE^{3}\backslash\Gamma),
H^{s}(\RE^{3}))
\,, \qquad 0\le s\le 2\,.
\end{equation}

\end{corollary}
In later proofs we will need the estimate provided in the following:
\begin{lemma}
There exist $c_{1}>0$, $c_{2}>0$ such that, for any $u\equiv u_{\inn}\oplus u_{\ex}\in H^{1}(\RE^{3}\backslash\Gamma)$ and for any $\varepsilon>0$,
there holds
\begin{equation}
\big|\langle \v u,u\rangle _{H^{1}(\RE^{3}\backslash\Gamma)^{*},H^{1}(\RE^{3}\backslash\Gamma)}\big| \leq
c_{1}\epsilon\left(\|\nabla u_{\inn}\|^{2}_{L^{2}(\Omega_{\inn})}+\|\nabla u_{\ex}\|^{2}_{L^{2}(\Omega_{\ex})}\right)+c_{2}(1+\epsilon^{-3})\|u\|^{2}_{L^{2}(\RE^{3})}\,. \label{V_est}%
\end{equation}
\end{lemma}
\begin{proof} By $H^{1}(\Omega_{\inn/\ex})  \hookrightarrow H^{3/4}(\Omega_{\inn/\ex}) \hookrightarrow L^{4}(\Omega_{\inn/\ex})$, by the Gagliardo-Niremberg 
inequalities (see \cite{BM} for the interior case and \cite{CM} for the exterior one)
$$
\|u_{\inn}\|_{L^{4}(\Omega_{\inn})}\lesssim\|u_{\inn}\|_{H^{3/4}(\Omega_{\inn})}\lesssim\|u_{\inn}\|^{3/4}_{H^{1}(\Omega_{\inn})}\|u_{\inn}\|^{1/4}_{L^{2}(\Omega_{\inn})}\,,
$$ 
$$
\|u_{\ex}\|_{L^{4}(\Omega_{\ex})}\lesssim\|\nabla u_{\ex}\|^{3/4}_{L^{2}(\Omega_{\ex})}\|u_{\ex}\|^{1/4}_{L^{2}(\Omega_{\ex})}
$$ 
and by Young's inequality
$$
xy\le \frac1{\alpha}\left({\epsilon}\ x^{\alpha}+{(\alpha-1)}\,{\epsilon^{-1/(\alpha-1)}}
\ y^{\frac{\alpha-1}{\alpha}}\right)\,,\qquad x,y,\epsilon>0,\, \alpha>1\,,
$$
one gets
$$
\|u\|_{L^{4}(\RE^{3})}^{2}\lesssim
\epsilon\left(\|\nabla u_{\inn}\|^{2}_{L^{2}(\Omega_{\inn})}+\|u\|^{2}_{L^{2}(\Omega_{\inn})}+\|\nabla u_{\ex}\|^{2}_{L^{2}(\Omega_{\ex})}\right)+\frac{1}{3}\,\epsilon^{-3}\|u\|^{2}_{L^{2}(\RE^{3})}\,.
$$
The proof is then concluded by
$$
\big|\langle \v u,u\rangle _{H^{1}(\RE^{3}\backslash\Gamma)^{*},H^{1}(\RE^{3}\backslash\Gamma)}\big| \leq
\|\v_{2}\|_{L^{2}(\RE^{3})}\|u\|^{2}_{L^{4}(\RE^{3})}+
\|\v_{\infty}\|_{L^{\infty}(\RE^{3})}\|u\|^{2}_{L^{2}(\RE^{3})}\,.
$$
\end{proof}
\begin{lemma}\label{vH-1} 
For any $z\in\CO$ sufficiently far away from $(-\infty,0]$, one has 
$\|\v R_{z}\|_{H^{-1}(\RE^{3}),H^{-1}(\RE^{3})}<1$ and 
\begin{equation}
(  1-\v R_{z})  ^{-1}=\sum_{k=0}^{+\infty}\frac{1}{k!}\left(  \v R_{z}\right)  ^{k}
\in{\B}(  H^{-1}(\RE^{3}))\,.
\end{equation}
Furthermore, 
$$
(  1-\v R_{z})  ^{-1}
\in{\B}(  H^{1}(\RE^{3}\backslash\Gamma)^{*})
$$
\end{lemma}
\begin{proof} By \eqref{V_est} and by the polarization identity, for any  $u$ and $v$ in $H^{1}(\RE^{3})$ one has
$$
\big|\langle \v u,v\rangle_{H^{-1}(\RE^{3}),H^{1}(\RE^{3})}\big| \leq
\frac14\Big(c_{1}\epsilon\,\langle -\Delta u,v\rangle _{H^{-1}(\RE^{3}),H^{1}(\RE^{3})}
+c_{2}(1+\epsilon^{-3})\langle u,v\rangle _{H^{-1}(\RE^{3}),H^{1}(\RE^{3})}\Big) 
$$
which gives 
\begin{align*}
&\|\v u\|_{H^{-1}(\RE^{3})}\le\frac14\Big( c_{1}\epsilon\,\|-\Delta u\|_{H^{-1}(\RE^{3})}+
c_{2}(1+\epsilon^{-3})\|u\|_{H^{-1}(\RE^{3})}\Big)\\
\le&\frac14\Big( c_{1}\epsilon\,\|(-\Delta+z) u\|_{H^{-1}(\RE^{3})}+
\big(c_{1}\epsilon\,|z|+c_{2}(1+\epsilon^{-3})\big)\|u\|_{H^{-1}(\RE^{3})}\Big)\,.
\end{align*}
The proof is then concluded by taking $u=R_{z}u_{\circ}$, $u_{\circ}\in H^{-1}(\RE^{3})$, and by 
\begin{align*}
\|R_{z}u_{\circ}\|_{H^{-1}(\RE^{3})}=&\|R_{1}^{1/2}R_{z}u_{\circ}\|_{L^{2}(\RE^{3})}
=
\|R_{z}R_{1}^{1/2}u_{\circ}\|_{L^{2}(\RE^{3})}\\
\le&\|R_{z}\|_{L^{2}(\RE^{3}),L^{2}(\RE^{3})}\|u_{\circ}\|_{H^{-1}(\RE^{3})}\\
\le& d_{z}^{-1}\,\|u_{\circ}\|_{H^{-1}(\RE^{3})}\,,
\end{align*}
where $d_{z}$ is the distance of $z$ from $[0,+\infty)$.\par
Let us now recall the well known resolvent identity in $\B(L^{2}(\RE^{3}))$
\be\label{res-id}
(1-\v R_{z})^{-1}=1-\v R_{z}^{\v}\,.
\ee 
Since the operators in both sides of the above identity are in $\B(H^{-1}(\RE^{3}))$, it extends to $\B(H^{-1}(\RE^{3}))$. By \eqref{Rvz-int}, 
$$
R_{z}^{\v}\in \B(H^{-1}(\RE^{3}),H^{1}(\RE^{3}))\hookrightarrow \B(H^{1}(\RE^{3}\backslash\Gamma)^{*},H^{1}(\RE^{3}\backslash\Gamma))\,;$$
by \eqref{v-sob}, 
$$
\v\in \B(H^{1}(\RE^{3}\backslash\Gamma),H^{1}(\RE^{3}\backslash\Gamma)^{*})\,;
$$ 
then
$$(1-\v R_{z}^{\v})\in \B(H^{1}(\RE^{3}\backslash\Gamma)^{*})\,.
$$
By \eqref{res-id}, this implies that $1-\v R_{z}$ is a bounded bijection from $H^{1}(\RE^{3}\backslash\Gamma)^{*}$ onto itself. Therefore, by the Inverse Mapping Theorem,  $(1-\v R_{z})^{-1}\in \B(H^{1}(\RE^{3}\backslash\Gamma)^{*})$ and \eqref{res-id} holds in  $\B(H^{1}(\RE^{3}\backslash\Gamma)^{*})$.
\end{proof}
\begin{remark} By Lemma \ref{vH-1},  
$$
\Lambda^{\!\v}_{z}|H^{1}(\RE^{3}\backslash\Gamma)=(  1- \v R_{z})  ^{-1}\v\,.
$$
By $(  1- \v R_{z})  ^{-1}\in\B(H^{-1}(\RE^{3}))$ and by $\v\in\B(H^{1}(\RE^{3}), H^{-1}(\RE^{3}))$ one gets $$(  1- \v R_{z})  ^{-1}\v\in\B(H^{1}(\RE^{3}),H^{-1}(\RE^{3}))\,.
$$ 
Thus, by \eqref{L-int} and \eqref{RF-int},
\be\label{LvHs1}
\Lambda^{\!\v}_{z}|H^{s}(\RE^{3})=(  1- \v R_{z})  ^{-1}\v\,,\qquad 1\le s\le 2\,.
\ee
By duality, similarly to Remark \ref{Lt},  $(  1-R_{z}\v)  ^{-1}\in{\B}(  H^{1}(\RE^{3}))$ and 
\eqref{LtL2} improves to
\be\label{LvHs2}
\Lambda^{\!\v}_{z}|H^{s}(\RE^{3})=\v(  1- \v R_{z})  ^{-1}\,,\qquad 0\le s\le 1\,.
\ee
\end{remark}
\subsection{\label{Sec_Layer}Boundary layer operators.}
We introduce the interior/exterior Dirichlet and Neumann trace operators
$$
\gamma_{0}^{\inn/\ex}:H^{s+1/2}(\Omega_{\inn/ex})\to B_{2,2}^{s}(\Gamma)\,,\qquad s>0\,, 
$$
$$
\gamma_{1}^{\inn/\ex}:H^{s+3/2}(\Omega_{\inn/ex})\to B_{2,2}^{s}(\Gamma)\,,\qquad s>0\,,
$$
where $\Omega_{\inn}\equiv\Omega$ and $\Omega_{\ex}:=\Omega_{\ex}$. The Besov-like trace spaces $B_{2,2}^{s}(  \Gamma )  $ 
identify with $H^{s}(\Gamma)  $ when $|s|\le k+1$ and
$\Gamma$ is of class $\mathcal{C}^{k,1}$ (see \cite{JoWa}). Then, we define  the bounded linear operators 
\be\label{g0}
\gamma_{0}:H^{s+1/2}(\RE^{3}\backslash\Gamma)\to B_{2,2}^{s}(\Gamma)\,,\quad\gamma_{0}u:=\frac12\,(\gamma_{0}^{\inn}(u|\Omega_{\inn})+\gamma_{0}^{\ex}(u|\Omega_{\ex}))\,,\qquad s>0\,,
\ee
\be\label{g1}
\gamma_{1}:H^{s+3/2}(\RE^{3}\backslash\Gamma)\to B_{2,2}^{s}(\Gamma)\,,\quad\gamma_{1}u:=\frac12\,(\gamma_{1}^{\inn}(u|\Omega_{\inn})+\gamma_{1}^{\ex}(u|\Omega_{\ex}))\,,\qquad s>0\,.
\ee
The corresponding trace jump bounded operators are defined by
\begin{equation}
[\gamma_{0}]:H^{s+1/2}(  \RE^3\backslash\Gamma )  \rightarrow B_{2,2}^{s}(\Gamma)\,,\quad[
\gamma_{0}]u:=\gamma_{0}^{\-}(u|\Omega_{\inn})-\gamma_{0}^{\+}(u|\Omega_{\ex})\,,
\end{equation}%
\begin{equation}
[\gamma_{1}]:H^{s+3/2}(  \RE^3\backslash\Gamma )\rightarrow B_{2,2}^{s}(\Gamma)\,,\quad[
\gamma_{1}]u:=\gamma_{1}^{\-}(u|\Omega_{\inn})-\gamma_{1}^{\+}(u|\Omega_{\ex})\,.
\end{equation}
By \cite[Lemma 4.3]{McL}, the trace maps $\gamma_{1}^{\-/\+}$ can be extended to the spaces $$H^{1}_{\Delta}(\Omega_{\-/\+}):=\{u_{\-/\+}\in H^{1}(\Omega_{\-/\+}):\Delta_{\Omega_{\-/\+}}u_{\-/\+}\in L^{2}(\Omega_{\-/\+})\}$$ 
as $H^{-1/2}(\Gamma)$-valued bounded operators:
$$
\gamma_{1}^{\-/\+}: H^{1}_{\Delta}(\Omega_{\-/\+})\to H^{-1/2}(\Gamma)\,.
$$
This gives the extensions of the maps $\gamma_{1}$ and $[\gamma_{1}]$ defined on $H^{1}_{\Delta}(\RE^{3}\backslash\Gamma):=H^{1}_{\Delta}(\Omega_{\-})\oplus H^{1}_{\Delta}(\Omega_{\+})$ with values in $H^{-1/2}(\Gamma)$. \par
Then, for any $z\in\CO\backslash(-\infty,0]$,
one defines the single and double-layer operators
\be\label{SL}
\SL_{z}:=(\gamma_{0}R_{\bar z})^{*}=R_{z}\gamma_{0}^{*}\in\B(B_{2,2}^{-s}(\Gamma),H^{3/2-s}(\RE^{3}))\,,
\qquad s>0\,,
\ee
\be\label{DL}
\DL_{z}:=(\gamma_{1}R_{\bar z})^{*}=R_{z}\gamma_{1}^{*}\in\B(B_{2,2}^{-s}(\Gamma),H^{1/2-s}(\RE^{3}))\,,\qquad s>0\,.
\ee
By \eqref{g0}, one has
\be\label{BSL}
S_{z}:=\gamma_{0}\SL_{z}\in \B((H^{s-1/2}(\Gamma),H^{s+1/2}(\Gamma)))\,,\qquad -1/2<s<1/2\,.
\ee
By the mapping properties of the double-layer operator, one gets\footnote{here and below we can avoid the introduction of the cutoff funcion $\chi$ appearing in \cite[Theorems 6.11-6.13]{McL} since we are dealing with the constant coefficients strongly elliptic operator $-\Delta+z$ (compare \cite[Theorem 6.1]{McL} with \cite[equation (6.10)]{McL})} (see \cite[Theorem 6.11]{McL})
\be\label{map}
\DL_{z}\in \B(H^{1/2}(\Gamma),H^{1}(\RE^{3}\backslash\Gamma))\,.
\ee
Hence,  by 
$$
(-(\Delta_{\Omega_{\-}}\oplus\Delta_{\Omega_{\+}})+z)\DL_{z}=0\,,
$$
one gets 
$$
\DL_{z}\in \B(H^{1/2}(\Gamma),H^{1}_{\Delta}(\RE^{3}\backslash\Gamma))\,.
$$
Thus
$$
D_{z}:=\gamma_{1}\DL_{z}\in \B(H^{1/2}(\Gamma),H^{-1/2}(\Gamma))\,.
$$
These mapping properties can be extended to a larger range of Sobolev spaces (see \cite[Theorem 6.12 and successive remarks]{McL}): 
\be\label{SLext}
\SL_{z}\in \B(H^{s-1/2}(\Gamma),H^{s+1}(\RE^{3}))\,,\qquad -1/2\le s\le 1/2\,,
\ee
\be\label{BSLext}
S_{z}\in \B(H^{s-1/2}(\Gamma),H^{s+1/2}(\Gamma))\,,\qquad
-1/2\le s\le 1/2\,, 
\ee
\be\label{DLext}
\DL_{z}\in \B(H^{s+1/2}(\Gamma),H^{s+1}(\RE^{3}\backslash\Gamma))\,,\qquad 
-1/2\le s\le 1/2\,, 
\ee
\be\label{BDLext}
D_{z}\in \B(H^{s+1/2}(\Gamma),H^{s-1/2}(\Gamma))\,,\qquad 
-1/2\le s\le 1/2
\ee
and, whenever $s\ge 0$ in \eqref{SLext}, \eqref{DLext} above, the following jump relations holds (see \cite[Theorem 6.11]{McL})
\be\label{jump0}
[\gamma_{0}]\SL_{z}=0\,,\quad [\gamma_{1}]\SL_{z}=-1\,,
\quad 
\ee
\be\label{jump1}
[\gamma_{0}]\DL_{z}=1\,,\quad [\gamma_{1}]\DL_{z}=0\,.
\ee
Whenever the boundary $\Gamma$ is of class ${\mathcal C}^{1,1}$ one gets an improvement as regards the regularity properties of the single- and double-layer operators (see \cite[Theorem 6.13 and Corollary 6.14]{McL}):
\be\label{SLext+}
\SL_{z}\in \B(H^{s-1/2}(\Gamma),H^{s+1}(\RE^{3}\backslash\Gamma))\,,\qquad 1/2<s\le 1\,,
\ee
\be\label{DLext+}
\DL_{z}\in \B(H^{s+1/2}(\Gamma),H^{s+1}(\RE^{3}\backslash\Gamma))\,,\qquad 
1/2<s\le 1\,, 
\ee
By \eqref{SL}, \eqref{DL} and \eqref{RLvz} one has
\begin{lemma} For any $z\in \varrho(\Delta+\v)\cap (\CO\backslash(-\infty,0])$,
\be\label{SLv} 
\SL^{\v}_{z}:=R^{\v}_{z}\gamma_{0}^{*}=\SL_{z}+R_{z}\Lambda_{z}^{\!\v}\SL_{z}\in\B(B_{2,2}^{-s}(\Gamma),H^{3/2-s}(\RE^{3}))\,,
\qquad 0<s\le 3/2\,,
\ee
\be\label{DLv}
\DL^{\v}_{z}:=R^{\v}_{z}\gamma_{1}^{*}=\DL_{z}+R_{z}\Lambda_{z}^{\!\v}\DL_{z}\in\B(B_{2,2}^{-s}(\Gamma),H^{1/2-s}(\RE^{3}))\,,\qquad 0<s\le 1/2\,.
\ee
\end{lemma}
By \eqref{SLext}, \eqref{DLext} and \eqref{RLvz}, one has 
\begin{lemma}
\be\label{SLvext}
\SL^{\v}_{z}\in \B(H^{s-1/2}(\Gamma),H^{s+1}(\RE^{3}))\,,\qquad -1/2\le s\le 1/2\,,
\ee
\be\label{DLvext}
\DL^{\v}_{z}\in \B(H^{s+1/2}(\Gamma),H^{s+1}(\RE^{3}\backslash\Gamma))\,,
\qquad -1/2\le s\le 1/2\,.
\ee
\end{lemma}
By \eqref{SLext+}, \eqref{DLext+} and \eqref{RLvz}, one has 
\begin{lemma} Let $\Gamma\in{\mathcal C}^{1,1}$. Then
\be\label{SLvext+}
\SL^{\v}_{z}\in \B(H^{s-1/2}(\Gamma),H^{s+1}(\RE^{3}\backslash\Gamma))\,,\qquad 1/2<s\le 1\,,
\ee
\be\label{DLvext+}
\DL^{\v}_{z}\in \B(H^{s+1/2}(\Gamma),H^{s+1}(\RE^{3}\backslash\Gamma))\,,\qquad 
1/2<s\le 1\,, 
\ee
\end{lemma}
By either \eqref{SLv} or \eqref{SLvext} one has
\be\label{BLv-}
\gamma_{0}\SL^{\v}_{z}=S_{z}+\gamma_{0}R_{z}\Lambda_{z}^{\!\v}\SL_{z}\in\B(H^{s-1/2}(\Gamma),H^{s+1/2}(\Gamma))\,,
\qquad -1/2<s<1/2\,.
\ee
Since $\gamma_{0}R_{z}=(R_{\bar z}\gamma_{0}^{*})^{*}=\SL_{\bar z}^{*}$, one gets the following improvement of \eqref{BLv-}:
\begin{lemma} 
\be\label{BLv}
S^{\v}_{z}:=S_{z}+\SL_{\bar z}^{*}\Lambda_{z}^{\!\v}\SL_{z}\in\B(H^{s-1/2}(\Gamma),H^{s+1/2}(\Gamma))\,,
\qquad -1/2\le s\le1/2\,.
\ee
\end{lemma}
\begin{proof} By \eqref{SLext} and duality, $\SL_{\bar z}^{*}\in \B(H^{-1-s}(\RE^{3}),H^{1/2-s}(\Gamma))$. The proof is then concluded by \eqref{BSLext}, \eqref{Lvz-int} and \eqref{SLext} .
\end{proof}
If $\Gamma\in{\mathcal C}^{1,1}$, then, by \eqref{DLvext+},
\be\label{DBLvext-}
\gamma_{1}\DL_{z}^{\v}=D_{z}+\gamma_{1}R_{z}\Lambda_{z}^{\!\v}\DL_{z}\in \B(H^{s+1/2}(\Gamma),H^{s-1/2}(\Gamma))\,,\qquad 1/2<s\le 1\,,
\ee
Since $\gamma_{1}R_{z}=(R_{\bar z}\gamma_{1}^{*})^{*}=\DL_{\bar z}^{*}$, one can improve \eqref{DBLvext-} even without requiring $\Gamma\in{\mathcal C}^{1,1}$:
\begin{lemma}
\be\label{DBLvext}
D^{\v}_{z}:=D_{z}+\DL_{\bar z}^{*}\Lambda_{z}^{\!\v}\DL_{z}\in \B(H^{s+1/2}(\Gamma),H^{s-1/2}(\Gamma))\,,
\qquad -1/2\le s\le 1/2\,.
\ee
\end{lemma}
\begin{proof} 
By \eqref{DLext} and duality, $\DL_{\bar z}^{*}\in \B(H^{s+1}(\RE^{3}\backslash\Gamma)^{*},H^{-s-1/2}(\Gamma))$. The proof is then concluded by \eqref{BDLext}, \eqref{Lvz-int} and \eqref{DLext} .
\end{proof}
In order to prove the jump relations of the double-layer operator relative to $\Delta+\v$ we  need a technical result:
\begin{lemma}\label{tech} If $v\in H^{1}(\RE^{3}\backslash\Gamma)^{*}$, then $[\gamma_{1}]R_{z}v=0$ in $H^{-1/2}(\Gamma)$ for any $z\in\CO\backslash(-\infty,0]$. 
\end{lemma}
\begin{proof} At first let us notice that it suffices to show that the result holds for a single $z\in\CO\backslash(-\infty,0]$. Indeed, by the resolvent identity $R_{w}v=R_{z}v+(z-w)R_{w}R_{z}v$, one gets $R_{w}R_{z}v\in H^{3}(\RE^{3})\subset\ker([\gamma_{1}])$. In particular, we choose $z$ such that $\ker(S_{z})=\{0\}$ (see, e.g., Lemma \eqref{coerc} below). \par 
Given $v\in H^{1}(\RE^{3}\backslash\Gamma)^{*}=
H^{-1}_{\overline\Omega}(\RE^{3})\oplus H^{-1}_{\Omega^{c}}(\RE^{3})\subseteq H^{-1}(\RE^{3})$ and $\chi\in{\mathcal C}_{\text{comp}}^{\infty}(\RE^{3})$ such that $\chi=1$ on a compact set containing an open  neighborhood of $\overline\Omega$, let us set $u:=\chi R_{z}v$. Since $\gamma_{1}^{\inn/\ex}u=\gamma_{1}^{\inn/\ex}R_{z}v$, it suffices to show that $[\gamma_{1}]u=0$. Let us define $u_{\inn/\ex}:=\chi R_{z}v|\Omega_{\inn/\ex}\in H^{1}(\Omega_{\inn/\ex})$, 
$f_{\inn/\ex}:=((-\Delta+z)\chi R_{z}v)|\Omega_{\inn/\ex}\in H^{1}(\Omega_{\inn/\ex})$
 and $g_{\inn/\ex}:=\gamma^{\inn/\ex}_{0}u_{\inn/\ex}\in H^{1/2}(\Gamma)$. Then $u_{\inn/\ex}$ solves the Dirichlet boundary value problems
\begin{equation}
\begin{cases}(  -\Delta_{\Omega_{\inn/\ex}}+z)u_{\inn/\ex}  =f_{\inn/\ex}\,,\\
\gamma^{\inn/\ex}_{0}u_{\inn/\ex}=g_{\inn/\ex}
\end{cases}
\end{equation}
and so, by \cite[Theorems 7.5 and 7.15]{McL} (notice that both $u_{\ex}$ and $f_{\ex}$ have a compact support; in particular the radiation condition ${\mathcal M}u_{\ex}=0$ there required is here satisfied),
$\psi_{\inn/\ex}:=\gamma_{1}^{\inn/\ex}u_{\inn/\ex}\in H^{-1/2}(\Gamma)$ satisfy the equations
\be
S_{z}\psi_{\inn/\ex}=\frac{1}{2}\,(1+D_{z})g_{\inn/\ex}-\gamma_{0}R_{z}v
\ee
Since $u_{\inn}\oplus u_{\ex}=\chi R_{z}v\in H^{1}(\RE^{3})$, one has $g_{\inn}=g_{\ex}$ and so 
$[\gamma_{1}]R_{z}v=\psi_{\inn}-\psi_{ex}=0$ is consequence of $\ker(S_{z})=\{0\}$.
\end{proof}
\begin{lemma} If $s\ge 0$ in \eqref{SLvext} and \eqref{DLvext}, then
\be\label{jumpv0}
[\gamma_{0}]\SL^{\v}_{z}=0\,,\quad [\gamma_{1}]\SL^{\v}_{z}=-1\,,
\quad 
\ee
\be\label{jumpv1}
[\gamma_{0}]\DL^{\v}_{z}=1\,,\quad [\gamma_{1}]\DL^{\v}_{z}=0\,.
\ee
\end{lemma}
\begin{proof} $[\gamma_{0}]\SL^{\v}_{z}=0$ is consequence of $\ran(\SL^{\v}_{z})\subseteq H^{1}(\RE^{3})$. By \eqref{RLvz}, $\ran(R_{z}\Lambda_{z}^{\!\v}\DL_{z})\subseteq H^{1}(\RE^{3})$ and so $[\gamma_{0}]\DL^{\v}_{z}=[\gamma_{0}]\DL_{z}+[\gamma_{0}]R_{z}\Lambda_{z}^{\!\v}\DL_{z}=[\gamma_{0}]\DL_{z}=1$. \par
Since $\Lambda_{z}^{\!\v}\SL_{z}\in\B(H^{s-1/2}(\Gamma),H^{1-s}(\RE^{3}\backslash\Gamma)^{*})$, 
$\Lambda_{z}^{\!\v}\DL_{z}\in\B(H^{s+1/2}(\Gamma),
H^{1-s}(\RE^{3}\backslash\Gamma)^{*})$ and $s\ge 0$, by Lemma \ref{tech} one gets 
$$
[\gamma_{1}]\SL^{\v}_{z}=[\gamma_{1}]\SL_{z}+[\gamma_{1}]R_{z}\Lambda_{z}^{\!\v}\SL_{z}=
[\gamma_{1}]\SL_{z}=-1\,,
$$
$$
[\gamma_{1}]\DL^{\v}_{z}=[\gamma_{1}]\DL_{z}+[\gamma_{1}]R_{z}\Lambda_{z}^{\!\v}\DL_{z}=
[\gamma_{1}]\DL_{z}=0\,.
$$
\end{proof}
When $\v=0$, it is well known that the boundary layer operators have bounded inverses. This
property is next extended to the operators relative to $\Delta+\v$.
\begin{lemma}\label{coerc} There exist $Z^{\circ}_{\v,d}$ and $Z^{\circ}_{\v,n}$, not empty open subsets of $\varrho(\Delta+\v)$, such that 
$$
\forall z\in Z^{\circ}_{\v,d}\,,\quad (S_{z}^{\v})^{-1}\in{\B}(  H^{1/2}(\Gamma)  ,H^{-1/2}(\Gamma) )\,,\qquad\forall z\in Z^{\circ}_{\v,n}\,,\quad
(D_{z}^{\v})  ^{-1}\in{\B}(  H^{-1/2}(\Gamma)  ,H^{1/2}(\Gamma) )\,.
$$
In particular, there exists $\lambda_{\v}>\sup\sigma(\Delta+\v)$ such that $[\lambda_{\v},+\infty)\subset Z^{\circ}_{\v,d}\cap Z^{\circ}_{\v,n}$; moreover, $Z^{\circ}_{\v,d}\cap Z^{\circ}_{0,d}\not=\varnothing $, $Z^{\circ}_{\v,n}\cap Z^{\circ}_{0,n}\not=\varnothing $,  and both  $Z^{\circ}_{\v,d}$ and $Z^{\circ}_{\v,n}$ can be chosen to be symmetric with respect to the real axis.
\end{lemma}
\begin{proof} At first, let us notice that it suffices to show that the bounded inverses exist for any real $\lambda \ge \lambda_{\v}$ for some $\lambda_{\v}>\sup\sigma(\Delta+\v)$. Then, by the continuity of the maps $z\mapsto S^{\v}_{z}$ and $z\mapsto D^{\v}_{z}$, the bounded inverses exist in a complex open neighbourhood of $[\lambda_{\v},+\infty)$.\par
We proceed as in the proof of \cite[Lemma 3.2]{JDE16}. By $(-(\Delta+\v)+\lambda)\SL^{\v}_{\lambda}|\Omega_{\inn/\ex}=0$, by Green's formula and by \eqref{jumpv0}, one gets, for any $\phi\in H^{-1/2}(\Gamma)$,
\begin{align*}
0=&\|\nabla \SL^{\v}_{\lambda}\phi\|^{2}_{L^{2}(\RE^{3})}
-\langle \v \SL^{\v}_{\lambda}\phi,\SL^{\v}_{\lambda}\phi\rangle _{H^{-1}(\RE^{3}),H^{1}(\RE^{3})}
+\lambda\,\|\SL^{\v}_{\lambda}\phi\|^{2}_{L^{2}(\RE^{3})}\\
&+\langle [  \gamma_{1}] \SL^{\v}_{\lambda}\phi ,\gamma_{0}\SL_{\lambda}\phi\rangle _{H^{-1/2}(\Gamma),H^{1/2}(\Gamma)}\\
=&\|\nabla \SL^{\v}_{\lambda}\phi\|^{2}_{L^{2}(\RE^{3})}
-\langle \v \SL^{\v}_{\lambda}\phi,\SL^{\v}_{\lambda}\phi\rangle _{H^{-1}(\RE^{3}),H^{1}(\RE^{3})}
+\lambda\,\|\SL^{\v}_{\lambda}\phi\|^{2}_{L^{2}(\RE^{3})}\\
&-\langle \phi ,S^{\v}_{\lambda}\phi\rangle_{H^{-1/2}(\Gamma),H^{1/2}(\Gamma)}\,.
\end{align*}
Then, by  \eqref{V_est},
\[
\langle \phi ,\gamma_{0}S^{\v}_{\lambda}\phi\rangle _{H^{-1/2}(\Gamma),H^{1/2}(\Gamma)}\geq(  1-c_{1}\varepsilon) \|\nabla \SL^{\v}_{\lambda}\phi\|^{2}_{L^{2}(\RE^{3})}
+(\lambda-c_{2}(1+\varepsilon^{-3}))\|\SL^{\v}_{\lambda}\phi\|^{2}_{L^{2}(\RE^{3})}\,.
\]
Choosing $\varepsilon>0$ such that $c_{1}\varepsilon<1$ and then $\lambda\in\varrho(\Delta+\v)$ such that $\lambda>c_{2}(1+\varepsilon^{-3})$ (this is always possible since $\Delta+\v$ in bounded from above), one gets
\[
\langle \phi,S_{\lambda}^{\v}\phi\rangle_{H^{-1/2}(\Gamma),H^{1/2}(\Gamma)}
\gtrsim\|\SL_{\lambda}^{\v}\phi \|_{H^{1}(  \mathbb{R}^{3})  }^{2}\,.
\]
By $\v\in\B(H^{1}(\RE^{3}\backslash\Gamma),H^{1}(\RE^{3}\backslash\Gamma)^{*})$, Green's formula applies to a couple $u_{\inn/\ex},v_{\inn/\ex}\in H^{1}(\Omega_{\inn/\ex})$ with $\Delta u_{\inn/\ex}\in L^{2}(\Omega_{\inn/\ex})$,
\begin{align}\label{Gf}
&\langle(-(\Delta+\v)+\lambda)u_{\inn/\ex},v_{\inn/\ex}\rangle_{H^{1}(\Omega_{\inn/\ex})^{*},H^{1}(\Omega_{\inn/\ex})}=
\langle\nabla u_{\inn/\ex},\nabla v_{\inn/\ex}\rangle_{L^{2}(\Omega_{\inn/\ex})}\nonumber\\
&-\langle \v u_{\inn/\ex},v_{\inn/\ex}\rangle_{H^{1}(\Omega_{\inn/\ex})^{*},
H^{1}(\Omega_{\inn/\ex})}+\lambda\,\langle u_{\inn/\ex},v_{\inn/\ex}\rangle_{L^{2}(\Omega_{\inn/\ex})}\\
&
\pm \langle \gamma_{1}^{\inn/\ex}u_{\inn/\ex},\gamma^{\inn/\ex}_{0} v_{\inn/\ex}
\rangle_{H^{-1/2}(\Gamma),H^{1/2}(\Gamma)}\,.\nonumber
\end{align}
By 
$$
\big|\langle \v u_{\inn/\ex},v_{\inn/\ex}\rangle_{H^{1}(\Omega_{\inn/\ex})^{*},
H^{1}(\Omega_{\inn/\ex})}\big|\lesssim \|u_{\inn/\ex}\|_{H^{1}(\Omega_{\inn/\ex})}
\|v_{\inn/\ex}\|_{H^{1}(\Omega_{\inn/\ex})}\,,
$$
\eqref{Gf} gives,
\begin{align*}
&\big|\langle \gamma_{1}^{\inn/\ex}u_{\inn/\ex},\gamma^{\inn/\ex}_{0} v_{\inn/\ex}
\rangle_{H^{-1/2}(\Gamma),H^{1/2}(\Gamma)}\big|\\
\lesssim
&\big(\|u_{\inn/\ex}\|_{H^{1}(\Omega_{\inn/\ex})}+\|(-(\Delta+\v)+\lambda)u_{\inn/\ex}\|_{H^{1}(\Omega_{\inn/\ex})^{*}}\big)
\|v_{\inn/\ex}\|_{H^{1}(\Omega_{\inn/\ex})}\,.
\end{align*}
Since $\gamma^{\inn/\ex}_{0}:H^{1}(\Omega_{\inn/\ex})\to H^{1/2}(\Gamma)$ is surjective, finally one gets 
\be\label{cmp}
\|\gamma_{1}^{\inn/\ex}u_{\inn/\ex}\|_{H^{-1/2}(\Gamma)}\lesssim 
\|u_{\inn/\ex}\|_{H^{1}(\Omega_{\inn/\ex})}+\|(-(\Delta+\v)+\lambda)u_{\inn/\ex}\|_{H^{1}(\Omega_{\inn/\ex})^{*}}\,.
\ee
Then, proceeding as in \cite[Lemma 3.2]{JDE16} (compare (3.31) there with \eqref{cmp} here), this
yields
\begin{equation*}
\langle \phi,S_{\lambda}^{\v}\phi\rangle_{H^{-1/2}(\Gamma),H^{1/2}(\Gamma)}
\gtrsim \|\phi\|_{H^{-1/2}(\Gamma)}^{2}\label{SL_coer}%
\end{equation*}
and so $(S_{\lambda}^{\v})  ^{-1}\in{\B}(  H^{1/2}(
\Gamma )  ,H^{-1/2}(\Gamma) )$ by the Lax-Milgram theorem.\par
As regards $D_{\lambda}^{\v}$, the proof is almost the same. By $(-(\Delta+\v)+\lambda)\DL^{\v}_{\lambda}|\Omega_{\inn/\ex}=0$, by Green's formula and by \eqref{jumpv1}, one gets, for any $\phi\in H^{1/2}(\Gamma)$,
\begin{align*}
0=&\|\nabla \DL^{\v}_{\lambda}\phi\|^{2}_{L^{2}(\Omega_{\inn})}+
\|\nabla \DL^{\v}_{\lambda}\phi\|^{2}_{L^{2}(\Omega_{\ex})}
-\langle \v \DL^{\v}_{\lambda}\phi,\DL^{\v}_{\lambda}\phi\rangle _{H^{1}(\RE^{3}\backslash\Gamma)^{*},H^{1}(\RE^{3}\backslash\Gamma)}
+\lambda\,\|\DL^{\v}_{\lambda}\phi\|^{2}_{L^{2}(\RE^{3})}\\
&+\langle D^{\v}_{\lambda}\phi ,\phi\rangle_{H^{-1/2}(\Gamma),H^{1/2}(\Gamma)}\,.
\end{align*}
which leads to 
\[
-\langle D_{\lambda}^{\v}\phi,\phi\rangle_{H^{-1/2}(\Gamma),H^{1/2}(\Gamma)}
\gtrsim\|\DL_{\lambda}^{\v}\phi \|_{H^{1}(  \RE^{3}\backslash\Gamma)  }^{2}\,.
\]
Then, proceeding as in \cite[Lemma 3.2]{JDE16}, by  \eqref{cmp}, this
yields
$$
-\langle D_{\lambda}^{\v}\phi,\phi\rangle_{H^{-1/2}(\Gamma),H^{1/2}(\Gamma)}
\gtrsim \|\phi\|_{H^{1/2}(\Gamma)}^{2}
$$
and so $(D_{\lambda}^{\v})  ^{-1}\in{\B}(  H^{-1/2}(
\Gamma )  ,H^{1/2}(\Gamma) )$ by the Lax-Milgram theorem.
\end{proof}

\section{Laplacians with regular and singular perturbations} 
Here we apply the abstract results in Section \ref{Sec_Krein}, presenting various examples were 
the self-adjoint operator $A$ is the free Laplacian $\Delta:H^{2}(\RE^{3})\subset L^{2}(\RE^{3})\to L^{2}(\RE^{3})$ and $A_{B_{1}}=\Delta+\v$. All over this section we consider a Kato-Rellich potential $\v=\v_{2}+\v_{\infty}$ of short-range type, i.e.,
\be\label{short}
\v_{2}\in L^{2}(\RE^{3}),\quad  \text{supp}(\v_{2})\ \text{bounded}, \qquad
\ |\v_{\infty}(x)|\lesssim\, (1+|x|\,)^{-\kappa(1+\varepsilon)}\,,\quad\kappa\ge 1\,,\quad\varepsilon>0\,.
\ee
We take 
$$
\h_{1}=H^{2}(\RE^{3})\hookrightarrow \b_{1}=H^{1}(\RE^{3}\backslash\Gamma)
\hookrightarrow\h_{1}^{\circ}=L^{2}(\RE^{3})
\,,
$$
and, introducing the multiplication operator  $\langle x\rangle$ by $\langle x\rangle u: x\mapsto(1+|x|^{2})^{1/2}u(x)$, we define
\be\label{tau1}
\tau_{1}:H^{2}(\RE^{3})\to H^{2}(\RE^{3})\,,\quad \tau_{1}u:=\langle x\rangle^{-s}u\,,\qquad
s\ge 0\,,
\ee
and
\be\label{B1}
B_{1}u:=\langle x\rangle^{2s}\v u\,,\qquad 2s<1+\varepsilon\,.
\ee
Further, we take either 
\be\label{td}
\tau_{2}=\gamma_{0}:H^{2}(\RE^{3})\to\h_{2}=B^{3/2}_{2,2}(\Gamma)\hookrightarrow\b_{2}=H^{s_{\circ}}(\Gamma)\,,\quad 0<s_{\circ}\le 1/2\,,
\ee
or
\be\label{tn}
\tau_{2}=\gamma_{1}:H^{2}(\RE^{3})\to\h_{2}=H^{1/2}(\Gamma)\hookrightarrow\b_{2}=H^{-1/2}(\Gamma)\,.
\ee
Hence, by what  is recalled in Subsection \ref{Sec_Layer},  either $G^{2}_{z}=\SL_{z}$ or $G^{2}_{z}=\DL_{z}$
and either
$$
\tau G_{z}(u\oplus\phi)=\langle x\rangle^{-s}R_{z}\langle x\rangle^{-s}u+S_{z}\phi
$$
or
$$
\tau G_{z}(u\oplus\phi)=\langle x\rangle^{-s}R_{z}\langle x\rangle^{-s}u+D_{z}\phi\,.
$$
Thus \eqref{tauG} holds. Notice that $\gamma^{*}_{0}\phi$ and $\gamma_{1}^{*}\phi$, whenever $\phi\in L^{2}(\Gamma)$, identify with the tempered distributions which act on a test function $f$ respectively as
$$
(\phi\delta_{\Gamma})f:=\int_{\Gamma}\phi(x)f(x)\,d\sigma_{\Gamma}(x)\,,\quad
(\phi\delta'_{\Gamma})f:=\int_{\Gamma}\phi(x)\nu(x)\!\cdot\!\nabla f(x)\,d\sigma_{\Gamma}(x)\,,
$$
where $\nu$ is the exterior normal to $\Gamma$. By a slight abuse of notation, in the following we set $\gamma_{0}^{*}\phi\equiv \phi\delta_{\Gamma}$ and $\phi\gamma_{0}^{*}\equiv\delta'_{\Gamma}\phi$ and so, either 
$$\tau^{*}(u\oplus\phi)=\langle x\rangle^{-s}u+\phi\delta_{\Gamma}
$$ 
or 
$$
\tau^{*}(u\oplus\phi)=\langle x\rangle^{-s}u+\phi\delta'_{\Gamma}\,.
$$
In this framework, given a couple of linear operators $B_{0}$ and $B_{2}$ as in \eqref{B012} and such that the triple $\Bi=(B_{0}, B_{1}, B_{2})$ satisfies the hypotheses in Theorem \ref{Th_Krein}, equation \eqref{resolvent} defines a self-adjoint operator $\Delta_{\Bi}$ representing a Laplacian with a Kato-Rellich potential and a distributional one supported on $\Gamma$.  Let us remark that, although $\tau_{1}$ and $B_{1}$ depend on the index $s$, the operator $\Delta_{\Bi}$ is $s$-independent whenever $B_{0}$ and $B_{2}$ are (see the next subsections).  The choice $s\not=0$ is a technical trick which we use to obtain  LAP and a representation formula for the scattering couple $(\Delta_{\Bi}, \Delta)$; whenever one is only interested in providing a resolvent formula for $\Delta_{\Bi}$, then the choice $s=0$ is preferable. In particular, the resolvent formula for $\Delta_{\Bi}$ holds in the setting $s=0$ for any Kato-Rellich potential.

\subsection{The Schr\"odinger operator.}\label{suff}
By our hypotheses on $\v$, one has $\langle x\rangle^{2s}\v\in L^{2}(\RE^{3})+L^{\infty}(\RE^{3})$ and so, by Lemma \ref{v},
$$
B_{1}\in\B(H^{1}(  \mathbb{R}^{3}\backslash\Gamma),H^{1}(  \mathbb{R}^{3}\backslash\Gamma)^{*})\,.
$$
Considering the weight $\varphi(x)=(1+|x|^{2})^{w/2}$, $w\in\RE$, we use the notation $L^{2}_{\varphi}(\RE^{3})\equiv L^{2}_{w}(\RE^{3})$; $H^{k}_{w}(\RE^{3})$, $H^{k}_{w}(\RE^{3}\backslash\Gamma)$ denotes the corresponding scales of weighted Sobolev spaces. \par 
Since 
$$
\langle x\rangle^{w}\in\B(H^{1}_{w'}(\RE^{3}\backslash\Gamma),H^{1}_{w'-w}(\RE^{3}\backslash\Gamma))
$$ 
and, by duality, 
$$
\langle x\rangle^{w}\in\B(H^{1}_{w'}(\RE^{3}\backslash\Gamma)^{*},H^{1}_{w'+w}(\RE^{3}\backslash\Gamma)^{*})\,,
$$ 
one gets 
\be\label{v-w}
\langle x\rangle^{-w-2s}B_{1}\langle x\rangle^{w}=\v\in \B(H^{1}_{w}(\RE^{3}\backslash\Gamma),H^{1}_{-w-2s}(\RE^{3}\backslash\Gamma)^{*})\,.
\ee  
Since 
\be\label{R-w}
R_{z}\in\B(H^{-1}_{w}(\RE^{3}),H^{1}_{w}(\RE^{3}))\hookrightarrow \B(H^{1}_{-w}(\RE^{3}\backslash\Gamma)^{*},H^{1}_{w}(\RE^{3}\backslash\Gamma))\,,
\ee
one has  
$$\tau_{1}G^{1}_{z}=\langle x\rangle^{-s}R_{z}\langle x\rangle^{-s}\in\B(H_{-w}^{1}(\RE^{3}\backslash\Gamma)^{*},H_{w+2s}^{1}(\RE^{3}\backslash\Gamma))\,.
$$ 
In particular, this gives  
$$
\tau_{1}G^{1}_{z}\in\B(H^{1}(\RE^{3}\backslash\Gamma)^{*}),H^{1}(\RE^{3}\backslash\Gamma))\,.$$ 
For $0\le 2s<1+\varepsilon$ we define
$$
M_{z}^{B_{1}}=1-B_{1}\tau_{1}G^{1}_{z}=1-\langle x\rangle^{s}\v R_{z}\langle x\rangle^{-s}=
\langle x\rangle^{s}(1-\v R_{z})\langle x\rangle^{-s}
\in\B(H^{1}(\RE^{3}\backslash\Gamma)^{*})\,.
$$ 
\begin{lemma}\label{5.1} Let $\v$ be as in \eqref{short}, with $\kappa=1$. Then, for $s$ such that $0\le 2s<1+\varepsilon$ and for $z\in\CO$ sufficiently far away from $(-\infty, 0]$,
$$
(1-\v R_{z})^{-1}\in \B(H_{-s}^{1}(\RE^{3}\backslash\Gamma)^{*})\,.
$$
Equivalently,
$$
(M_{z}^{B_{1}})^{-1}\in\B(H^{1}(\RE^{3}\backslash\Gamma)^{*})\,.
$$
\end{lemma}
\begin{proof} Here we use the same kind of arguments as in the second part of the proof of Lemma \ref{vH-1}. Thus we start from  the resolvent identity
\be\label{res-id2}
(1-\v R_{z})^{-1}=1-\v R_{z}^{\v}\,.
\ee 
By Lemma \ref{vH-1}, such an equality holds in $\B(H^{1}(\RE^{3}\backslash\Gamma)^{*})$. By \eqref{Rvz-int}, 
$$
R_{z}^{\v}\in \B(H^{-1}(\RE^{3}),H^{1}(\RE^{3}))\hookrightarrow \B(H^{1}(\RE^{3}\backslash\Gamma)^{*},H^{1}(\RE^{3}\backslash\Gamma))\hookrightarrow \B(H_{-s}^{1}(\RE^{3}\backslash\Gamma)^{*},H_{-s}^{1}(\RE^{3}\backslash\Gamma))\,;$$
by \eqref{v-w}, 
$$
\v\in \B(H_{-s}^{1}(\RE^{3}\backslash\Gamma),H_{-s}^{1}(\RE^{3}\backslash\Gamma)^{*})\,;
$$ 
then
$$
(1-\v R_{z}^{\v})\in \B(H_{-s}^{1}(\RE^{3}\backslash\Gamma)^{*})\,.
$$
Analogously,
$$
(1-\v R_{z})\in \B(H_{-s}^{1}(\RE^{3}\backslash\Gamma)^{*})\,.
$$
By \eqref{res-id2}, this implies that $1-\v R_{z}$ is a bounded bijection from $H^{1}_{-s}(\RE^{3}\backslash\Gamma)^{*}$ onto itself. Therefore, by the Inverse Mapping Theorem,  $(1-\v R_{z})^{-1}\in \B(H_{-s}^{1}(\RE^{3}\backslash\Gamma)^{*})$ and \eqref{res-id} holds in  $\B(H_{-s}^{1}(\RE^{3}\backslash\Gamma)^{*})$.
\end{proof}
Choosing $\Bi=(1,B_{1},0)$, and whenever $Z_{B_{1}}\not=\varnothing$, by Corollary \ref{cor1} the operator $\Delta_{B_{1}}:=\Delta_{(1,B_{1},0)}$ is defined according to the relation  
$$
R_{z}^{B_{1}}:=(-\Delta_{B_{1}}+z)^{-1}=R_{z}+R_{z}\langle x\rangle^{-s}(M_{z}^{B_{1}})^{-1}B_{1}\langle x\rangle^{-s}R_{z}\,,\qquad  z\in Z_{B_{1}}=\varrho(\Delta_{B_{1}})\cap(\CO\backslash(-\infty,0])\,.
$$
By Lemma \ref{5.1}, $Z_{B_{1}}\not=\varnothing$ and by the relation
\be\label{LbLv}
\Lambda_{z}^{B_{1}}=(M_{z}^{B_{1}})^{-1}B_{1}=\langle x\rangle^{s}(1-\v R_{z})^{-1}\langle  x\rangle^{s}\v=\langle x\rangle^{s}\Lambda_{z}^{\!\v}\langle x\rangle^{s}
\,,
\ee
one has 
\be\label{LbLv2}
\Lambda_{z}^{B_{1}}\in{\B}(H^{1}(  \mathbb{R}^{3}\backslash\Gamma),H^{1}(  \mathbb{R}^{3}\backslash\Gamma)^{*})\,.
\ee
Therefore, Theorem \ref{KR} (see also  Remark \ref{tpt}) yields 
$$
R^{\v}_{z}=(-(\Delta+\v)+z)^{-1}=R_{z}+R_{z}\Lambda_{z}^{\v}R_{z}=(-\Delta_{B_{1}}+z)^{-1}
\,,\qquad z\in\varrho(\Delta+\v)\cap\CO\backslash(-\infty,0]\,.
$$
The above relation shows that $\Delta_{B_{1}}$ coincides with the Schr\"odinger operator $\Delta+\v$ provided by the Kato-Rellich theorem. This also shows that $\Delta_{B_{1}}$ is $s$-independent. 
Nevertheless, the operator $\Lambda_{z}^{B_{1}}$ depends on the choice of $s$ and the relations  
\eqref{LbLv} and \eqref{LbLv2} with $s\not=0$ are key objects in our analysis of LAP and scattering theory in the general case. 
\subsection{Asymptotic completeness and scattering matrix.} Before discussing the validity of  our assumptions, we provide the following general results on the scattering couple $(\Delta_{\Bi},\Delta)$.
\begin{theorem} Assume \eqref{short} with $\kappa=1$ and let $\tau_{1}$, $\tau_{2}$ and $B_{1}$ be defined as in \eqref{tau1}-\eqref{tn}. If $\Bi$ is such that (H1)-(H6) hold, then the scattering couple $(\Delta_{\Bi},\Delta)$ is asymptotically complete.
\end{theorem}
\begin{proof} By hypothesis \eqref{short} with $\kappa=1$, it is well known that for $\Delta_{B_{1}}=  \Delta+\v$ one has $\sigma_{ess}(\Delta_{B_{1}})=(-\infty,0]$; moreover, by \cite[Thm. 3.1]{Agmon}, $\sigma_{p}(\Delta_{B_{1}})\cap(-\infty,0)$ is discrete in $(-\infty,0)$. Hence, by \cite[Thm. 4.2]{Agmon}, 
$e(\Delta_{B_{1}})\cap (-\infty,0)$ is countable with $\{0\}$ as the eventual set of accumulations points. Therefore, by Theorem \ref{AC}, $\sigma_{sc}(\Delta_{\Bi})=\varnothing$ and $(\Delta_{\Bi},\Delta)$ is asymptotically complete. 
\end{proof}
In the framework of this section, Theorem \ref{S-matrix} rephrases as
\begin{theorem}\label{Llambda} Assume \eqref{short} with $\kappa=1$ and let $\tau_{1}$, $\tau_{2}$ and $B_{1}$ be defined as in \eqref{tau1}-\eqref{tn}. If $\Bi$ is such that (H1)-(H7) hold, then  the scattering matrix of the couple $(\Delta_{\Bi},\Delta)$ has the representation
\begin{equation}
{\Sc}^{\Bi}_{\lambda}=1-2\pi i\, {\mathsf L}_{\lambda}
{\mathsf\Lambda}^{\!\Bi,+}_{\lambda}{\mathsf L}_{\lambda}^{*}\,,\quad \lambda\in(-\infty,0]\cap(\RE\backslash{e}(\Delta_{\Bi}))\,,
\end{equation}
where  
$$
{\mathsf\Lambda}^{\!\Bi,\pm}_{\lambda}=\lim_{\varepsilon\searrow 0}{\mathsf\Lambda}^{\!\Bi}_{\lambda\pm i\varepsilon}\,,
$$
the limit existing in $\B(H_{s}^{1}(\RE^{3}\backslash\Gamma)\oplus H^{t}(\Gamma),H_{s}^{1}(\RE^{3}\backslash\Gamma)^{*}\oplus H^{-t}(\Gamma))$, 
\begin{align*}
&{\mathsf\Lambda}_{z}^{\Bi}:=\begin{bmatrix}\Lambda^{\!\v}_{z}+\Lambda^{\!\v}_{z}G^{2}_{z}\widehat \Lambda^{\Bi}_{z}(G^{2}_{\bar z})^{*}\Lambda^{\!\v}_{z}&
\Lambda^{\!\v}_{z}G^{2}_{z}\widehat \Lambda^{\Bi}_{z}\\
\widehat \Lambda^{\Bi}_{z}(G^{2}_{\bar z})^{*}\Lambda_{z}^{\v}&
\widehat \Lambda^{\Bi}_{z}\,;
\end{bmatrix}\\
=&\left(1+\begin{bmatrix}\Lambda^{\!\v}_{z}&0\\0&\widehat \Lambda^{\Bi}_{z}
\end{bmatrix}\begin{bmatrix}G^{2}_{z}\widehat \Lambda^{\Bi}_{z}(G^{2}_{\bar z})^{*}&
G^{2}_{z}\\
(G^{2}_{\bar z})^{*}&0
\end{bmatrix}\,\right)\begin{bmatrix}\Lambda^{\!\v}_{z}&0\\0&\widehat \Lambda^{\Bi}_{z}
\end{bmatrix}\,
\end{align*}
and 
$$
{\mathsf L}_{\lambda}:H_{s}^{1}(\RE^{3}\backslash\Gamma)^{*}\oplus H^{-t}(\Gamma)\to(L^{2}(M)_{ac})_{\lambda}\,,\qquad{\mathsf L}_{\lambda}(u\oplus\phi):=\frac{|\lambda|^{\frac{1}4}}{2^{\frac12}}\,\left({\mathsf L}^{1}_{\lambda}u+{\mathsf L}^{2}_{\lambda}\phi\right)\,,
$$
with
$$
\text{$G^{2}_{z}=\SL_{z}$ and $t=s_{\circ}$ if $\tau_{2}=\gamma_{0}$, $\quad G^{2}_{z}=\DL_{z}$ and
$t=\frac12$ if $\tau_{2}=\gamma_{1}$,}
$$
$$
{\mathsf L}^{1}_{\lambda}u(\xi):=\widehat{u}\,(|\lambda|^{1/2}\xi)\,,
\qquad
{\mathsf L}^{2}_{\lambda}\phi(\xi):=\frac1{(2\pi)^{\frac32}}\,\langle \tau_{2}(\chi u^{\xi}_{\lambda}),\phi\rangle_{H^{t}(\Gamma),H^{-t}(\Gamma)}\,.
$$
Here  $\widehat u$ denotes the Fourier transform, ${\mathbb S}^{2}$ denotes the 2-dimensional unitary sphere in $\RE^{3}$, $u^{\xi}_{\lambda}$ is the plane wave with direction $\xi\in{\mathbb S}^{2}$ and wavenumber $|\lambda|^{\frac12}$, i.e., $u^{\xi}_{\lambda}(x)=e^{i\,|\lambda|^{\frac12}\xi\cdot x}$ and $\chi\in C^{\infty}_{comp}(\RE^{3})$ is such that $\chi|\Gamma=1$.
\end{theorem}  
\begin{proof} Taking into account the definition in \eqref{LLL}, let us set 
\begin{align*}
{\mathsf L}_{\lambda}(u\oplus\phi):=&-
{\L}_{\lambda}(\langle x\rangle^{s}u\oplus\phi)=
-(\mu-\lambda)(FG_{\mu}(\langle x\rangle^{s}u\oplus\phi))_{\lambda}\\
=&-(\mu-\lambda)(FR_{\mu}\tau_{1}^{*}\langle x\rangle^{s}u+FR_{\mu}\tau_{2}^{*}\phi))_{\lambda}\,.
\end{align*} 
The unitary 
map $F:L^{2}(\RE^{3})\to \int^{\oplus}_{(-\infty,0)}L^{2}({\mathbb S}^{2})\,d\lambda\equiv L^{2}((-\infty,0);L^{2}({\mathbb S}^{2}))$ diagonalizing $A=\Delta$ is given by 
\be\label{fourier}
(Fu)_\lambda(\xi):=-\frac{|\lambda|^{\frac{1}4}}{2^{\frac12}}\, \widehat u(|\lambda|^{1/2}\xi)\,.
\ee 
Therefore, by  $(\mu-\lambda)\widehat{ R_{\mu}f}(|\lambda|^{1/2}\xi)=-\widehat f(|\lambda|^{1/2}\xi)$, 
one gets 
\begin{align*}
(\mu-\lambda)(FR_{\mu}\tau_{1}^{*}\langle x\rangle^{s}u)_{\lambda}(\xi)=-
\frac{|\lambda|^{\frac{1}4}}{2^{\frac12}}\,\widehat{u}\,(|\lambda|^{1/2}\xi)
\,.
\end{align*}
This gives ${\mathsf L}^{1}_{\lambda}$. As regards ${\mathsf L}^{2}_{\lambda}$, the computation was given in \cite[Theorem 5.1]{JMPA}. \par
The results about ${\mathsf\Lambda}^{\!\Bi}_{z}$ are direct consequences of the definition of ${\mathsf L}_{\lambda}$, Theorem \ref{S-matrix} and relations  \eqref{LB-new}, \eqref{LB-new2}, \eqref{LbLv}.
\end{proof}
\begin{remark} Let us notice that, whenever $u\in L^{2}_{w}(\RE^{3})$, $w>3/2$,
$$
{\mathsf L}^{1}_{\lambda}u(\xi)=\frac1{(2\pi)^{\frac32}}\,\langle u^{\xi}_{\lambda},u\rangle_{L^{2}_{-w}(\RE^{3}),L^{2}_{w}(\RE^{3})}
$$
and so ${\mathsf L}^{1}_{\lambda}$ and ${\mathsf L}^{2}_{\lambda}$ have a similar structure.
\end{remark}
\subsection{Checking the conditions (H1)-(H7).} Next we discuss the validity of (H1)-(H7) in our framework. In particular we show that (H1), (H2), (H4.2)-(H7) hold with the choice $\kappa=1$ in \eqref{short}, without the need to specify the operators $B_{0}$ and $B_{2}$. We prove  (H3) with $\kappa=2$, while  the validity of  
(H4.1), i.e. the semi-boundedness of $A_{\Bi}$, will be checked case by case in the analysis of each model. \par
As in the previous subsections we use the weight $\varphi(x)=(1+|x|^{2})^{w/2}$, $w\in\RE$; the notation for the corresponding weighted spaces are: $L^{2}_{w}(\RE^{3})$, $H^{k}_{w}(\RE^{3})$ and $H^{k}_{w}(\RE^{3}\backslash\Gamma)$. 
From now on, the parameter $s$ in the definitions \eqref{tau1} and \eqref{B1} is restricted to the range  
\be\label{bound1}
1<2s<1+\varepsilon\,.
\ee
Be aware that in the following proofs the index $s$ labeling the weighted spaces fulfills the bounds \eqref{bound1}.  
\begin{lemma} Let $\v$ be short-range as in \eqref{short}, with $\kappa=1$.
Then hypotheses (H1), (H2), (H6), (H7.1), (H7.2), (H7.3) hold true. 
\end{lemma}
\begin{proof} By \cite[Lemma 1, page 170]{ReSi IV}, $R_{z}=(-\Delta+z)^{-1}\in\B(L^{2}_{s}(\RE^{3}))$ for any $z\in \CO\backslash(-\infty,0]$. Therefore, by the resolvent identity $R^{\v}_{z}=R_{z}(1-\v R^{\v}_{z})$, $z\in\varrho(\Delta+\v)$, and by $R^{\v}_{z}\in \B(L_{s}^{2}(\RE^{3}), H^{2}(\RE^{3}))$, hypothesis (H1) is consequence of $\v=\v_{2}+\v_{\infty}\in\B(H^{2}(\RE^{3}),L_{s}^{2}(\RE^{3}))$. Since $\v_{2}$ has a compact support,  $\v_{2}\in \B(H^{2}(\RE^{3}), L_{s}^{2}(\RE^{3}))$ by Lemma \ref{v}.  As regards $\v_{\infty}$, one has
\begin{align*}
\|\v_{\infty} u\|^{2}_{L_{s}^{2}(\RE^{3})}=\int_{\RE^{3}}|\v_{\infty} u|^{2}(1+|x|^{2})^{s}dx\le c\int_{\RE^{3}}(1+|x|)^{-2(1+\varepsilon)}(1+|x^{2}|)^{s}
|u|^{2}dx
\le \,c\,\|u\|^{2}_{L^{2}(\RE^{3})}\,.
\end{align*}    
By \cite[Theorem 4.1]{Agmon}, LAP holds for $A=\Delta$; hence (H7.1) is satisfied. By the short-range hypothesis on $\v$ and by \cite[Theorem 4.2]{Agmon}, LAP holds for $A_{B_{1}}\equiv\Delta+\v$ as well and, by \cite[Theorems 6.1 and 7.1]{Agmon} asymptotic completeness holds for the scattering couple $(\Delta_{B_{1}},A)\equiv(\Delta+\v,\Delta)$. Hence hypotheses (H1), (H2) 
and (H6) are verified.  \par 
By $R_{z}\in\B(L^{2}_{-s}(\RE^{3}),H^{2}_{-s}(\RE^{3}))$, one gets $G_{z}^{1*}=\langle x\rangle^{-s}R_{z}\in\B(L^{2}_{-s}(\RE^{3}),H^{2}(\RE^{3}))$ and so, by duality, $G_{z}^{1}\in\B(H^{-2}(\RE^{3}), L^{2}_{s}(\RE^{3}))$; moreover, by   $R^{\pm}_{\lambda}\in\B(L^{2}_{s}(\RE^{3}),H^{2}_{-s}(\RE^{3}))$ and by a similar duality argument, one gets $G_{\lambda}^{1,\pm}\in\B(H^{-2}(\RE^{3}), L^{2}_{-s}(\RE^{3}))$. Thus hypothesis (H7.2) holds.\par 
By \eqref{LbLv} and \eqref{LbLv2}, (H7.3) is equivalent  to the existence in $\B(H_{-s}^{2}(\RE^{3}),H_{s}^{-2}(\RE^{3}))$ of $\lim_{\epsilon\searrow 0}\Lambda_{\lambda\pm i\epsilon}^{\!\v}=\lim_{\epsilon\searrow 0}(1-\v R_{\lambda\pm i\epsilon})^{-1}\v$. By \eqref{short}, $\v\in \B(H_{-s}^{2}(\RE^{3}),L_{s}^{2}(\RE^{3}))$. Then, $\lim_{\epsilon\searrow 0}(1-\v R_{\lambda\pm i\epsilon})^{-1}$ exists in $\B(L^{2}_{s}(\RE^{3}))$ (see \cite[proof of Theorem XIII.33, page 177]{ReSi IV}) and so (H7.3) holds.
\end{proof}
\begin{lemma}\label{5.4} Let $\v$ be short-range as in \eqref{short}, with $\kappa=2$.
Then hypothesis (H3) holds true. 
\end{lemma}
\begin{proof} The proof is the same as the one for \cite[Lemma 4.5]{JDE18}, once one proves that 
\be\label{once}
\v R^{\v,\pm}_{\lambda}\in\B(L^{2}_{2s}(\RE^{3}))\,.
\ee
Since $R^{\v,\pm}_{\lambda}\in\B(L^{2}_{2s}(\RE^{3}), H^{2}_{-2s}(\RE^{3}))$, 
\eqref{once} is consequence of 
\be\label{vw}
\v=\v_{2}+\v_{\infty}\in \B(H^{2}_{-2s}(\RE^{3}), L^{2}_{2s}(\RE^{3}))\,.
\ee
Lemma \ref{v} entails $\v_{2}\in \B(H^{2}(\RE^{3}), L^{2}(\RE^{3}))$ and so, since $\v_{2}$ has a compact support, one gets that $\v_{2}$ satisfies \eqref{vw}. As regards $\v_{\infty}$, one has, by $1<2s<1+\varepsilon$,
\begin{align*}
\|\v_{\infty} u\|^{2}_{L^{2}_{2s}(\RE^{3})}=&\int_{\RE^{3}}|\v_{\infty} u|^{2}(1+|x|^{2})^{2s}dx\le c\int_{\RE^{3}}(1+|x|)^{-4(1+\varepsilon)}(1+|x|^{2})^{4s}
|u|^{2}(1+|x|^{2})^{-2s}dx\\
\le& \,c\,\|u\|^{2}_{L^{2}_{-2s}(\RE^{3})}\le c\,\|u\|^{2}_{H^{2}_{-2s}(\RE^{3})}
\end{align*}    
and so $\v_{\infty}$ satisfies \eqref{vw} as well.
\end{proof}
\begin{lemma} Let $\v$ be short-range as in \eqref{short}, with $\kappa=1$ and 
let $\tau_{2}$ be either as in \eqref{td} or as in \eqref{tn}. Then hypotheses (H4.2), (H5) and (H7.4) hold true. 
\end{lemma}
\begin{proof}
By the continuity of $z\mapsto R^{\pm}_{z}$ as a $\B(H^{-1}_{s}(\RE^{3}), H_{-s}^{1}(\RE^{3}))$-valued map, one gets the continuity of $z\mapsto G^{1,\pm}_{z}=R^{\pm}_{z}\langle x\rangle^{-s}$ as a $\B(H^{-1}(\RE^{3}), H_{-s}^{1}(\RE^{3}))$-valued map. Hence, given $\chi\in{\mathcal C}_{\text{comp}}^{\infty}(\RE^{3})$ such that $\chi=1$ on a compact set containing an open  neighborhood of $\overline\Omega$, one gets the continuity of $z\mapsto \chi R^{\pm}_{z}\langle x\rangle^{-s}$ as a $\B(H^{1}(\RE^{3}\backslash\Gamma)^{*}, H^{1}(\RE^{3}))$-valued map. Therefore, $z\mapsto \gamma_{0}G^{1,\pm}_{z}=\gamma_{0}R^{\pm}_{z}\langle x\rangle^{-s}=\gamma_{0}\chi R^{\pm}_{z}\langle x\rangle^{-s}$ is continuous as a $\B(H^{1}(\RE^{3}\backslash\Gamma)^{*}, H^{1/2}(\Gamma))$-valued map.  The continuity of 
$z\mapsto \gamma_{1}G^{1,\pm}_{z}=\gamma_{1}R^{\pm}_{z}\langle x\rangle^{-s}=\gamma_{1}\chi R^{\pm}_{z}\langle x\rangle^{-s}$ as a $\B(H^{1}(\RE^{3}\backslash\Gamma)^{*}, H^{-1/2}(\Gamma))$-valued map follows in an analogous way using the same reasoning as in the proof of Lemma \ref{tech}. In conclusion,  hypothesis (H7.4) holds true.\par
Since $\Gamma$ is compact, the embeddings $\h_{2}\hookrightarrow\b_{2}$, where $\h_{2}$ and $\b_{2}$ are as in \eqref{td} and \eqref{tn}, are compact by standard results on Sobolev embeddings. 
Since $\v\in\B(L^{2}_{-(2s+\eta)}(\RE^{3}),L^{2}_{1+\varepsilon-(2s+\eta)}(\RE^{3}))$ and $(1+\v R_{z})^{-1}\in \B(L^{2}(\RE^{3}))$, by taking $\eta=1+\epsilon-2s>0$, one gets $\Lambda^{\!\v}_{z}\in \B(L^{2}_{-(2s+\eta)}(\RE^{3}),L^{2}(\RE^{3}))$. Hence, by the resolvent formula \eqref{RF-int} and by $R_{z}\in \B(L^{2}_{-(2s+\eta)}(\RE^{3}),H^{2}_{-(2s+\eta)}(\RE^{3}))$, one gets $R_{z}^{\v}\in \B(L^{2}_{-(2s+\eta)}(\RE^{3}),H^{2}_{-(2s+\eta)}(\RE^{3}))$. This entails $\gamma_{0}R_{z}^{B_{1}}=\gamma_{0}R_{z}^{\v}=\gamma_{0}\chi R_{z}^{\v}\in \B(L^{2}_{-(2s+\gamma)}(\RE^{3}),B^{2}_{2,2}(\Gamma))$ and $\gamma_{1}R_{z}^{B_{1}}=\gamma_{1}R_{z}^{\v}=\gamma_{1}\chi R_{z}^{\v}\in \B(L^{2}_{-(2s+\eta)}(\RE^{3}),H^{1/2}(\Gamma))$. Then, by duality, one gets $G_{z}^{B_{1}}\in \B(\h_{2}^{*},L^{2}_{2s+\eta}(\RE^{3}))$. This shows that (H4.2) holds. \par
By \cite[Theorem 4.2]{Agmon}, the map $(\RE\backslash{e}(A_{B_{1}}))\cup\CO_{\pm}\ni z\mapsto R^{B_{1},\pm}_{z}=R_{z}^{\v,\pm}\in \B(L^{2}_{s}(\RE^{3}),H^{2}_{-s}(\RE^{3}))$ is continuous. Hence, $z\mapsto\gamma_{0}R^{B_{1},\pm}_{z}=\gamma_{0}R_{z}^{\v,\pm}=\gamma_{0}\chi R_{z}^{\v,\pm}$ and  $z\mapsto\gamma_{1}R^{B_{1},\pm}_{z}=\gamma_{1}R_{z}^{\v,\pm}=\gamma_{1}\chi R_{z}^{\v,\pm}$ are continuous as  $\B(L^{2}_{s}(\RE^{3}),B^{3/2}_{2,2}(\Gamma))$-valued and  $\B(L^{2}_{s}(\RE^{3}),H^{1/2}(\Gamma))$-valued maps respectively. Then, by duality, $z\mapsto G^{B_{1},\pm}$ is continuous on $(\RE\backslash{e}(A_{B_{1}}))\cup\CO_{\pm}$ as a $\B(\h_{2}^{*}, L^{2}_{-s}(\RE^{3}))$-valued map. Since both $\gamma_{0}:H^{2}(\RE^{2})\to B^{3/2}_{2,2}(\Gamma)$ and $\gamma_{1}:H^{2}(\RE^{2})\to H^{1/2}(\Gamma)$ are surjective,  $G^{B_{1},\pm}_{z}\in \B(\h_{2}^{*}, L^{2}_{-s}(\RE^{3}))$ is the adjoint of a surjective map and hence is injective. Thus we proved that (H5) holds.
\end{proof}
\section{Applications.}
\subsection{\label{Sec_delta} Short-range potentials and semi-transparent boundary conditions
of $\delta_{\Gamma }$-type} 
Here we take 
$$
\h_{2}= B^{3/2}_{2,2}(\Gamma)\hookrightarrow\b_{2}=\b_{2,2}= H^{s_{\circ}}(\Gamma)
\hookrightarrow\h_{2}^{\circ}=L^{2}(\Gamma)\,,\quad 0< s_{\circ}<1/2\,,
$$
$$
\tau_{2}=\gamma_{0}:H^{2}(\RE^{3})\to  B^{3/2}_{2,2}(\Gamma)\,,\qquad
B_{0}=1\,,
\qquad B_{2}=\alpha\,,
$$
where 
$$
\alpha\in\B(H^{s_{\circ}}(\Gamma),H^{-s_{\circ}}(\Gamma))\,,\quad \alpha^{*}=\alpha\,.
$$  
Let us notice (see \cite[Remark 2.6]{JDE18}) that in the case $\alpha$ is the multiplication operator associated to a real-valued function 
$\alpha$, then $\alpha\in L^{p}(\Gamma)$, $p>2$, fulfills our hypothesis. \par
For any $z\in\CO\backslash(-\infty,0]$, one has
\be
M_{z}^{\Bi}
=1-\begin{bmatrix}
\langle x\rangle^{2s}\v  & 0\\
0 & \alpha%
\end{bmatrix}
\begin{bmatrix}
\langle x\rangle^{-s}R_{z}\langle x\rangle^{-s}  & \langle x\rangle^{-s}R_{z}\gamma_{0}^{*}\\
\gamma_{0}R_{z} \langle x\rangle^{-s}& \gamma_{0}R_{z}\gamma_{0}^{*}%
\end{bmatrix}
=\begin{bmatrix}\langle x\rangle^{s}&0\\0&1\end{bmatrix}{M}_{z}^{\v,\alpha}\begin{bmatrix}\langle x\rangle^{-s}&0\\0&1\end{bmatrix}\,,
\ee
\be
{M}_{z}^{\v,\alpha}:=%
\begin{bmatrix}
1-\v R_{z} & -\v \SL_{z}\\
-\alpha\SL_{\bar z}^{*} & 1-\alpha S_{z}%
\end{bmatrix}
\,.
\ee
By the mapping properties provided in Sections \ref{Sec_V} and \ref{Sec_Layer}, by \eqref{v-w} and \eqref{R-w} with $w=-s$, one gets
$$
M_{z}^{\v,\alpha}\in{\B}(H_{-s}^{1}(  \mathbb{R}^{3}\backslash\Gamma)^{*}\oplus H^{-s_{\circ}}(\Gamma))\,.
$$
According to 
\cite[Lemma 5.8]{JMPA}, for any $z\in \CO\backslash((-\infty,0]\cup\sigma_{\alpha})$, where $\sigma_{\alpha}\subset (0,+\infty)$ is discrete in $(0,+\infty)$, one has
\begin{equation}
(M_{z}^{B_{0},B_{2}})^{-1}=(M_{z}^{\alpha})^{-1}:=
(  1-\alpha S_{z})  ^{-1}\in{\B}(  H^{-s_{\circ}}(
\Gamma )  )  \,.\label{S-alpha}%
\end{equation}
Thus $$
Z_{B_{0},B_{2}}=Z_{\alpha}:=\{z\in \CO\backslash(-\infty,0]:(M_{w}^{\alpha})^{-1}\in{\B}(  H^{-s_{\circ}}(
\Gamma )  ) ,\ w=z,\bar z\}\supseteq \CO\backslash((-\infty,0]\cup\sigma_{\alpha})
$$ 
and 
$$ 
\Lambda_{z}^{B_{0},B_{2}}=(M_{z}^{B_{0},B_{2}})^{-1}B_{2}=\Lambda_{z}^{\!\alpha}:=(1-\alpha S_{z})^{-1}\alpha\in{\B}( H^{s_{\circ}}(
\Gamma ) ,  H^{-s_{\circ}}(\Gamma )  ) \,.
$$
By \cite[Corollary 2.4]{JDE18}, for any $z\in \varrho(\Delta+\v)\backslash\sigma_{\v,\alpha}$, where $\sigma_{\v,\alpha}\subset \varrho(\Delta+\v)\cap\RE$ is discrete in $\varrho(\Delta+\v)\cap\RE$, 
\begin{equation}
(\widehat M_{z}^{\Bi})^{-1}=(\widehat M_{z}^{\v,\alpha})^{-1}:=
(  1-\alpha S_{z}^{\v})  ^{-1}
\in{\B}(  H^{-s_{\circ}}(\Gamma )  )  
\,.\label{Sv-alpha}%
\end{equation}
Thus 
$$
\widehat Z_{\Bi}=\widehat Z_{\v,\alpha}:=\{z\in \varrho(\Delta+\v):(\widehat M_{w}^{\v,\alpha})^{-1}\in{\B}(  H^{-s_{\circ}}(
\Gamma )  ) ,\ w=z,\bar z\}\supseteq \varrho(\Delta+\v)\backslash\sigma_{\v,\alpha}
$$ 
and
$$
\widehat\Lambda^{\Bi}_{z}=(\widehat M_{z}^{\Bi})^{-1}B_{2}=
\widehat\Lambda_{z}^{\v,\alpha}:=(1-\alpha S^{\v}_{z} )^{-1}\alpha
\in{\B}(H^{s_{\circ}}(\Gamma ),   H^{-s_{\circ}}(\Gamma )  )\,.
$$
Hence,   
$$
\Lambda_{z}^{\Bi}=
\begin{bmatrix}\langle x\rangle^{s}&0\\0&1\end{bmatrix}({M}_{z}^{\v,\alpha})^{-1}\begin{bmatrix}\langle x\rangle^{-s}&0\\0&1\end{bmatrix}\begin{bmatrix}
\langle x\rangle^{2s}\v  & 0\\
0 & \alpha%
\end{bmatrix}
=\begin{bmatrix}\langle x\rangle^{s}&0\\0&1\end{bmatrix}{\mathsf\Lambda}^{\!\Bi}_{z}\begin{bmatrix}\langle x\rangle^{s}&0\\0&1\end{bmatrix}\,,
$$
where, by Theorem \ref{Llambda},
\begin{align*}
&{\mathsf\Lambda}^{\!\Bi}_{z}={\mathsf\Lambda}_{z}^{\!\v,\alpha}:=
\begin{bmatrix}
\Lambda_{z}^{\!\v}+\Lambda_{z}^{\!\v}\SL_{z}\widehat\Lambda_{z}^{\v,\alpha}\SL_{\bar z}^{*}\Lambda_{z}^{\!\v}& \Lambda_{z}^{\!\v}\SL_{z}\widehat\Lambda_{z}^{\v,\alpha}
\\\widehat\Lambda_{z}^{\v,\alpha}\SL_{\bar z}^{*}\Lambda_{z}^{\!\v} & 
\widehat\Lambda_{z}^{\v,\alpha}
\end{bmatrix}\\
=&\begin{bmatrix}\Lambda^{\!\v}_{z}&0\\0&\widehat\Lambda_{z}^{\v,\alpha}\end{bmatrix}
\left(1+\begin{bmatrix}
\SL_{z}\widehat\Lambda_{z}^{\v,\alpha}\SL_{\bar z}^{*}& \SL_{z}
\\ \SL_{\bar z}^{*} & 0
\end{bmatrix}\right)
\begin{bmatrix}\Lambda^{\!\v}_{z}&0\\0&\widehat\Lambda_{z}^{\v,\alpha}\end{bmatrix}
\,.
\end{align*}
One has 
\be\label{Lbbvalpha}
{\mathsf\Lambda}_{z}^{\!\v,\alpha}
\in{\B}(H_{-s}^{1}(  \mathbb{R}^{3}\backslash\Gamma)\oplus H^{s_{\circ}}(\Gamma),H_{-s}^{1}(  \mathbb{R}^{3}\backslash\Gamma)^{*}\oplus H^{-s_{\circ}}(\Gamma))\,.
\ee
By Theorems \ref{Th_Krein} and \ref{Th-alt-res}, there follows 
\begin{align}
R_{z}^{\v,\alpha}=&R_{z}+
\begin{bmatrix}R_{z}\langle x\rangle^{-s}&\SL_{z}\end{bmatrix}
\begin{bmatrix}\langle x\rangle^{s}&0\\0&1\end{bmatrix}{\mathsf\Lambda}_{z}^{\!\v,\alpha}\begin{bmatrix}\langle x\rangle^{s}&0\\0&1\end{bmatrix}
\begin{bmatrix}\langle x\rangle^{2s}\v \langle x\rangle^{-s}R_{z}\\ \alpha\SL_{\bar z}^{*}\end{bmatrix}
\label{Rv-alpha-0}
\\ 
=&R_{z}+
\begin{bmatrix}R_{z}&\SL_{z}\end{bmatrix}\begin{bmatrix}\Lambda^{\!\v}_{z}&0\\0&\widehat\Lambda_{z}^{\v,\alpha}\end{bmatrix}
\left(1+\begin{bmatrix}
\SL_{z}\widehat\Lambda_{z}^{\v,\alpha}\SL_{\bar z}^{*}& \SL_{z}
\\ \SL_{\bar z}^{*} & 0
\end{bmatrix}\right)
\begin{bmatrix}\Lambda^{\!\v}_{z}&0\\0&\widehat\Lambda_{z}^{\v,\alpha}\end{bmatrix}
\begin{bmatrix}R_{z}\\ \SL_{\bar z}^{*}\end{bmatrix}
\label{Rv-alpha-1}\\ 
=&R_{z}^{\v}+\SL_{z}^{\v}{\widehat\Lambda}_{z}^{\v,\alpha}{\SL_{\bar z}^{\v}}^{*}\,.
\label{Rv-alpha-2}
\end{align}
is the resolvent of a self-adjoint operator $\Delta^{\!\v,\delta,\alpha}$; \eqref{Rv-alpha-0} holds for any 
$z\in \varrho(\Delta^{\!\v,\delta,\alpha})\cap\CO\backslash(-\infty,0]$, both
\eqref{Rv-alpha-1} and \eqref{Rv-alpha-2} hold for any 
$z\in \varrho(\Delta^{\!\v,\delta,\alpha})\cap\varrho(\Delta+\v)$. 
\par
By Theorem \ref{Th-add},
$$
\Delta^{\!\v,\delta,\alpha}u=\Delta u+\v u+(\alpha\gamma_{0}u)\delta_{\Gamma}\,.
$$
By \eqref{Rv-alpha-2} and by the mapping properties of $\SL^{\v}_{z}$, one has 
$$
\dom(\Delta^{\!\v,\delta,\alpha})\subseteq H^{3/2-s_{\circ}}(\RE^{3})\,.
$$ 
Moreover, by $R^{\v}_{z}u\in H^{2}(\RE^{3})$, so that $[\gamma_{1}]R^{\v}_{z}u=0$, and by \eqref{jumpv0}, one gets $[\gamma_{1}]R^{\v,\alpha}_{z}u=-{\widehat\Lambda}_{z}^{\v,\alpha}{\SL_{\bar z}^{\v}}^{*}u=-\widehat\rho_{\Bi}(R^{\v,\alpha}_{z}u)$. Hence, by Theorem \ref{alt-abc}, 
$$
u\in\dom(\Delta^{\!\v,\delta,\alpha})\quad\Longrightarrow\quad\alpha\gamma_{0}u+[\gamma_{1}]u=0\,.
$$
Since $\widehat Z_{\v,\alpha}$ contains a positive half-line, $\Delta^{\!\v,\delta,\alpha}$ is bounded from above and hypothesis (H4.1) holds. The scattering couple $(\Delta^{\!\v,\delta,\alpha},\Delta)$ is asymptotically complete and the corresponding scattering matrix is given by 
$$
{\Sc}_{\lambda}^{\v,\alpha}=1-2\pi i\,{\mathsf L}_{\lambda}{\mathsf\Lambda}^{\v,\alpha,+}_{\lambda}{\mathsf L}_{\lambda}^{*}\,,\quad \lambda\in(-\infty,0]\backslash(\sigma^{-}_{p}(\Delta+\v)\cup \sigma^{-}_{p}(\Delta^{\!\v,\delta,\alpha}))\,,
$$
where ${\mathsf L}_{\lambda}$ is given in Theorem \ref{Llambda} and ${\mathsf\Lambda}^{\v,\alpha,+}_{\lambda}:=\lim_{\epsilon\searrow 0}{\mathsf\Lambda}^{\v,\alpha}_{\lambda+i\epsilon}$. This latter limit exists by Lemma \ref{rmH7}; in particular, by \eqref{LBpm2}, 
\begin{align*}
{\mathsf\Lambda}^{\v,\alpha,+}=&\left(1+\begin{bmatrix}
(1-\v R^{+}_{\lambda})^{-1}\v&0\\
0&(1-\alpha S^{\v,+}_{\lambda})^{-1}\alpha\end{bmatrix}\begin{bmatrix}
\SL^{+}_{\lambda}
(1-\alpha S^{\v,+}_{\lambda})^{-1}\alpha(\SL^{-}_{\lambda})^{*}& \SL^{+}_{\lambda}
\\(\SL^{-}_{\lambda})^{*}& 0
\end{bmatrix}\right)\times\\
&\times\begin{bmatrix}
(1-\v R^{+}_{\lambda})^{-1}\v&0\\
0&(1-\alpha S^{\v,+}_{\lambda})^{-1}\alpha\end{bmatrix}\,,
\end{align*}
where
$$
R^{\pm}_{\lambda}:=\lim_{\epsilon\searrow 0}R_{\lambda\pm i\epsilon}\,,\qquad\SL^{\pm}_{\lambda}:=\lim_{\epsilon\searrow 0}\SL_{\lambda\pm i\epsilon}\,,\qquad 
S^{\v,\pm}_{\lambda}:=\lim_{\epsilon\searrow 0}\gamma_{0}\SL^{\v}_{\lambda\pm i\epsilon}\,.
$$
\subsection{\label{Sec_dirichlet} Short-range potentials and Dirichlet boundary conditions.} 
Here we take 
$$
\h_{2}= B^{3/2}_{2,2}(\Gamma)\hookrightarrow\b_{2}=H^{1/2}(\Gamma)
\hookrightarrow\h_{2}^{\circ}=L^{2}(\Gamma)\hookrightarrow \b_{2,2}=\b_{2}^{*}= H^{-1/2}(\Gamma)\,,
$$
$$
\tau_{2}=\gamma_{0}:H^{2}(\RE^{3})\to  B^{3/2}_{2,2}(\Gamma)\,,\qquad
B_{0}=0\,,\qquad B_{2}=1\,.
$$
For any $z\in\CO\backslash(-\infty,0]$, one has
$$
M_{z}^{\Bi}
=\begin{bmatrix}
1 & 0\\
0 & 0%
\end{bmatrix}-
\begin{bmatrix}
\langle x\rangle^{2s}\v  & 0\\
0 & 1%
\end{bmatrix}
\begin{bmatrix}
\langle x\rangle^{-s}R_{z}\langle x\rangle^{-s}  & \langle x\rangle^{-s}R_{z}\gamma_{0}^{*}\\
\gamma_{0}R_{z}\langle x\rangle^{-s} & \gamma_{0}R_{z}\gamma_{0}^{*}%
\end{bmatrix}=\begin{bmatrix}\langle x\rangle^{s}&0\\0&1\end{bmatrix}M_{z}^{\v,d}\begin{bmatrix}\langle x\rangle^{s}&0\\0&1\end{bmatrix}\,,
$$
$$
M_{z}^{\v,d}:=%
\begin{bmatrix}
1-\v R_{z} & -\v \SL_{z}\\
-\SL_{\bar z}^{*} & -S_{z}%
\end{bmatrix}
\,.
$$
By the mapping properties provided in Sections \ref{Sec_V} and \ref{Sec_Layer}, by \eqref{v-w} and \eqref{R-w} with $w=-s$, one gets
$$
M_{z}^{\v,d}\in{\B}(H_{-s}^{1}(  \mathbb{R}^{3}\backslash\Gamma)^{*}\oplus H^{-1/2}(\Gamma),H_{-s}^{1}(  \mathbb{R}^{3}\backslash\Gamma)^{*}\oplus H^{1/2}(\Gamma))\,.
$$
By Lemma \ref{coerc} with $\v=0$, for any $z\in Z^{\circ}_{0,d}\not=\varnothing $, 
$$
(M_{z}^{B_{0},B_{2}})^{-1}=\Lambda_{z}^{B_{0},B_{2}}=
(M_{z}^{d})^{-1}=\Lambda_{z}^{\!d}:=-S_{z}^{-1}\in\B(H^{1/2}(\Gamma),H^{-1/2}(\Gamma))\,.
$$
Thus,
$$
Z_{B_{0},B_{2}}=Z_{d}:=\{z\in \CO\backslash(-\infty,0]: (M_{z}^{d})^{-1}\in\B(H^{1/2}(\Gamma),H^{-1/2}(\Gamma))\}\supseteq Z^{\circ}_{0,d}\,.
$$
By Lemma \ref{coerc} again, for any $z\in Z^{\circ}_{\v,d}\not=\varnothing $, 
$$
(\widehat M_{z}^{B_{0},B_{2}})^{-1}=(\widehat \Lambda_{z}^{B_{0},B_{2}})^{-1}=(\widehat M_{z}^{\v,d})^{-1}=\widehat \Lambda_{z}^{\v,d}:=-(S^{\v}_{z})^{-1}\in\B(H^{1/2}(\Gamma),H^{-1/2}(\Gamma))\,.
$$
Thus,
$$
\widehat Z_{\Bi}=\widehat Z_{\v,d}:=\{z\in \varrho(\Delta+\v): (\widehat M_{z}^{\v,d})^{-1}\in\B(H^{1/2}(\Gamma),H^{-1/2}(\Gamma))\}\supseteq Z^{\circ}_{\v,d}\,.
$$
Hence,
\begin{align*}
\Lambda^{\! \Bi}_{z}=\begin{bmatrix}\langle x\rangle^{s}&0\\0&1\end{bmatrix}(M_{z}^{\v,d})^{-1}
\begin{bmatrix}\langle x\rangle^{-s}&0\\0&1\end{bmatrix}\begin{bmatrix}
\langle x\rangle^{2s}\v  & 0\\
0 & 1%
\end{bmatrix}=\begin{bmatrix}\langle x\rangle^{s}&0\\0&1\end{bmatrix}{\mathsf\Lambda}_{z}^{\!\Bi}\begin{bmatrix}\langle x\rangle^{s}&0\\0&1\end{bmatrix}\,,
\end{align*}
where, by Theorem \ref{Llambda},
\begin{align*}
&{\mathsf\Lambda}_{z}^{\!\Bi}={\mathsf\Lambda}_{z}^{\!\v,d}:=\begin{bmatrix}
\Lambda_{z}^{\!\v}-\Lambda_{z}^{\!\v}\SL_{z}(S^{\v}_{z})^{-1}\SL_{\bar z}^{*}\Lambda_{z}^{\!\v}& -\Lambda_{z}^{\!\v}\SL_{z}(S^{\v}_{z})^{-1}\\
-(S^{\v}_{z})^{-1}\SL_{\bar z}^{*}\Lambda_{z}^{\!\v} & -(S^{\v}_{z})^{-1}
\end{bmatrix}\\
=&\begin{bmatrix}\Lambda^{\!\v}_{z}&0\\0&-(S^{\v}_{z})^{-1}\end{bmatrix}
\left(1+\begin{bmatrix}
-\SL_{z}(S^{\v}_{z})^{-1}\SL_{\bar z}^{*}& \SL_{z}
\\ \SL_{\bar z}^{*} & 0
\end{bmatrix}\right)
\begin{bmatrix}\Lambda^{\!\v}_{z}&0\\0&-(S^{\v}_{z})^{-1}\end{bmatrix}
\end{align*}
One has 
\be\label{Lbbv-dir}
{\mathsf\Lambda}_{z}^{\!\v,d}
\in{\B}(H_{-s}^{1}(  \mathbb{R}^{3}\backslash\Gamma)\oplus H^{1/2}(\Gamma),H_{-s}^{1}(  \mathbb{R}^{3}\backslash\Gamma)^{*}\oplus H^{-1/2}(\Gamma))\,.
\ee
By Theorems \ref{Th_Krein} and \ref{Th-alt-res}, there follows that
\begin{align}
&R_{z}^{\v,d}=R_{z}+
\begin{bmatrix}R_{z}\langle x\rangle^{-s}&\SL_{z}\end{bmatrix}
\begin{bmatrix}\langle x\rangle^{s}&0\\0&1\end{bmatrix}{\mathsf\Lambda}_{z}^{\!\v,d}\begin{bmatrix}\langle x\rangle^{s}&0\\0&1\end{bmatrix}
\begin{bmatrix}\langle x\rangle^{2s}\v\langle x\rangle^{-s} R_{z}\\\SL^{*}_{\bar z}\end{bmatrix}\label{Rv-dir-0}
\\
=&R_{z}+
\begin{bmatrix}R_{z}&\SL_{z}\end{bmatrix}
\begin{bmatrix}{\Lambda}^{\v}_{z}&0\\0&-(S^{\v}_{z})^{-1}\end{bmatrix}
\left(1+\begin{bmatrix}
-\SL_{z}(S^{\v}_{z})^{-1}\SL_{\bar z}^{*}& \SL_{z}
\\ \SL_{\bar z}^{*} & 0
\end{bmatrix}\right)
\begin{bmatrix}\Lambda^{\!\v}_{z}&0\\0&-(S^{\v}_{z})^{-1}\end{bmatrix}
\begin{bmatrix}R_{z}\\\SL^{*}_{\bar z}\end{bmatrix}\label{Rv-dir-1}
\\
=&R_{z}^{\v}-\SL^{\v}(S_{z}^{\v})^{-1}{\SL_{\bar z}^{\v}}^{*}
\label{Rv-dir-2}
\end{align}
is the resolvent of a self-adjoint operator $\Delta^{\!\v,d}$; \eqref{Rv-dir-0} holds for any 
$z\in \varrho(\Delta^{\!\v,d})\cap\CO\backslash(-\infty,0]$, both  
\eqref{Rv-dir-1} and 
\eqref{Rv-dir-2} hold for any 
$z\in \varrho(\Delta^{\!\v,d})\cap\varrho(\Delta+\v)$. By \eqref{Rv-alpha-2} and by the mapping properties of $\SL^{\v}_{z}$, one has 
$$
\dom(\Delta^{\!\v,d})\subseteq H^{1}(\RE^{3})\,.
$$ 
By Theorem \ref{alt-abc} and by $[\gamma_{1}]u=-\widehat\rho_{\Bi}u$ for any $u\in\dom(\Delta^{\!\v,d})$, one gets
$$
\Delta^{\!\v,d}\,u=\Delta u+\v u-([\gamma_{1}]u)\delta_{\Gamma}
$$
and
$$
u\in\dom(\Delta^{\!\v,d})\quad\Longrightarrow\quad\gamma_{0}u=0\,.
$$
Therefore, $\dom(\Delta^{\!\v,d})\subseteq H_{0}^{1}(\Omega_{\inn})\oplus H_{0}^{1}(\Omega_{\ex})$.
Since $\widehat Z_{\v,\alpha}$ contains a positive half-line, $\Delta^{\!\v,d}$ is bounded from above and hypothesis (H4.1) holds. The scattering couple $(\Delta^{\!\v,d},\Delta)$ is asymptotically complete and the corresponding scattering matrix is given by 
$$
{\Sc}_{\lambda}^{\v,d}=1-2\pi i\,{\mathsf L}_{\lambda}{\mathsf\Lambda}^{\v,d,+}_{\lambda}{\mathsf L}_{\lambda}^{*}\,,\quad \lambda\in(-\infty,0]\backslash(\sigma^{-}_{p}(\Delta+\v)\cup \sigma^{-}_{p}(\Delta^{\!\v,d}))\,,
$$
where ${\mathsf L}_{\lambda}$ is given in Theorem \ref{Llambda} and ${\mathsf\Lambda}^{\v,d,+}_{\lambda}:=\lim_{\epsilon\searrow 0}{\mathsf\Lambda}^{\v,d}_{\lambda+i\epsilon}$. This latter limit exists by Lemma \ref{rmH7}; in particular, by \eqref{LBpm2}, 
\begin{align*}
{\mathsf\Lambda}^{\v,d,+}=&\left(1+\begin{bmatrix}
(1-\v R^{+}_{\lambda})^{-1}\v&0\\
0&-(S^{\v,+}_{\lambda})^{-1}\end{bmatrix}\begin{bmatrix}
-\SL^{+}_{\lambda}
(S^{\v,+}_{\lambda})^{-1}(\SL^{-}_{\lambda})^{*}& \SL^{+}_{\lambda}
\\(\SL^{-}_{\lambda})^{*}& 0
\end{bmatrix}\right)\times\\
&\times\begin{bmatrix}
(1-\v R^{+}_{\lambda})^{-1}\v&0\\
0&-(S^{\v,+}_{\lambda})^{-1}\end{bmatrix}\,,
\end{align*}
where
$$
R^{\pm}_{\lambda}:=\lim_{\epsilon\searrow 0}R_{\lambda\pm i\epsilon}\,,\qquad\SL^{\pm}_{\lambda}:=\lim_{\epsilon\searrow 0}\SL_{\lambda\pm i\epsilon}\,,\qquad 
S^{\v,\pm}_{\lambda}:=\lim_{\epsilon\searrow 0}\gamma_{0}\SL^{\v}_{\lambda\pm i\epsilon}\,.
$$
\subsection{\label{Sec_neumann} Short-range potentials and Neumann boundary conditions.} 
Here we take 
$$
\h_{2}=\b_{2}^{*}=\b_{2,2}= H^{1/2}(\Gamma)\hookrightarrow\h_{2}^{\circ}=L^{2}(\Gamma)\hookrightarrow\b_{2}=\h^{*}_{2}=\b_{2,2}^{*}=H^{-1/2}(\Gamma)
\,,
$$
$$
\tau_{2}=\gamma_{1}:H^{2}(\RE^{3})\to H^{1/2}(\Gamma)\,,\qquad
B_{0}=0\,,\qquad B_{2}=1\,.
$$
For any $z\in\CO\backslash(-\infty,0]$, one has
$$
M_{z}^{\Bi}
=\begin{bmatrix}
1 & 0\\
0 & 0%
\end{bmatrix}-
\begin{bmatrix}
\langle x\rangle^{2s}\v  & 0\\
0 & 1%
\end{bmatrix}
\begin{bmatrix}
\langle x\rangle^{-s}R_{z}\langle x\rangle^{-s}  & \langle x\rangle^{-s}R_{z}\gamma_{1}^{*}\\
\gamma_{1}R_{z}\langle x\rangle^{-s} & \gamma_{1}R_{z}\gamma_{0}^{*}%
\end{bmatrix}
=\begin{bmatrix}\langle x\rangle^{s}  & 0\\
0 & 1
\end{bmatrix}M_{z}^{\v,n}
\begin{bmatrix}\langle x\rangle^{s}  & 0\\
0 & 1%
\end{bmatrix}\,,
$$
$$
M_{z}^{\v,n}:=%
\begin{bmatrix}
1-\v R_{z} & -\v \DL_{z}\\
-\DL_{\bar z}^{*} & -D_{z}%
\end{bmatrix}
\,.
$$
By the mapping properties provided in Sections \ref{Sec_V} and \ref{Sec_Layer}, by \eqref{v-w} and \eqref{R-w} with $w=-s$, one gets
$$
M_{z}^{\v,n}\in{\B}(H_{-s}^{1}(  \mathbb{R}^{3}\backslash\Gamma)^{*}\oplus H^{1/2}(\Gamma),
H_{-s}^{1}(  \mathbb{R}^{3}\backslash\Gamma)^{*}\oplus H^{-1/2}(\Gamma))\,.
$$
By Lemma \ref{coerc} with $\v=0$, for any $z\in Z^{\circ}_{0,n}\not=\varnothing $, 
$$
(M_{z}^{B_{0},B_{2}})^{-1}=\Lambda_{z}^{B_{0},B_{2}}=
(M_{z}^{n})^{-1}=\Lambda_{z}^{\!n}:=-D_{z}^{-1}\in\B(H^{-1/2}(\Gamma),H^{1/2}(\Gamma))\,.
$$
Thus,
$$
Z_{B_{0},B_{2}}=Z_{n}:=\{z\in \CO\backslash(-\infty,0]: (M_{z}^{n})^{-1}\in\B(H^{-1/2}(\Gamma),H^{1/2}(\Gamma))\}\supseteq Z^{\circ}_{0,n}\,.
$$
By Lemma \ref{coerc} again, for any $z\in Z^{\circ}_{\v,n}\not=\varnothing $, 
$$
(\widehat M_{z}^{B_{0},B_{2}})^{-1}=(\widehat \Lambda_{z}^{B_{0},B_{2}})^{-1}=(\widehat M_{z}^{\v,n})^{-1}=\widehat \Lambda_{z}^{\v,n}:=-(D^{\v}_{z})^{-1}\in\B(H^{-1/2}(\Gamma),H^{1/2}(\Gamma))\,.
$$
Thus,
$$
\widehat Z_{\Bi}=\widehat Z_{n}:=\{z\in \varrho(\Delta+\v): (\widehat M_{z}^{\v,n})^{-1}\in\B(H^{-1/2}(\Gamma),H^{1/2}(\Gamma))\}\supseteq Z^{\circ}_{\v,n}\,.
$$
Hence, 
\begin{align*}
\Lambda^{\! \Bi}_{z}=\begin{bmatrix}\langle x\rangle^{s}&0\\0&1\end{bmatrix}(M_{z}^{\v,n})^{-1}
\begin{bmatrix}\langle x\rangle^{-s}&0\\0&1\end{bmatrix}\begin{bmatrix}
\langle x\rangle^{2s}\v  & 0\\
0 & 1%
\end{bmatrix}=\begin{bmatrix}\langle x\rangle^{s}&0\\0&1\end{bmatrix}{\mathsf\Lambda}_{z}^{\!\Bi}\begin{bmatrix}\langle x\rangle^{s}&0\\0&1\end{bmatrix}\,,
\end{align*}
where, by Theorem \ref{Llambda},
\begin{align*}
&{\mathsf\Lambda}_{z}^{\!\Bi}={\Lambda}_{z}^{\!\v,n}:=\begin{bmatrix}
\Lambda_{z}^{\!\v}-\Lambda_{z}^{\!\v}\DL_{z}(D^{\v}_{z})^{-1}\DL_{\bar z}^{*}\Lambda_{z}^{\!\v}& -\Lambda_{z}^{\!\v}\DL_{z}(D^{\v}_{z})^{-1}\\
-(D^{\v}_{z})^{-1}\DL_{\bar z}^{*}\Lambda_{z}^{\!\v} & -(D^{\v}_{z})^{-1}
\end{bmatrix}\\
=&\begin{bmatrix}\Lambda^{\!\v}_{z}&0\\0&-(D^{\v}_{z})^{-1}\end{bmatrix}
\left(1+\begin{bmatrix}
-\DL_{z}(D^{\v}_{z})^{-1}\DL_{\bar z}^{*}& \DL_{z}
\\ \DL_{\bar z}^{*} & 0
\end{bmatrix}\right)
\begin{bmatrix}\Lambda^{\!\v}_{z}&0\\0&-(D^{\v}_{z})^{-1}\end{bmatrix}
\end{align*}
One has 
\be\label{Lbbv-neu}
{\mathsf\Lambda}_{z}^{\!\v,n}
\in{\B}(H_{-s}^{1}(  \mathbb{R}^{3}\backslash\Gamma)\oplus H^{-1/2}(\Gamma),H_{-s}^{1}(  \mathbb{R}^{3}\backslash\Gamma)^{*}\oplus H^{1/2}(\Gamma))\,.
\ee
By Theorems \ref{Th_Krein} and \ref{Th-alt-res}, there follows that 
\begin{align}
&R_{z}^{\v,n}
=R_{z}+
\begin{bmatrix}R_{z}\langle x\rangle^{-s}&\DL_{z}\end{bmatrix}
\begin{bmatrix}\langle x\rangle^{s}&0\\0&1\end{bmatrix}{\mathsf\Lambda}_{z}^{\!\v,n}\begin{bmatrix}\langle x\rangle^{s}&0\\0&1\end{bmatrix}\begin{bmatrix}\langle x\rangle^{2s}\v \langle x\rangle^{-s}R_{z}\\ \DL^{*}_{\bar z}\end{bmatrix}\label{Rv-neu-0}\\ 
=&R_{z}+
\begin{bmatrix}R_{z}&\DL_{z}\end{bmatrix}
\begin{bmatrix}\Lambda^{\!\v}_{z}&0\\0&\!\!\!-(D^{\v}_{z})^{-1}\end{bmatrix}
\left(1+\begin{bmatrix}
-\DL_{z}(D^{\v}_{z})^{-1}\DL_{\bar z}^{*}& \DL_{z}
\\ \DL_{\bar z}^{*} & 0
\end{bmatrix}\right)
\begin{bmatrix}\Lambda^{\!\v}_{z}&0\\0&\!\!\!-(D^{\v}_{z})^{-1}\end{bmatrix}
\begin{bmatrix}R_{z}\\ \DL^{*}_{\bar z}\end{bmatrix}\label{Rv-neu-1}\\ 
=&R_{z}^{\v}-\DL_{z}^{\v}(D_{z}^{\v})^{-1}{\DL^{\v}_{\bar z}}^{*}
\label{Rv-neu-2}
\end{align}
is the resolvent of a self-adjoint operator $\Delta^{\!\v,n}$; \eqref{Rv-neu-0} holds for any 
$z\in \varrho(\Delta^{\!\v,n})\cap\CO\backslash(-\infty,0]$, both \eqref{Rv-neu-1} and \eqref{Rv-neu-2} hold for any 
$z\in \varrho(\Delta^{\!\v,n})\cap\varrho(\Delta+\v)$. By \eqref{Rv-alpha-2} and by the mapping properties of $\DL^{\v}_{z}$, one has 
$$
\dom(\Delta^{\!\v,n})\subseteq H^{1}(\RE^{3}\backslash\Gamma)\,.
$$ 
By Theorem \ref{alt-abc} and by $[\gamma_{0}]u=\widehat\rho_{\Bi}u$ for any $u\in\dom(\Delta^{\!\v,n})$, one gets
$$
\Delta^{\!\v,n}\,u=\Delta u+\v u+([\gamma_{0}]u)\delta'_{\Gamma}
$$
and 
$$
u\in\dom(\Delta^{\!\v,n})\quad\Longrightarrow\quad\gamma_{1}u=0\,.
$$
Since $\widehat Z_{\v,n}$ contains a positive half-line, $\Delta^{\!\v,n}$ is bounded from above and hypothesis (H4.1) holds. The scattering couple $(\Delta^{\!\v,n},\Delta)$ is asymptotically complete and the corresponding scattering matrix is given by 
$$
{\Sc}_{\lambda}^{\v,n}=1-2\pi i\,{\mathsf L}_{\lambda}{\mathsf\Lambda}^{\v,n,+}_{\lambda}{\mathsf L}_{\lambda}^{*}\,,\quad \lambda\in(-\infty,0]\backslash(\sigma^{-}_{p}(\Delta+\v)\cup \sigma^{-}_{p}(\Delta^{\!\v,n}))\,,
$$
where ${\mathsf L}_{\lambda}$ is given in Theorem \ref{Llambda} and ${\mathsf\Lambda}^{\v,n,+}_{\lambda}:=\lim_{\epsilon\searrow 0}{\mathsf\Lambda}^{\v,n}_{\lambda+i\epsilon}$. This latter limit exists by Lemma \ref{rmH7}; in particular, by \eqref{LBpm2}, 
\begin{align*}
{\mathsf\Lambda}^{\v,n,+}=&\left(1+\begin{bmatrix}
(1-\v R^{+}_{\lambda})^{-1}\v&0\\
0&-(D^{\v,+}_{\lambda})^{-1}\end{bmatrix}\begin{bmatrix}
-\DL^{+}_{\lambda}
(D^{\v,+}_{\lambda})^{-1}(\DL^{-}_{\lambda})^{*}& \DL^{+}_{\lambda}
\\(\DL^{-}_{\lambda})^{*}& 0
\end{bmatrix}\right)\times\\
&\times\begin{bmatrix}
(1-\v R^{+}_{\lambda})^{-1}\v&0\\
0&-(D^{\v,+}_{\lambda})^{-1}\end{bmatrix}\,,
\end{align*}
where
$$
R^{\pm}_{\lambda}:=\lim_{\epsilon\searrow 0}R_{\lambda\pm i\epsilon}\,,\qquad\DL^{\pm}_{\lambda}:=\lim_{\epsilon\searrow 0}\DL_{\lambda\pm i\epsilon}\,,\qquad 
D^{\v,\pm}_{\lambda}:=\lim_{\epsilon\searrow 0}\gamma_{1}\DL^{\v}_{\lambda\pm i\epsilon}\,.
$$
\subsection{\label{Sec_delta'} Short-range potentials and semi-transparent boundary conditions
of $\delta'_{\Gamma }$-type} 
Here we take 
$$
\h_{2}=\b_{2}^{*}=\b_{2,2}= H^{1/2}(\Gamma)\hookrightarrow\h_{2}^{\circ}=L^{2}(\Gamma)\hookrightarrow \b_{2}=\h^{*}_{2}=\b_{2,2}^{*}=H^{-1/2}(\Gamma)
\,,
$$
$$
\tau_{2}=\gamma_{1}:H^{2}(\RE^{3})\to H^{1/2}(\Gamma)\,,\qquad
B_{0}=\theta\,,\qquad B_{2}=1\,,
$$
where 
$$ 
\theta\in\B(H^{s_{\circ}}(\Gamma),H^{-s_{\circ}}(\Gamma))\,,\quad 
0<s_{\circ}<1/2\,,\quad\theta^{*}=\theta\,.
$$
Let us notice (see \cite[Remark 2.6]{JDE18}) that in the case $\theta$ is the multiplication operator associated to a real-valued function 
$\theta$, then $\theta\in L^{p}(\Gamma)$, $p>2$, fulfills our hypothesis. 
Let us also remark that $\B(H^{s_{\circ}}(\Gamma),H^{-s_{\circ}}(\Gamma))\subseteq
\B(H^{1/2}(\Gamma),H^{-1/2}(\Gamma))=\B(\b_{2}^{*},\b_{2,2}^{*})$.\par
For any $z\in\CO\backslash(-\infty,0]$, one has
$$
M_{z}^{\Bi}
=\begin{bmatrix}
 1  & 0\\
0 & \theta%
\end{bmatrix}
-
\begin{bmatrix}
\langle x\rangle^{2s}\v  & 0\\
0 & 1%
\end{bmatrix}
\begin{bmatrix}
\langle x\rangle^{-s}R_{z}\langle x\rangle^{-s}  & \langle x\rangle^{-s}R_{z}\gamma_{1}^{*}\\
\gamma_{1}R_{z}\langle x\rangle^{-s} & \gamma_{1}R_{z}\gamma_{1}^{*}%
\end{bmatrix}=
\begin{bmatrix}\langle x\rangle^{s} & 0\\
0 & 1\end{bmatrix} 
M_{z}^{\v,\theta}
\begin{bmatrix}\langle x\rangle^{-s} & 0\\
0 & 1\end{bmatrix}\,,
$$
$$
M_{z}^{\v,\theta}:=%
\begin{bmatrix}
1-\v R_{z} & -\v \DL_{z}\\
-\DL_{\bar z}^{*} & \theta-D_{z}
\end{bmatrix}\,.%
$$
By the mapping properties provided in Sections \ref{Sec_V} and \ref{Sec_Layer}, by \eqref{v-w} and \eqref{R-w} with $w=-s$,  one gets
$$
M_{z}^{\v,\theta}\in{\B}(H_{-s}^{1}(  \mathbb{R}^{3}\backslash\Gamma)^{*}\oplus H^{1/2}(\Gamma),H_{-s}^{1}(\mathbb{R}^{3}\backslash\Gamma)^{*}\oplus H^{-1/2}(\Gamma))\,.
$$
\begin{lemma}\label{LLt} Let $Z^{\circ}_{\v,n}\not=\varnothing$ be given as in Lemma \ref{coerc}. Then,
$$
\forall z\in \widehat Z^{\circ}_{\v,n}:=Z^{\circ}_{\v,n}\cap\CO\backslash\RE\,,\qquad (1-\theta (D^{\v}_{z})^{-1})^{-1}\in \B(  H^{-1/2}(  \Gamma)) \,.
$$
\end{lemma}
\begin{proof}
We follow the same the arguments as in the proof of \cite[Lemma 5.4]{JMPA}. Since, by the compact embedding $H^{-s_{\circ}}(\Gamma)\hookrightarrow H^{-1/2}(\Gamma)$, $\theta(  D_{z}^{\v})^{-1}\in{\B}(H^{-1/2}(\Gamma))$ is compact, 
by the Fredholm alternative, $1-\theta(  D_{z}^{\v})  ^{-1}$ has a bounded inverse if and only if it has trivial kernel. Let $\varphi\in H^{-1/2}(\Gamma)$ be such that $D_{z}^{\v}\varphi=\theta\varphi$; using the self-adjointness of $\theta$, we get%
\[
(  D_{z}^{\v}-D_{\bar z}^{\v})  \varphi=0\,.
\]
By the resolvent identity,
\[
\text{Im}(z)\gamma_{1}R_{\bar z}^{\v}R_{z}^{\v}%
\gamma_{1}^{\ast}\varphi=0\,.
\]
This gives 
\begin{equation}
\|R_{z}^{\v}\gamma_{1}^{\ast}\varphi\| _{L^{2}(\mathbb{R}^{3})}=0\,. 
\end{equation}
Since $(R_{z}^{\v}\gamma_{1}^{\ast})^{\ast}=\gamma_{1}R_{\bar
{z}}^{\v}\in{\B}(  L^{2}(  \mathbb{R}^{3}),H^{1/2}(\Gamma))  $ is surjective, then
$R_{z}^{\v}\gamma_{1}^{\ast}\in{\B}(  H^{-1/2}(  \Gamma)  ,L^{2}(  \mathbb{R}^{3}))  $ has closed
range by the closed range theorem and, by \cite[Theorem 5.2, p. 231]{Kato},
\[
\|R_{z}^{\v}\gamma_{1}^{\ast}\varphi\|_{L^{2}(
\mathbb{R}^{3})  }\gtrsim\|\varphi\|_{H^{-1/2}(\Gamma)  }\,.
\]
Thus $\ker(1-\theta (D_{z}^{\v})^{-1})=\{0\}$ and the proof is done.
\end{proof}
According to Lemma \ref{LLt} with $\v=0$, for any $z\in \widehat Z^{\circ}_{0,n}\not=\varnothing $,
$$
(M_{z}^{B_{0},B_{2}})^{-1}=(M_{z}^{\theta})^{-1}=\Lambda_{z}^{\!\theta}:=(\theta-D_{z})^{-1}=
-D_{z}^{-1}(1-\theta D_{z}^{-1})^{-1}\in\B(H^{-1/2}(\Gamma),H^{-1/2}(\Gamma))\,.
$$
Thus
$$
Z_{B_{0},B_{2}}=Z_{\theta}:=\{z\in \CO\backslash(-\infty,0]:(M_{z}^{\theta})^{-1}\in\B(H^{-1/2}(\Gamma),H^{-1/2}(\Gamma))\}\supseteq\widehat Z^{\circ}_{0,n}\,.
$$
According to Lemma \ref{LLt} again, for any $z\in \widehat Z^{\circ}_{\v,n}\not=\varnothing $,
$$
(\widehat M_{z}^{B_{0},B_{2}})^{-1}=(\widehat M_{z}^{\v,\theta})^{-1}=\widehat \Lambda_{z}^{\v,\theta}:=(\theta-D_{z}^{\v})^{-1}=
-(D^{\v}_{z})^{-1}(1-\theta (D^{\v}_{z})^{-1})^{-1}\in\B(H^{-1/2}(\Gamma),H^{-1/2}(\Gamma))\,.
$$
Thus
$$
\widehat Z_{\Bi}=\widehat Z_{\v,\theta}:=\{z\in \varrho(\Delta+\v):(\widehat M_{z}^{\v,\theta})^{-1}\in\B(H^{-1/2}(\Gamma),H^{-1/2}(\Gamma))\}\supseteq\widehat Z^{\circ}_{\v,n}\,.
$$
Hence,  
$$
\Lambda_{z}^{\Bi}=\begin{bmatrix}\langle x\rangle^{s} & 0\\
0 & 1\end{bmatrix} 
(M_{z}^{\v,\theta})^{-1}
\begin{bmatrix}\langle x\rangle^{-s} & 0\\
0 & 1\end{bmatrix}\begin{bmatrix}
\langle x\rangle^{2s}\v  & 0\\
0 & 1
\end{bmatrix}=\begin{bmatrix}\langle x\rangle^{s} & 0\\
0 & 1\end{bmatrix} 
{\mathsf\Lambda}_{z}^{\!\Bi}
\begin{bmatrix}\langle x\rangle^{s} & 0\\
0 & 1\end{bmatrix}\,,
$$
where, by Theorem \ref{Llambda},
\begin{align*}
&{\mathsf\Lambda}_{z}^{\!\Bi}={\mathsf\Lambda}_{z}^{\!\v,\theta}:=\begin{bmatrix}
\Lambda_{z}^{\!\v,\theta}\Lambda^{\!\v}_{z}+\Lambda^{\!\v}_{z} \DL_{z}\widehat\Lambda^{\v,\theta}_{z}\DL_{\bar z}^{*}\Lambda^{\!\v}_{z}
& \Lambda^{\!\v}_{z} \DL_{z}\widehat\Lambda^{\v,\theta}_{z}\\
\widehat\Lambda^{\v,\theta}_{z}\DL_{\bar z}^{*}\Lambda^{\!\v}_{z}& \widehat\Lambda^{\v,\theta}_{z}\end{bmatrix}\\
=&
\begin{bmatrix}\Lambda^{\!\v}_{z}&0\\0&\widehat\Lambda_{z}^{\v,\theta}\end{bmatrix}
\left(1+\begin{bmatrix}
\DL_{z}\widehat\Lambda_{z}^{\v,\theta}\DL_{\bar z}^{*}& \DL_{z}
\\ \DL_{\bar z}^{*} & 0
\end{bmatrix}\right)
\begin{bmatrix}\Lambda^{\!\v}_{z}&0\\0&\widehat\Lambda_{z}^{\v,\theta}\end{bmatrix}
\,.
\end{align*}
One has 
\be\label{Ltbbteta}
{\mathsf\Lambda}_{z}^{\!\v,\theta}
\in{\B}(H_{-s}^{1}(  \mathbb{R}^{3}\backslash\Gamma)\oplus H^{-1/2}(\Gamma),H_{-s}^{1}(  \mathbb{R}^{3}\backslash\Gamma)^{*}\oplus H^{1/2}(\Gamma))\,.
\ee
By Theorems \ref{Th_Krein} and \ref{Th-alt-res}, there follows that
\begin{align}
R_{z}^{\v,\theta}
=&R_{z}+\begin{bmatrix}R_{z}\langle x\rangle^{-s}&\DL_{z}\end{bmatrix}
\begin{bmatrix}\langle x\rangle^{s}&0\\0&1\end{bmatrix}{\mathsf\Lambda}_{z}^{\!\v,\theta}\begin{bmatrix}\langle x\rangle^{s}&0\\0&1\end{bmatrix}
\begin{bmatrix}\langle x\rangle^{2s}\v \langle x\rangle^{-s}R_{z}\\\DL^{*}_{\bar z}\end{bmatrix}
\label{Rv-teta-0}\\
=&R_{z}+
\begin{bmatrix}R_{z}&\DL_{z}\end{bmatrix}
\begin{bmatrix}\Lambda^{\!\v}_{z}&0\\0&\widehat\Lambda_{z}^{\v,\theta}\end{bmatrix}
\left(1+\begin{bmatrix}
\DL_{z}\widehat\Lambda_{z}^{\v,\theta}\DL_{\bar z}^{*}& \DL_{z}
\\ \DL_{\bar z}^{*} & 0
\end{bmatrix}\right)
\begin{bmatrix}\Lambda^{\!\v}_{z}&0\\0&\widehat\Lambda_{z}^{\v,\theta}\end{bmatrix}
\begin{bmatrix}R_{z}\\\DL^{*}_{\bar z}\end{bmatrix}
\label{Rv-teta-1}\\ 
=&R_{z}^{\v}+\DL_{z}^{\v}{\widehat\Lambda}_{z}^{\v,\theta}{\DL_{\bar z}^{\v}}^{*}\,.
\label{Rv-teta-2}
\end{align}
is the resolvent of a self-adjoint operator $\Delta^{\!\v,\delta'\!,\theta}$; \eqref{Rv-teta-0} holds for any 
$z\in \varrho(\Delta^{\!\v,\delta'\!,\theta})\cap\CO\backslash(-\infty,0]$, both \eqref{Rv-teta-1} and \eqref{Rv-teta-2} hold for any 
$z\in \varrho(\Delta^{\!\v,\delta'\!,\theta})\cap\varrho(\Delta+\v)$. By \eqref{Rv-alpha-2} and by the mapping properties of $\DL^{\v}_{z}$, one has 
$$
\dom(\Delta^{\!\v,\delta'\!,\theta})\subseteq H^{1}(\RE^{3}\backslash\Gamma)\,.
$$ 
By $R^{\v}_{z}u\in H^{2}(\RE^{3})$, so that $[\gamma_{1}]R^{\v}_{z}u=0$, and by \eqref{jumpv1}, one gets $[\gamma_{0}]R^{\v,\theta}_{z}u={\widehat\Lambda}_{z}^{\v,\theta}{\DL_{\bar z}^{\v}}^{*}u=\widehat\rho_{\Bi}(R^{\v,\theta}_{z}u)$. Hence, by Theorem \ref{alt-abc}, 
$$
\Delta^{\!\v,\delta'\!,\theta}u=\Delta u+\v u+([\gamma_{0}]u)\delta'_{\Gamma}
$$
and
$$
u\in\dom(\Delta^{\!\v,\delta'\!,\theta})\quad\Longrightarrow\quad\gamma_{1}u=\theta[\gamma_{0}]u\,.
$$
Proceeding as in \cite[Subsection 5.5]{JMPA} (see the proof of Theorem 5.15 there), $\Delta^{\!\v,\delta'\!,\theta}$ is bounded from above and so hypothesis (H4.1) holds. The scattering couple $(\Delta^{\!\v,\delta'\!,\theta},\Delta)$ is asymptotically complete and the corresponding scattering matrix is given by 
$$
{\Sc}_{\lambda}^{\v,\theta}=1-2\pi i\,{\mathsf L}_{\lambda}{\mathsf\Lambda}^{\v,\theta,+}_{\lambda}{\mathsf L}_{\lambda}^{*}\,,\quad \lambda\in(-\infty,0]\backslash(\sigma^{-}_{p}(\Delta+\v)\cup \sigma^{-}_{p}(\Delta^{\!\v,\delta'\!,\theta}))\,,
$$
where ${\mathsf L}_{\lambda}$ is given in Theorem \ref{Llambda} and ${\mathsf\Lambda}^{\v,\theta,+}_{\lambda}:=\lim_{\epsilon\searrow 0}{\mathsf\Lambda}^{\v,\theta}_{\lambda+i\epsilon}$. This latter limit exists by Lemma \ref{rmH7}; in particular, by \eqref{LBpm2}, 
\begin{align*}
{\mathsf\Lambda}^{\v,\theta,+}=&\left(1+\begin{bmatrix}
(1-\v R^{+}_{\lambda})^{-1}\v&0\\
0&(\theta-D^{\v,+}_{\lambda})^{-1}\end{bmatrix}\begin{bmatrix}
\DL^{+}_{\lambda}
(\theta-D^{\v,+}_{\lambda})^{-1}\alpha(\DL^{-}_{\lambda})^{*}& \DL^{+}_{\lambda}
\\(\DL^{-}_{\lambda})^{*}& 0
\end{bmatrix}\right)\times\\
&\times\begin{bmatrix}
(1-\v R^{+}_{\lambda})^{-1}\v&0\\
0&(\theta-D^{\v,+}_{\lambda})^{-1}\alpha\end{bmatrix}\,,
\end{align*}
where
$$
R^{\pm}_{\lambda}:=\lim_{\epsilon\searrow 0}R_{\lambda\pm i\epsilon}\,,\qquad\DL^{\pm}_{\lambda}:=\lim_{\epsilon\searrow 0}\SL_{\lambda\pm i\epsilon}\,,\qquad 
D^{\v,\pm}_{\lambda}:=\lim_{\epsilon\searrow 0}\gamma_{0}\DL^{\v}_{\lambda\pm i\epsilon}\,.
$$

\end{document}